\documentclass[review]{elsarticle}
\usepackage[raggedright]{titlesec} 
\usepackage{amsmath}
\usepackage{amssymb}
\usepackage{subfigure}
\usepackage[colorlinks=true]{hyperref}
\usepackage{graphicx,epstopdf}
\usepackage{setspace}
\usepackage{accents}
\usepackage{in dentfirst}
\usepackage{float}
\usepackage{amsthm}
\newtheorem{theorem}{Theorem}[section]
\newtheorem{conjecture}{Conjecture}[section]
\newtheorem{lemma}[theorem]{Lemma}

\newtheorem{remark}{Remark}
\newtheorem{proposition}[theorem]{Proposition}
\newtheorem{corollary}[theorem]{Corollary}
\usepackage{amsmath}
\usepackage{graphics} 
\usepackage{graphicx}  
\usepackage{epsfig} 

\usepackage{float}
\usepackage{epstopdf}
\usepackage{caption}
\usepackage{amsmath}
\usepackage{wrapfig}
\usepackage{graphicx}
\usepackage{amssymb}

\journal{XX}

\numberwithin{equation}{section}

\numberwithin{equation}{section}
\begin{document}
	
	\begin{frontmatter}
		
		\title{Dynamical Analysis of a Lotka-Volterra Competition Model with both Allee and Fear Effect}
		
		\author{Shangming Chen$^{1}$}
		\ead{210320019@fzu.edu.cn}
		
		\author{Fengde Chen$^{1}$}
		\ead{fdchen@fzu.edu.cn }
		
		\author{Vaibhava Srivastava$^{2}$}
		\ead{vaibhava@iastate.edu}

		\author{Rana D. Parshad$^{2}$\corref{my corresponding author}}
		\cortext[my corresponding author]{Corresponding author}
		\ead{rparshad@iastate.edu}

		\address{1) School of Mathematics and Statistics,
			Fuzhou University,\\Fuzhou, Fujian, 350108, P. R. China \\ 2) Department of Mathematics, Iowa State University,\\ Ames, IA 50011, USA.}
		
		\begin{abstract}
			
			Population ecology theory is replete with density dependent processes. However trait-mediated or behavioral indirect interactions can both reinforce or oppose density-dependent effects.
			This paper presents the first two species competitive ODE and PDE systems where an Allee effect, which is a density dependent process and the fear effect, which is non-consumptive and behavioral are \emph{both} present. The stability of the equilibria is discussed analytically using the qualitative theory of ordinary differential equations. It is found that the Allee effect and the fear effect change the extinction dynamics of the system and the number of positive equilibrium points, but they do not affect the stability of the positive equilibria. We also observe some special dynamics that induce bifurcations in the system by varying the Allee or fear parameter. Interestingly we find that the Allee effect working in conjunction with the fear effect, can bring about several qualitative changes to the dynamical behavior of the system with only the fear effect in place, in regimes of small fear. That is, for small amounts of the fear parameter, it can change a competitive exclusion type situation to a strong competition type situation. It can also change a weak competition type situation to a bi-stability type situation. However for large fear regimes the Allee effect reinforces the dynamics driven by the fear effect.
			The analysis of the corresponding spatially explicit model is also presented. To this end the comparison principle for parabolic PDE is used. The conclusions of this paper have strong implications for conservation biology, biological control as well as the preservation of biodiversity.
		\end{abstract}
		
		\begin{keyword}
			Competition Model \sep Allee Effect \sep Fear Effect \sep Stability \sep Bifurcation \sep Reaction-Diffusion System \\
		\end{keyword}
		
	\end{frontmatter}

	\section{Introduction}
	
	The mechanisms of competition, non-consumptive effects or trait mediated indirect interactions and density dependent effects such as Allee effects are central tenets of population ecology theory. They shape the dynamics of many eco-systems \cite{Chesson08, peckarsky2008revisiting, stephens1999consequences, Chess00, courchamp2008allee}. 
	It is typically assumed that species interactions are governed by their respective densities. However, a species could react to the presence of a second species by altering it's phenotype or behavior, consequently effecting population density or fitness of the other species. Thus such trait-mediated or behavioral indirect interactions can both reinforce or oppose density-dependent effects \cite{werner2003review, schmitz2004trophic}, which has strong consequences for community ecology, and food chain dynamics.
	We present and analyze in this work a first model that \emph{connects} trait-mediated indirect interactions focusing on fear effects with density dependent effects herein Allee effects, in a two species competition setting. Our analysis reveals that (1) the Allee effect opposes the fear effect, in the regimes of small fear, (2) the Allee effect reinforces the fear effect, in the regimes of large fear. The details of how these two mechanisms work in conjunction to create such novel dynamics is investigated and presented in the current manuscript.
	
	We begin by fixing ideas about the Allee effect, which is crucial to the ensuing analysis. The ecologist Allee \cite{1} discovered in 1931 that a population's growth rate was correlated with its density - that is being rare may introduce a fitness cost, leading to population decline. He also illustrates that ``clustering" benefits population growth and survival, whereas extreme sparseness and overcrowding can prevent growth and negatively affect reproduction, thus each species has its optimum density. In 1953, Odum \cite{2}, another ecologist, introduced the term ``Allee effect" for this phenomenon. The Allee effect can be divided into three types: species size, structure, and behavioral effects. The species size says that when a population size falls below a certain threshold, the reproduction and maintenance of the population will be affected, leading to a significant increase in the likelihood of population extinction. The single species model with an Allee effect can be modeled by the following ordinary differential equation, 
	
	\begin{equation}
		\displaystyle\frac{\mathrm{d} x}{\mathrm{d} t} =rx\left ( 1-\frac{x}{K} \right ) \left ( x-m \right ) ,
		\label{1.1}
	\end{equation}
	is called the single species model with multiplicative Allee effect. $r$ is the endogenous growth rate of the species, $K$ is the environmental accommodation, and $(x - m)$ is the Allee effect. If $0 < m < K$, we claim that the species is subject to a strong Allee effect. When the species size is below the threshold $m$, the endogenous growth rate of the population is negative, and there is a risk of extinction. If $-K < m \le 0$, the species suffers from a weak Allee effect. At this time, the species' growth slows down, but there is no risk of extinction.
	
	Research on the Allee effect has important implications regarding the conservation of species diversity. According to \eqref{1.1}, Zhu et al. \cite{3} proposed a single-species logistic model with strong Allee effects and feedback control:
	\begin{equation}
		\left\{\begin{array}{l}
			\displaystyle\frac{d x}{d t}=r x\left(1-\displaystyle\frac{x}{k}\right)(x-m)-a x u \vspace{2ex}\\
			\displaystyle\frac{d u}{d t}=-b u+c x
		\end{array}\right.
		\label{1.2}
	\end{equation}
	The authors' results suggest species will become extinct if the feedback control variables and the Allee effect are large enough. In addition, the authors also studied the saddle-node bifurcation, supercritical, and subcritical Hopf bifurcation that occurs with parameter variation. By calculating the universal unfolding near a cusp, it is concluded that system \eqref{1.2} has a Bogdanov-Takens bifurcation of codimension 2. For more related studies on the multiplicative Allee effect, see \cite{4,5}. The Allee effect could also be \emph{weak}, where the growth rate is always positive, but less pronounced at lower densities \cite{Cet99,SS99}. It has been observed that an Allee effect is significant at very low population size and with bias in sex ratio \cite{PK17,W12}. Researchers \cite{Net18,PK17} observed that extinction of Atlantic cod ({\it Gadus morhua}) in the southern Gulf of St. Lawrence and the depletion of Atlantic herring ({\it Clupea harengus}) population in the North Sea are due to predation-driven Allee effect, therein prompting the risk of population extinction.
	
	The study of the predator-prey model is also a central topic in ecology and evolutionary biology \cite{Ab00}. It was once commonly believed that predators could only affect prey populations by direct consumption of prey. However, even the presence of a predator may alter the behavior and physiology of prey. Prey perceive the risk of predation and respond with a range of anti-predatory responses, such as changes in habitat selection and foraging behavior \cite{peckarsky2008revisiting, Pol89}. These changes, in various forms, may ultimately affect the overall reproductive rate of the prey population. We refer to this particular biological phenomenon as the ``fear" effect. The first experiment on the species fears effect was done by Zanette et al. \cite{6}. They isolated the effects of perceived predation risk in a free-living population of song sparrows by actively eliminating direct predation and used playbacks of predator calls and sounds to manipulate perceived risk. The research showed that under the influence of fear effect only, the number of offspring produced by the species was reduced by $40\%$ annually.
	
	In 2016, Wang et al. \cite{7} considered the fear effect for the first time based on the classical two-species Lotka-Volterra predator-prey model:
	\begin{equation}
		\left\{\begin{array}{l}
			\displaystyle\frac{d x}{d t}=r x f(k, y)-d x-a x^{2}-g(x) y, \vspace{2ex}\\
			\displaystyle\frac{d y}{d t}=-m y+c g(x) y,
		\end{array}\right.
		\label{1.3}
	\end{equation}
	where $a$ represents the mortality rate due to intraspecific competition of the prey, $g(x)$ is the functional predation rate of the predator, and $f(k, y)=\displaystyle\frac{1}{1+ky}$ represents the anti-predation response of the prey due to the fear of the predator, i.e., the fear effect function. The researchers found that under conditions of Hopf bifurcation, an increase in fear level may shift the direction of Hopf bifurcation from supercritical to subcritical when the birth rate of prey increases accordingly. Numerical simulations also suggest that the anti-predator defenses of animals increase as the rate of predator attack increases.
	
	Based on \eqref{1.3}, Sasmal et al. \cite{8} proposed for the first time a predator-prey system with multiplicative Allee effect and fear effect for prey species. The model is specified as follows:
	\begin{equation}
		\left\{\begin{array}{l}
			\displaystyle\frac{d x}{d t}=r x\left(1-\displaystyle\frac{x}{k}\right)(x-\theta) \displaystyle\frac{1}{1+f y}-a x y, \vspace{2ex}\\
			\displaystyle\frac{d y}{d t}=a \alpha x y-m y.
		\end{array}\right.
		\label{1.4}
	\end{equation}
	The authors' study showed that the fear effect did not change the stability of the equilibrium point. However, with the more substantial fear effect, the final population density of the predator will decrease. The Multiplicative Allee effect will cause a subcritical Hopf bifurcation in system \eqref{1.4}, thus producing a stable limit cycle. More related studies can be found in \cite{9,10,11,12,13,14,15}. We note that the Allee effect in competitive systems have been considered as well. This starts with the work of Wang \cite{wang1999competitive}. Herein a weak Allee effect is modeled as affecting both competitors. Jang modeled the strong Allee effect in both competitors \cite{jang2013lotka}. Also, Desilva and Jang model a two species competitive system, with a strong Allee effect in one competitor and stocking effect. These works primarily focus on equilibrium analysis and not bifurcation analysis.
	
	The effect of fear on predator-prey systems has been extensively studied, but in competitive systems fear has rarely been considered. However, there is strong evidence that fear exists in purely competitive systems \emph{without} predation effects, or where predation effects are negligible \cite{Pringle19, 18, Chesson08}.
	The Barred Owl (\emph{Strix varia}) is a species of Owl native to eastern North America. During the last century, they have expanded their range westward and have been recognized as an invasion of the western North American ecosystem. Currently, their range overlaps with the Spotted Owl (\emph{Strix occidentalis}). The Spotted Owl is native to northwestern and western North America, which has led to intense competition between two species \cite{16}. The Barred Owl has a strong negative impact on the  Spotted Owl, and field observations have reported that barred owls frequently attack spotted owls \cite{17}. There is also evidence that barred owls actively and unilaterally drive spotted owls out of shared habitat \cite{18}. 
	There is also other very recent empirical evidence to support such investigations. In \cite{Pringle19} a series of 6 year long experiments are conducted in various Caribean islands that aim to refute the theory of adaptive predation - which suggests that predators reduce dominant competitors, thus preventing competitive exclusion and enhancing coexistence in food webs. However, non-consumptive effects such as fear of depredation can have strongly influencing effects \cite{Pringle19, schmitz2019fearful, peckarsky2008revisiting, schmitz2017predator}. \cite{Pringle19} considers a series of experiments with two competing species of lizards, brown anolis (\emph{Anolis sagrei}) that dwells on tree trunks, and green anolis (\emph{Anolis smaragdinus}) that dwells on tree canopies. The experiments show that typically these species co-exist - due to a clear niche separation. However, the introduction of an intraguild predator, the curly tailed lizard (\emph{Leiocephalus carinatus}) that dwells on the ground, causes (non-consumptive) fear driven effects. The brown anolis being fearful of possible depredation (as the lower half of tree trunks are within striking distance of the curly tailed lizard) moves upwards into the canopy, which is occupied by the green anolis. Herein interspecific competition intensifies leading to a loss of co-existence. However, what is most crucial in this study, is that the fecal analysis of the curly tail lizard shows that its diet included the brown anolis, in only 2 out of 51 samples examined. Thus the new dispersal pattern of the brown anolis and the ``refuge competition" is driven strongly by a non-consumptive fear effect and not a consumptive one. Thus the brown anolis and the curly tail lizards are really competitors, as they have a strong overlap in dietary niche for several insects. However, the brown anolis is clearly fearful of the curly tailed lizard \cite{Pringle19}. There is further evidence of non consumptive effects such as fear among competing aphid species, as well as competitors that feed on aphids \cite{michaud2016extending, bayoumy2018beyond}. Such interplay between competition and predation has been investigated \cite{Chesson08}, where it is proposed that in many ecological processes, competition and predation are interlinked, and depending on niche overlap, one of them will dominate to drive the underlying dynamics.
	
	Such evidence motivates us to consider the fear effect in a purely competitive two-species model, in which one competitor causes fear to the other. Thus, Srivastava et al. \cite{19} considered the classical two-group Lotka-Volterra competition model with only one competitor causing fear to the other competitor:
	\begin{equation}
		\left\{\begin{array}{l}
			\displaystyle\frac{d u}{d t}=a_{1} u-b_{1} u^{2}-c_{1} u v, \vspace{2ex}\\
			\displaystyle\frac{d v}{d t}=\displaystyle\frac{a_{2} v}{1+ku} - b_{2} v^{2}-c_{2} u v.
		\end{array}\right.
		\label{1.5}
	\end{equation}
	
	They find that the presence of fear can have several interesting dynamical effects on the classical competitive scenarios. That is \eqref{1.5} can produce dynamic phenomena such as saddle-node bifurcation and transcritical bifurcation, which are drivers to change the dynamics that we see in classical competitive systems.
	Notably, for fear levels in certain regimes, novel bi-stability dynamics is established. Such dynamics have also been recently observed in \cite{parshad2021some}.
	Furthermore, in the spatially explicit setting, the effects of several spatially heterogeneous fear functions are investigated. Particularly under certain integral restrictions on the fear function, a weak competition type situation can change to competitive exclusion. 
	
	Inspired by the above the ideas in the above works, we propose a novel Lotka-Volterra competition model, where the first species is affected by the multiplicative strong Allee effect, while the second species produces a fear effect on the first species, 
	
	\begin{equation}
		\left\{\begin{array}{l}
			\displaystyle\frac{\mathrm{d} x_1}{\mathrm{d} \tau} =x_1(r_1-\alpha_1 x_1)(x_1-m)\displaystyle\frac{1}{1+kx_2}-\beta_1 x_1 x_2, \vspace{2ex}\\
			\displaystyle\frac{\mathrm{d} x_2}{\mathrm{d} \tau} =x_2(r_2-\alpha_2 x_2-\beta_2 x_1).
		\end{array}\right.
		\label{1}
	\end{equation}
	where $0<m<\displaystyle\frac{r_1}{\alpha_1}$. This paper has the following innovations:
	\begin{enumerate}
		\item This is the first model to consider \emph{both} the strong Allee effect and fear effect on species density in the two-species Lotka-Volterra competition model. Also note by setting $m=0$, we are in a weak Allee type setting.
		\item The Allee effect parameter and the fear effect parameter affect the existence and number of positive equilibrium points in the ODE model.
		\item The Allee effect parameter and the fear effect parameter affect the extinction state of system \eqref{1}, but do not change the stability of the positive equilibria.
		\item Changing the fear parameter leads to several bifurcations in system \eqref{1}.
		\item In the PDE case restrictions are derived between the Allee threshold and the fear parameter, that would yield competitive exclusion of $u$, the competitor subject to the Allee and fear effects.
		\item In the PDE case restrictions are also derived to provide conditions under which one has initial condition dependent attraction to $(u^{*},0)$ or $(0,v^{*})$.
	\end{enumerate}
	
	The rest of this paper is organized as follows: The positivity and boundedness of the solution of system \eqref{1} are proved in Section 2. We examine the existence and stability of all equilibria in Section 3 and Section 4. In Section 5, we analyze the bifurcation of the system around the positive equilibria. In Section 6 we present the analysis of the spatially explicit model or the PDE case. We end this paper with a discussion and conclusion.
	
	\section{Preliminaries}
	
	The intrinsic growth rates of either species are $r_1$, $r_2$, and the environmental carrying capacities are $r_1/\alpha_1$ and $r_2/\alpha_2$, respectively, according to the logistic rule of growth. We changed system \eqref{1} to the following form:
	\begin{equation}
		\left\{\begin{array}{l}
			\displaystyle\frac{\mathrm{d} x_1}{\mathrm{d} \tau} =r_1 x_1(1-\displaystyle\frac{x_1}{k_1})(x_1-m)\displaystyle\frac{1}{1+kx_2}-\beta_1 x_1 x_2, \vspace{2ex}\\
			\displaystyle\frac{\mathrm{d} x_2}{\mathrm{d} \tau} =r_2 x_2(1-\displaystyle\frac{x_2}{k_2})-\beta_2 x_1x_2.
		\end{array}\right.
		\label{2}
	\end{equation}
	
	In order to reduce the parameters of system \eqref{2}, the following dimensionless quantities are applied to the non-dimensionalize model system \eqref{2}
	\begin{equation}
		t=r_1 k_1\tau,\quad\frac{x_1}{k_1}=x,\quad\frac{x_2}{k_2}=y,\quad \frac{m}{k_1}=p, \quad kk_2=q, \quad\frac{\beta_1 k_2}{r_1 k_1}=a, \quad\frac{r_2}{r_1 k_1}=b, \quad\frac{\beta_2 k_1}{r_2}=c,\nonumber
	\end{equation}
	then system \eqref{2} becomes the following system:
	\begin{equation}
		\left\{\begin{array}{l}
			\displaystyle\frac{\mathrm{d} x}{\mathrm{d} t} =x \left [(1-x)(x-p)\displaystyle\frac{1}{1+qy}-ay \right ]=xf(x,y)\equiv F(x,y), \vspace{2ex}\\
			\displaystyle\frac{\mathrm{d} y}{\mathrm{d} t} =by\left (1-y-cx \right )=yg(x,y)\equiv G(x,y).
		\end{array}\right.
		\label{3}
	\end{equation}
	All parameters in system \eqref{3} are positive and $0<p<1$. Based on biological considerations, the initial condition of system \eqref{3} satisfies
	\begin{equation}
		x(0)>0, \quad y(0)>0.
		\label{4}
	\end{equation}
	
	\begin{proposition}
		All solutions of system \eqref{3} are positive.
	\end{proposition}
	\begin{proof}
		Since
		$$
		x(t)=x(0) \mathrm{exp} \left [ \int_{0}^{t} f(x(s),y(s))\mathrm{d}s\ \right ] >0,
		$$
		and
		$$
		y(t)=y(0) \mathrm{exp} \left [ \int_{0}^{t} g(x(s),y(s))\mathrm{d}s\ \right ] >0.
		$$
		So all solutions of system \eqref{3} with initial condition \eqref{4} are positive.
	\end{proof}
	
	\begin{lemma} \cite{20}
		If $a,b>0$ and $x^{'}(t)\le (\ge)x(t)(a-bx(t))$ with $x(0)>0$, then
		$$
		\limsup\limits_{t\rightarrow +\infty} x(t) \le \frac{a}{b}( \liminf\limits_{t\rightarrow +\infty} x(t) \ge \frac{a}{b} ).
		$$
	\end{lemma}
	\begin{proposition}
		The solutions of system \eqref{3} are bounded.
	\end{proposition}
	
	\begin{proof}
		For the boundedness of $y(t)$, according to the second equation of system \eqref{3},
		$$
		\displaystyle\frac{\mathrm{d} y}{\mathrm{d} t} =by\left (1-y-cx \right ) \le y(b-by),
		$$
		by applying Lemma 2.2 to the above inequality, we have
		$$
		\limsup\limits_{t\rightarrow +\infty} y(t) \le \frac{b}{b} =1.
		$$
		
		Next, we discuss the boundedness of $x(t)$. For any $x(0)>1$,
		$$
		\displaystyle\frac{\mathrm{d} x}{\mathrm{d} t} =x \left [(1-x)(x-p)\displaystyle\frac{1}{1+qy}-ay \right ]<0
		$$
		as long as $x>1$. Moreover, along $x=1$, we have
		$$
		\displaystyle\frac{\mathrm{d} x}{\mathrm{d} t} =-ay <0.
		$$
		Obviously, there is no equilibrium point in the region $\left \{(x,y)\mid x>1,y\ge 0 \right \} $. Thus, any positive solution satisfies $x(t)\le\mathrm{max} \left \{ x(0),1 \right \}=\overline{x}$ (say) for all $t\ge0$. Therefore, combining the above analysis, the solutions of system \eqref{3} are bounded.
	\end{proof}
	
	\begin{proposition}
		If $0 < x(0) \le p$, $y(0) \ge 0$ and $(x(0), y(0))\ne(p, 0)$, then $\lim\limits_{t \rightarrow \infty } (x(t),y(t))=(0,0)$.
	\end{proposition}
	\begin{proof}
		The proof is similar to Theorem 2.1 in \cite{21}, and we omit the detailed procedure here.
	\end{proof}
	
	\section{Existence of Equilibria}
	Obviously, system \eqref{3} has a constant equilibrium point $E_0(0,0)$ and three boundary equilibria $E_1(1,0)$, $E_2(p,0)$, $E_3(0,1)$. In the following, we discuss the existence of positive equilibria.
	
	The intersections of two isoclines $f(x,y)=0$, $g(x,y)=0$ in the first quadrant is the point of positive equilibria. Denote the positive equilibria of system \eqref{3} as $E_{i*}(x_i,y_i)$ (i=1, 2), from $f(x,y)=g(x,y)$, we obtain
	\begin{equation}
		A_1x^2+A_2x+A_3=0,
		\label{5}
	\end{equation}
	where
	$$
	A_1=ac^2q+1>0,
	$$
	$$
	A_2=-(2acq+ac+p+1)<0,
	$$
	$$
	A_3=a+aq+p>0.
	$$
	Denote the discriminant of \eqref{5} as $\Delta(q)=A_2^2-4A_1A_3$.
	When $\Delta>0$, \eqref{5} has two real roots, which can be expressed as follows:
	$$
	x_1=\displaystyle\frac{-A_2-\sqrt{\Delta} }{2A_1}, \quad x_2=\displaystyle\frac{-A_2+\sqrt{\Delta} }{2A_1}.
	$$
	The real root of \eqref{5} is $x_i$ and $y_i=1-cx_i$. From Proposition 2.1 and 2.3, we know that $0< x(t) < 1$ and $0< y(t) < 1$, from which we give the following theorem.
	
	\begin{theorem}\label{thm:exist}
		The positive equilibrium point of the system (2.2) is shown below:
		
		CASE \uppercase\expandafter{\romannumeral1}: $\Delta>0$
		\begin{enumerate}
			\item $2A_1+A_2>0$
			\begin{enumerate}
				\item For $0<c<1$
				\begin{enumerate}
					\item System \eqref{3} exists two positive equilibria $E_{1*}$ and $E_{2*}$(Fig. 3(a)).
				\end{enumerate}
				\item For $c=1$
				\begin{enumerate}
					\item System \eqref{3} exists only one positive equilibrium point $E_{1*}$(Fig. 3(b)).
				\end{enumerate}
				\item For $1<c\le \displaystyle\frac{1}{q}+1$
				\begin{enumerate}
					\item System \eqref{3} exists only one positive equilibrium point $E_{1*}$ if $0<p<\displaystyle\frac{1}{c}$(Fig. 3(c)).
				\end{enumerate}
			\end{enumerate}
			\item $2A_1+A_2<0$
			\begin{enumerate}
				\item For $1<c<\displaystyle\frac{1}{q}+1$
				\begin{enumerate}
					\item System \eqref{3} exists only one positive equilibrium point $E_{1*}$ if $0<p<\displaystyle\frac{1}{c}$(Fig. 3(d)).
				\end{enumerate}
			\end{enumerate}
			\item $2A_1+A_2=0$
			\begin{enumerate}
				\item For $c>1$
				\begin{enumerate}
					\item System \eqref{3} exists only one positive equilibrium point $E_{1*}$ if $0<p<\displaystyle\frac{1}{c}$(Fig. 3(e)).
				\end{enumerate}
			\end{enumerate}
		\end{enumerate}
		CASE \uppercase\expandafter{\romannumeral2}: $\Delta=0$
		\begin{enumerate}
			\item $2A_1+A_2>0$
			\begin{enumerate}
				\item For $0<c<1$
				\begin{enumerate}
					\item System \eqref{3} exists only one positive equilibrium point $E_{3*}$(Fig. 3(f)).
				\end{enumerate}
			\end{enumerate}
		\end{enumerate}
	\end{theorem}
	\begin{proof}
		\it {For} \it {CASE \uppercase\expandafter{\romannumeral1}}, we know that $0<x_1<x_2$. When the condition $2A_1+A_2>0$ is satisfied, $0<x_1<1$ is always true. Next, we discuss $x_2$. If $0<x_2<1$, a simple calculation shows that the parameter $c$ must satisfy $0<c<1$ or $c>\displaystyle\frac{1}{q}+1$. For $0<c<1$, $y_2=1-cx_2$ obviously satisfies $0 < y_2 < 1$. For $c>\displaystyle\frac{1}{q}+1$, if $0 < y_2 < 1$ is true then the parameters of system \eqref{3} must also satisfy $2+c(-ac-p-1)>0$ and $cp>1$. However, the above two conditions cannot hold simultaneously, so $c>\displaystyle\frac{1}{q}+1$ is also not valid either. On the contrary, if $x_2\ge1$, then we can get $1\le c\le \displaystyle\frac{1}{q}+1$. When $c=1$, we can guarantee that the equilibrium point $E_{1*}$ must be a positive equilibrium point. If $1<c\le \displaystyle\frac{1}{q}+1$, it can be seen by calculating $0<y_1<1$ that the equilibrium point $E_{1*}$ can only be in the first quadrant if $0 < p < \displaystyle\frac{1}{c}$.
		
		When the condition $2A_1+A_2<0$ is satisfied, $x_2>1$ always holds. Let us discuss the conditions under which $0<x_1 < 1$ can hold. The inequality shows that the parameter $c$ needs to satisfy $1<c<\displaystyle\frac{1}{q}+1$, but $E_{1*}$ may be located in the first or fourth quadrant. Thus we make $0 < p < \displaystyle\frac{1}{c}$ so that the equilibrium point $E_{1*}$ is a positive equilibrium point.
		
		When the condition $2A_1+A_2=0$ is satisfied, $0<x_1<1<x_2$ always holds. If $0 < c< 1$, based on $2A_1+A_2=0$ we can find $a = \displaystyle\frac{p -1}{c \left(2 c q -2 q -1\right)}\triangleq a_{*}$. Substitute $a=a_{*}$ into $\Delta$ and get
		$$
		\Delta=-\displaystyle\frac{4 \left(-1+\left(c -1\right) q \right) \left(-1+\left(-2+\left(p +1\right) c \right) q \right) \left(c -1\right) \left(p -1\right)}{\left(2 c q -2 q -1\right)^{2} c}<0.
		$$
		Therefore $2A_1+A_2=0$ and $\Delta>0$ cannot hold simultaneously, i.e., $c\notin (0,1)$. While if $c > 1$, it also needs to satisfy both $0 < p < \displaystyle\frac{1}{c}$. Let $c=1$, and then the calculation finds that $2A_1+A_2=0$ will also make $\Delta$ equal to $0$. The results contradict the previous assumptions, i.e., $c\ne 1$.
		
		For \it {CASE \uppercase\expandafter{\romannumeral2}}, in order for $0<x_3<1$ to hold, then $2A_1+A_2>0$ must be satisfied. $E_{1*}$ is a positive equilibrium point of system \eqref{3} when $0<c < 1$. Conversely when $c > 1$, in order to satisfy $y_3 > 0$, the parameters of system \eqref{3} also need to satisfy $a c^{2}+c p +c -2< 0$. From $\Delta=0$ we get
		$$
		q=\frac{a^{2} c^{2}+\left(-4+\left(2 p +2\right) c \right) a +\left(p -1\right)^{2}}{4 a \left(c -1\right) \left(c p -1\right)}\triangleq q_{*}.
		$$
		Substitute $q=q_{*}$ into $2A_1+A_2$ and get
		$$
		\displaystyle\frac{\left(a c -p +1\right) \left(a \,c^{2}+c p +c -2\right)}{2(cp-1)}.
		$$
		The condition $0<cp<1$ must be satisfied if the above equation is greater than 0. But under $0<cp<1$, we get $q=q_* < 0$. In summary, $c\notin (1,+\infty )$. Let $c=1$, and then the calculation finds that $\Delta=0$ will also make $2A_1+A_2$ equal to $0$. The results contradict the previous assumptions, i.e., $c\ne 1$.
	\end{proof}
	
	\section{Stability of Equilibria}
	\subsection{Stability of boundary equilibria}
	The Jacobian matrix of system \eqref{3} is
	\begin{equation}
		J(E)=\begin{bmatrix}\displaystyle\frac{\left(1-x \right) \left(x -p \right)}{q y +1}-a y +x \left(-\displaystyle\frac{x -p}{q y +1}+\displaystyle\frac{1-x}{q y +1}\right)
			& x \left(-\displaystyle\frac{\left(1-x \right) \left(x -p \right) q}{\left(q y +1\right)^{2}}-a \right)\\
			-b c y &b \left(-c x -y +1\right)-b y
		\end{bmatrix}.
		\label{6}
	\end{equation}
	
	The Jacobian matrix at $E_0(0,0)$ is given by
	$$
	J(E_0)=\begin{bmatrix}-p
		& 0\\
		0 &b
	\end{bmatrix}.
	$$
	The two eigenvalues of $J(E_0)$ are $\lambda_{10}=-p<0$ and $\lambda_{20}=b>0$, so $E_0$ is a saddle.
	
	The Jacobian matrix at $E_3(0,1)$ is given by
	$$
	J(E_3)=\begin{bmatrix}-\displaystyle\frac{p}{q+1}-a
		& 0\\
		-bc &-b
	\end{bmatrix}.
	$$
	The two eigenvalues of $J(E_3)$ are $\lambda_{13}=-\displaystyle\frac{p}{q+1}-a<0$ and $\lambda_{23}=-b<0$, so $E_3$ is a stable node. Then we discuss the stability of the boundary equilibria $E_1(1,0)$, $E_2(p,0)$.
	\begin{theorem}
		The stability of the boundary equilibrium point $E_1$ is shown below:
		\begin{enumerate}
			\item If $c>1$, $E_1$ is a stable node (Fig. 1(a)).
			\item If $0<c<1$, $E_1$ is a saddle (Fig. 1(b)).
			\item If $c=1$,
			\begin{enumerate}
				\item $E_1$ is an attracting saddle-node, and the parabolic sector is on the lower half-plane when $p<1-a$ (Fig. 1(c)).
				\item $E_1$ is an attracting saddle-node, and the parabolic sector is on the upper half-plane when $p>1-a$ (Fig. 1(d)).
				\item $E_1$ is a nonhyperbolic saddle when $p=1-a$ (Fig. 1(e)).
			\end{enumerate}
		\end{enumerate}
	\end{theorem}
	\begin{proof}
		The Jacobian matrix at $E_1(1,0)$ is given by
		$$
		J(E_1)=\begin{bmatrix}-1+p
			& -a\\
			0 &b (-c+1)
		\end{bmatrix}.
		$$
		The two eigenvalues of $J(E_1)$ are $\lambda_{11}=-1+p<0$ and $\lambda_{21}=b(-c+1)$. $E_1$ is a stable node if $\lambda_{21}=b(-c+1)<0$, i.e., $c>1$. $E_1$ is a saddle if $\lambda_{21}=b(-c+1)>0$, i.e., $0<c<1$. For $\lambda_{21}=b(-c+1)=0$, i.e., $c=1$, we conduct the following discussion.
		
		We move equilibrium $E_1$ to the origin by transforming $(X, Y)=(x-1, y)$ and make Taylor's expansion around the origin, then system \eqref{3} becomes
		$$
		\left\{\begin{array}{l}
			\displaystyle\frac{\mathrm{d} X}{\mathrm{d} t} =\left(p -1\right) X -a Y +\left(-2+p \right) X^{2}-\left(q p +a -q \right) Y X -X^{3}-q \left(-2+p \right) Y \,X^{2}+\left(p -1\right) q^{2} X \,Y^{2}+P_0(X, Y), \vspace{2ex}\\
			\displaystyle\frac{\mathrm{d} Y}{\mathrm{d} t} =-bXY-bY^2,
		\end{array}\right.
		$$
		where $P_i(X, Y)$ are power series in $(X, Y)$ with terms $X^IY^J$ satisfying $I+J \ge 4$ (the same below).
		
		In the next step, we make the following transformations to the above system
		$$
		\begin{bmatrix}
			X \\Y
			
		\end{bmatrix}=\begin{bmatrix}
			1 &\displaystyle \frac{a}{p-1} \\
			0 &1
		\end{bmatrix}\begin{bmatrix}
			X_1 \\Y_1
			
		\end{bmatrix},
		$$
		and letting $\tau=(p-1) t$, for which we will retain $t$ to denote $\tau$ for notational simplicity, we get
		\begin{equation}
			\left\{\begin{array}{l}
				\displaystyle\frac{\mathrm{d} X_1}{\mathrm{d} t} =X_1+a_{11}X_1Y_1+a_{20}X^2_1+a_{02}Y^2_1+a_{30}X_1^3+a_{03}Y_1^3+a_{21}X_1^2Y_1+a_{12}X_1Y_1^2+P_1(X_1,Y_1), \vspace{2ex}\\
				\displaystyle\frac{\mathrm{d} Y_1}{\mathrm{d} t} =b_{02}Y_1^2+b_{11}X_1Y_1,
			\end{array}\right.
			\label{7}
		\end{equation}
		where
		\begin{flalign}
			\begin{split}
				& a_{11}=\frac{-p^{2} q +a b +a p +2 p q -3 a -q}{\left(p -1\right)^{2}}, \quad a_{20}=\frac{-2+p}{p -1}, \quad a_{02}=\frac{a \left(-p^{2} q +a b +b p +2 p q -a -b -q \right)}{\left(p -1\right)^{3}},\\
				& a_{30}=-\frac{1}{p -1}, \quad a_{21}=-\frac{p^{2} q -3 p q +3 a +2 q}{\left(p -1\right)^{2}}, \quad a_{12}=-\frac{-p^{3} q^{2}+2 a \,p^{2} q +3 p^{2} q^{2}-6 a p q -3 p q^{2}+3 a^{2}+4 a q +q^{2}}{\left(p -1\right)^{3}},\\
				& \quad a_{03}=-\frac{a \left(-p q +a +q \right) \left(p^{2} q -2 p q +a +q \right)}{\left(p -1\right)^{4}}, \quad b_{11}=-\frac{b}{p -1}, \quad b_{02}=-\frac{b \left(a +p -1\right)}{\left(p -1\right)^{2}}.\nonumber
			\end{split}&
		\end{flalign}
		\begin{figure}[H]
			\centering
			\vspace{-0cm}
			\setlength{\abovecaptionskip}{-0pt}
			\subfigcapskip=-30pt
			\subfigure[]
			{\scalebox{0.45}[0.45]{
					\includegraphics{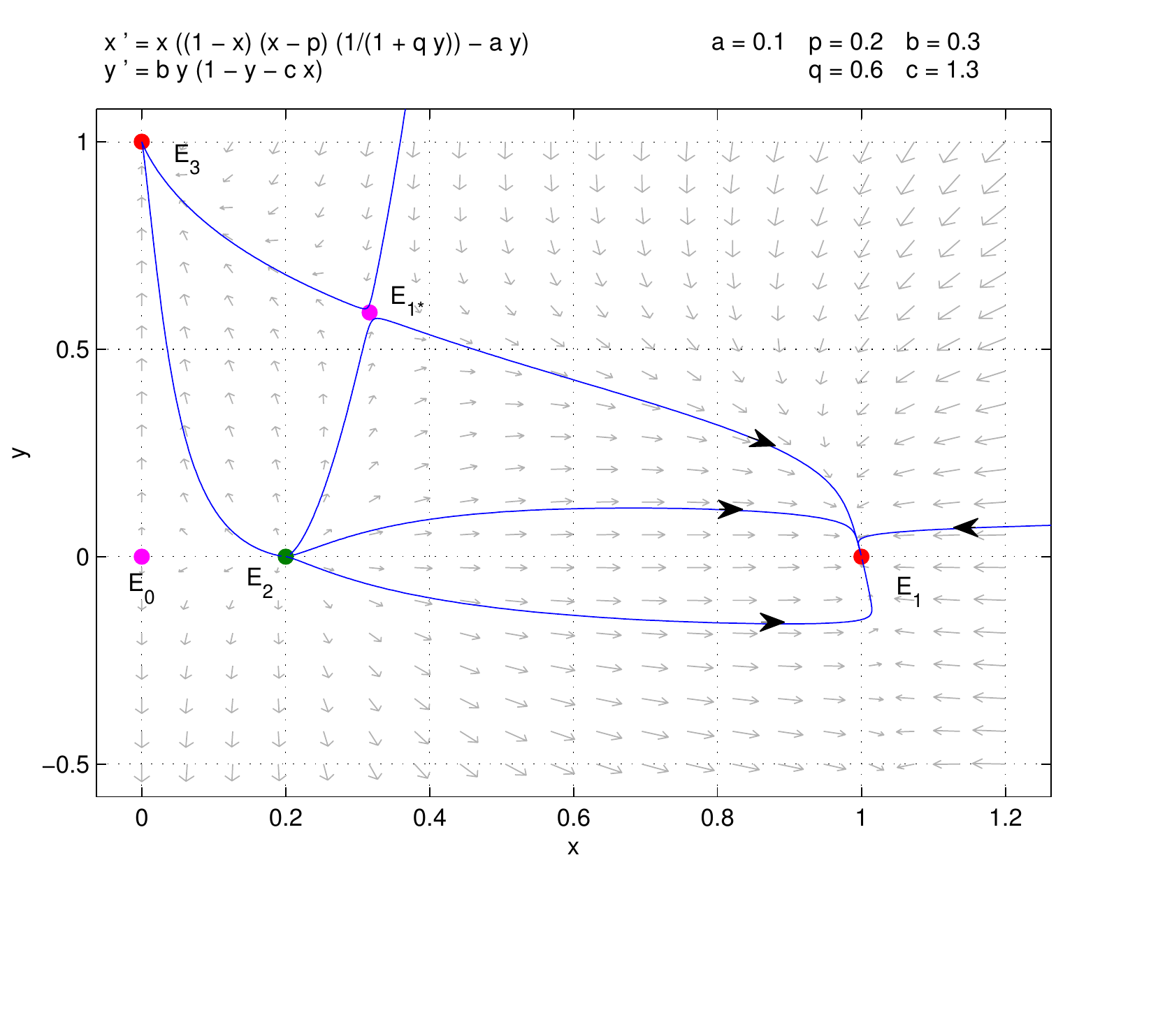}}}
			\subfigure[]
			{\scalebox{0.45}[0.45]{
					\includegraphics{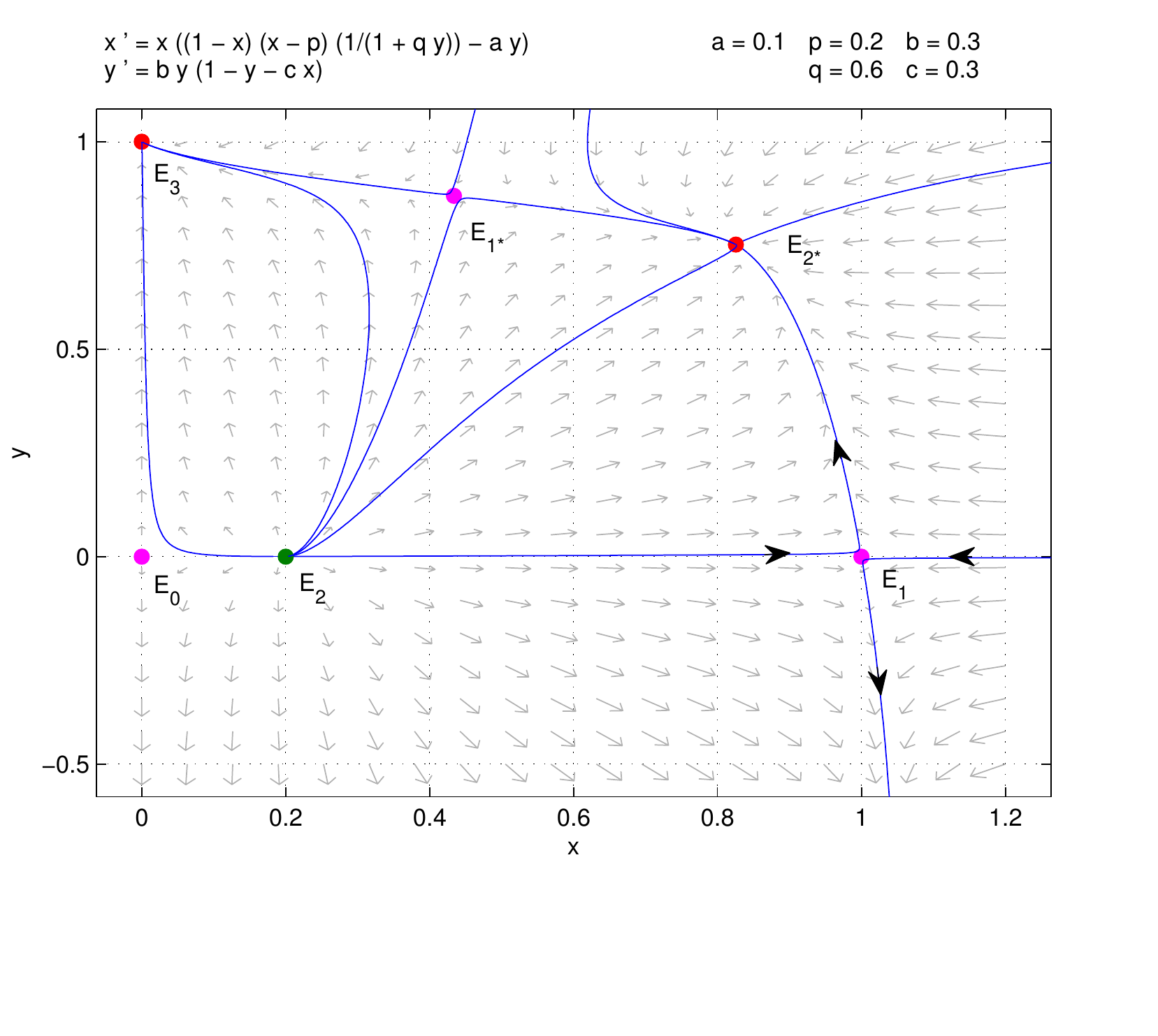}}}
			\subfigure[]
			{\scalebox{0.45}[0.45]{
					\includegraphics{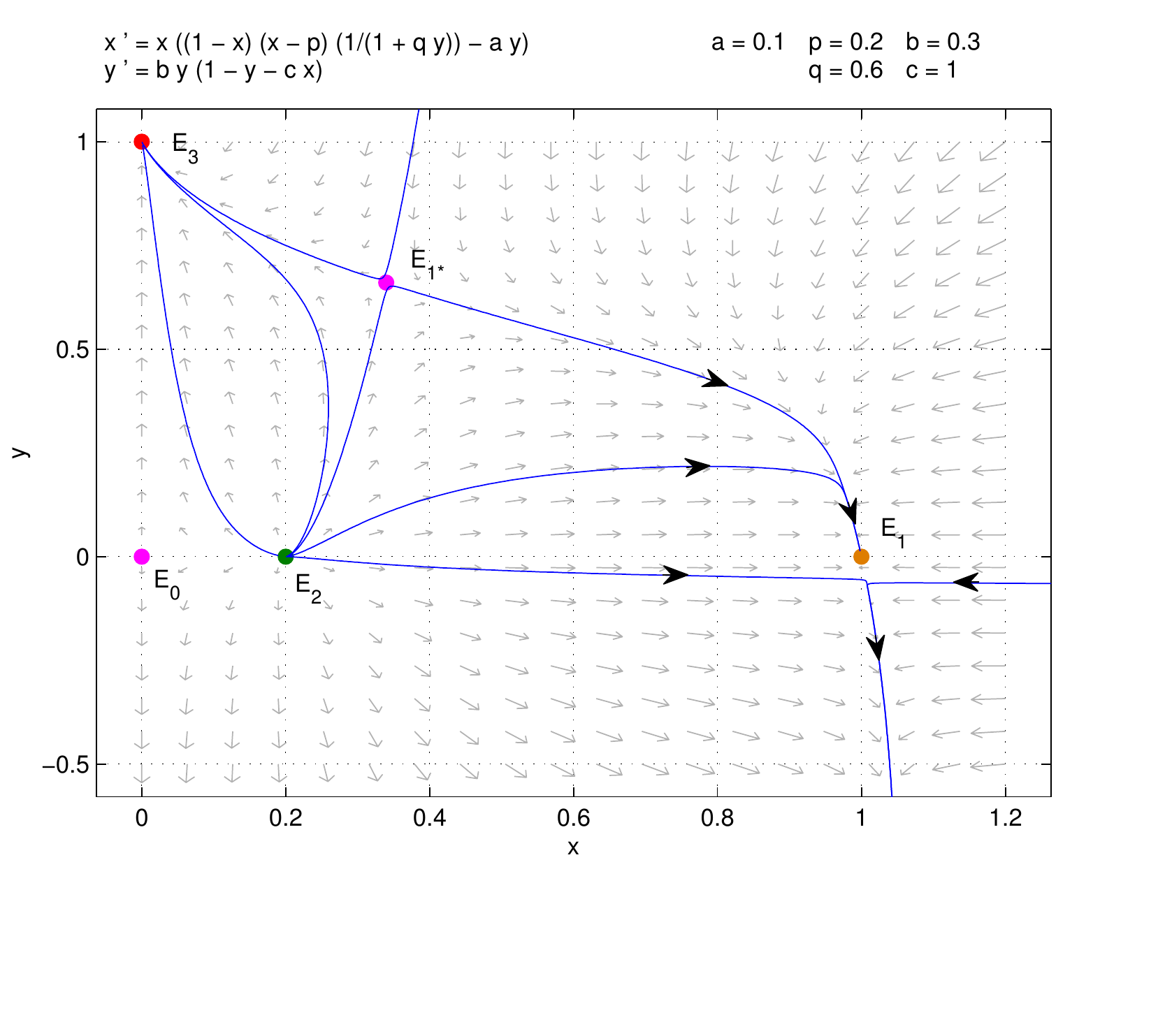}}}
			\subfigure[]
			{\scalebox{0.45}[0.45]{
					\includegraphics{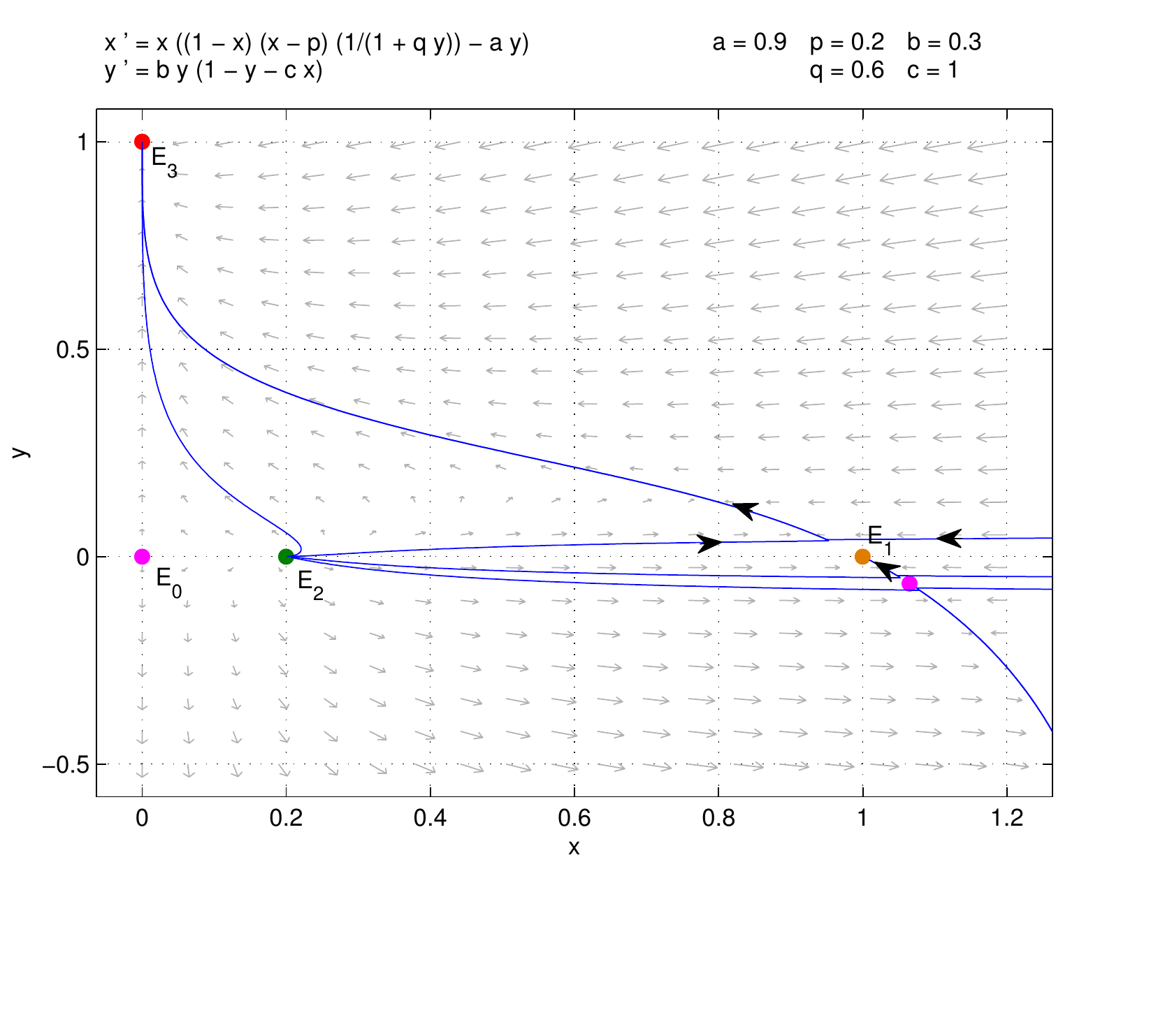}}}
			\subfigure[]
			{\scalebox{0.45}[0.45]{
					\includegraphics{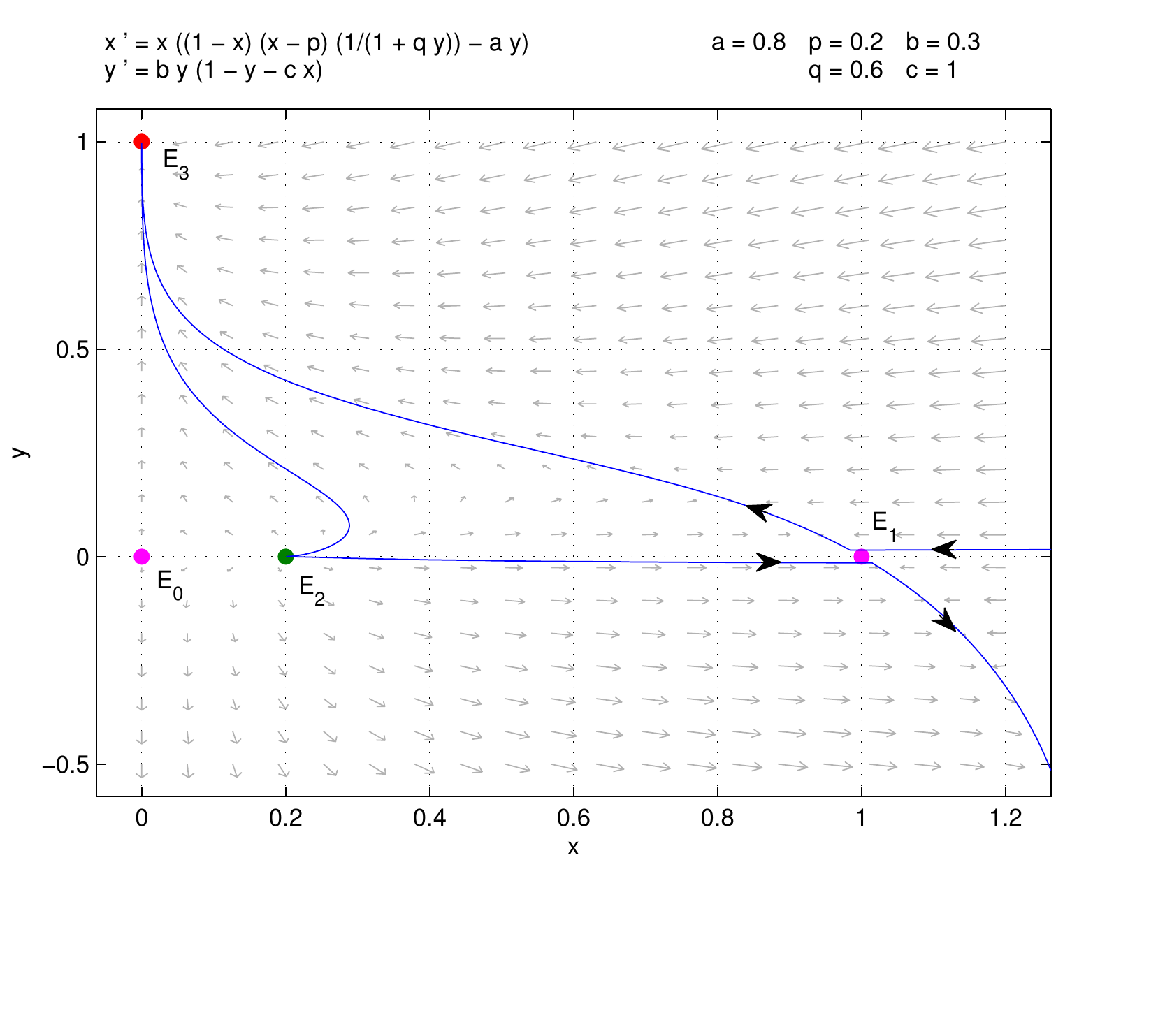}}}
			\caption{Red, green, pink, and orange points indicate stable node, unstable node (source), saddle, and saddle-node, respectively. (a) $c>1$. (b) $0<c<1$. (c) $c=1$ and $p<1-a$. (d) $c=1$ and $p>1-a$. (e) $c=1$ and $p=1-a$.}
			\label{Fig1}
		\end{figure} 
		Hence by Theorem 7.1 in Chapter 2 in \cite{22}, if $b_{02}>0$, i.e., $p<1-a$, $E_1$ is an attracting saddle-node, and the parabolic sector is on the lower half-plane. If $b_{02}<0$, i.e., $p>1-a$, $E_1$ is an attracting saddle-node, and the parabolic sector is on the upper half-plane. If $b_{02}=0$, i.e. $p=1-a$, system \eqref{7} becomes
		\begin{equation}
			\left\{\begin{array}{l}
				\displaystyle\frac{\mathrm{d} X_1}{\mathrm{d} t} =X_1+a_{11}X_1Y_1+a_{20}X^2_1+a_{02}Y^2_1+a_{30}X_1^3+a_{03}Y_1^3+a_{21}X_1^2Y_1+a_{12}X_1Y_1^2+P_1(X_1,Y_1), \vspace{2ex}\\
				\displaystyle\frac{\mathrm{d} Y_1}{\mathrm{d} t} =b_{11}X_1Y_1.
			\end{array}\right.
			\label{8}
		\end{equation}
		By using the first equation of system \eqref{8}, we obtain the implicit function
		$$
		X_1=-a_{02}Y_1^2+(a_{11}a_{02}-a_{03})Y_1^3+\cdots \cdots
		$$
		and
		$$
		\displaystyle\frac{\mathrm{d} Y_1}{\mathrm{d} t}=-a_{02}b_{11}Y_1^3+\cdots \cdots,
		$$
		where
		$$
		-a_{02}b_{11}=\displaystyle\frac{ab\left [-q(p-1)^2-a\right ] }{(p-1)^4}<0.
		$$
		According to Theorem 7.1 again, $E_1$ is a nonhyperbolic saddle.
	\end{proof}
	
	\begin{theorem}
		The stability of the boundary equilibrium point $E_2$ is shown below:
		\begin{enumerate}
			\item If $0<p<\displaystyle\frac{1}{c}$, $E_2$ is a unstable node (Fig. 2(a)).
			\item If $p>\displaystyle\frac{1}{c}$, $E_2$ is a saddle (Fig. 2(b)).
			\item If $p=\displaystyle\frac{1}{c}$, $E_2$ is a repelling saddle-node (Fig. 2(c)).
		\end{enumerate}
	\end{theorem}
	\begin{proof}
		The Jacobian matrix at $E_2(p,0)$ is given by
		$$
		J(E_2)=\begin{bmatrix}p(1-p)
			& -pa\\
			0 &b (-cp+1)
		\end{bmatrix}.
		$$
		The two eigenvalues of $J(E_2)$ are $\lambda_{12}=p(1-p)>0$ and $\lambda_{22}=b (-cp+1)$. $E_2$ is an unstable node if $\lambda_{22}=b (-cp+1)>0$, i.e., $0<p<\displaystyle\frac{1}{c}$. $E_2$ is a saddle if $\lambda_{22}=b (-cp+1)<0$, i.e., $p>\displaystyle\frac{1}{c}$. For $\lambda_{22}=b (-cp+1)=0$, i.e. $p=\displaystyle\frac{1}{c}$, we conduct the following discussion.
		
		We move equilibrium $E_2$ to the origin by transforming $(X, Y)=(x-p, y)$ and make Taylor's expansion around the origin, then system \eqref{3} becomes
		$$
		\left\{\begin{array}{l}
			\begin{aligned}
				\displaystyle\frac{\mathrm{d} X}{\mathrm{d} t}=&-p \left(p -1\right) X -p a Y -\left(2 p -1\right) X^{2}-\left(-p^{2} q +p q +a \right) Y X -X^{3}\\
				&+q \left(2 p -1\right) YX^{2}-p \left(p -1\right) q^{2} XY^{2}+P_2(X, Y),
			\end{aligned}
			\vspace{2ex}\\
			\displaystyle\frac{\mathrm{d} Y}{\mathrm{d} t} =-\frac{b Y X}{p}-bY^{2}.
		\end{array}\right.
		$$
		\begin{figure}[H]
			\centering
			\vspace{-0cm}
			\setlength{\abovecaptionskip}{-0pt}
			\subfigcapskip=-30pt
			\subfigure[]
			{\scalebox{0.45}[0.45]{
					\includegraphics{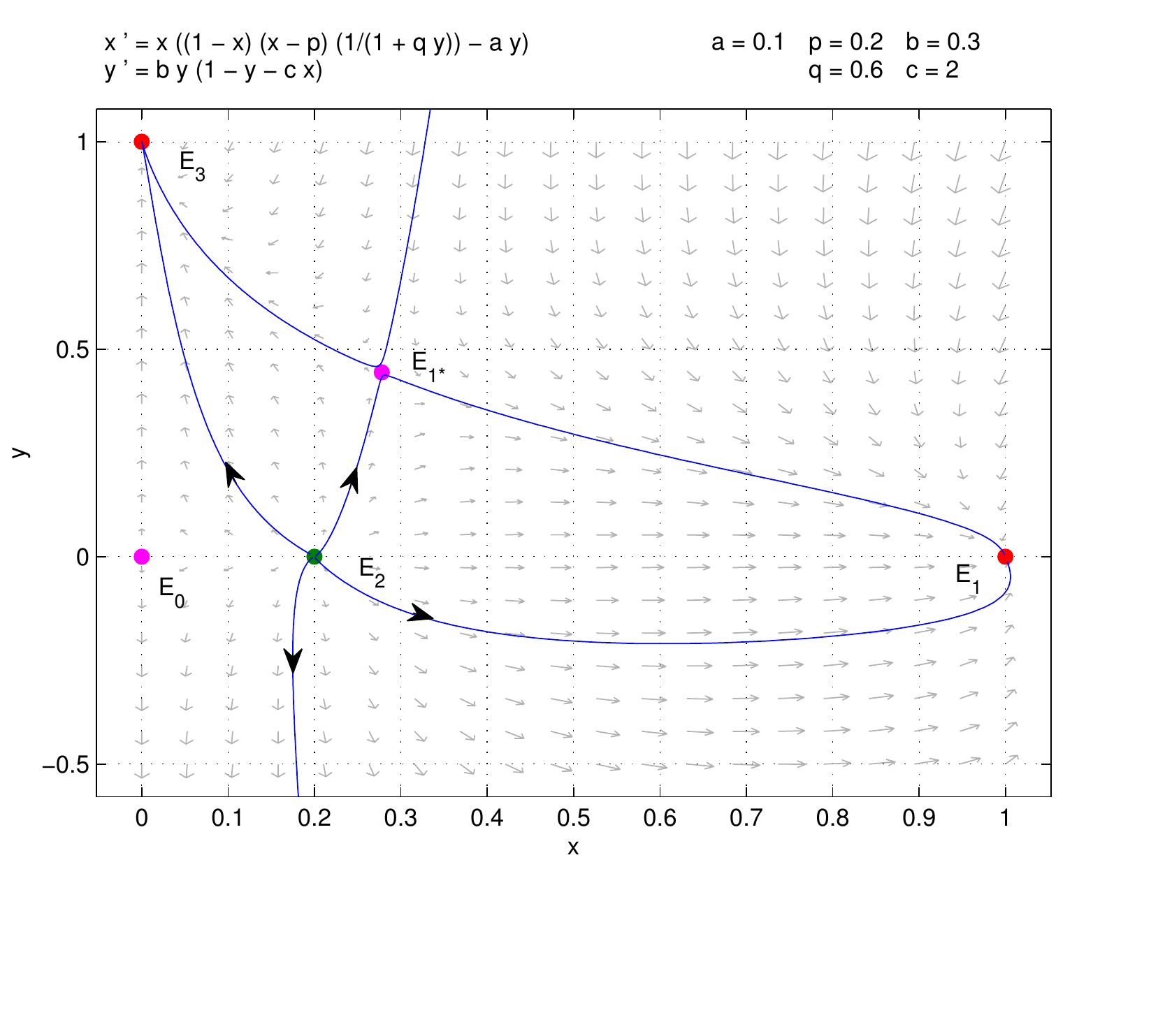}}}
			\subfigure[]
			{\scalebox{0.45}[0.45]{
					\includegraphics{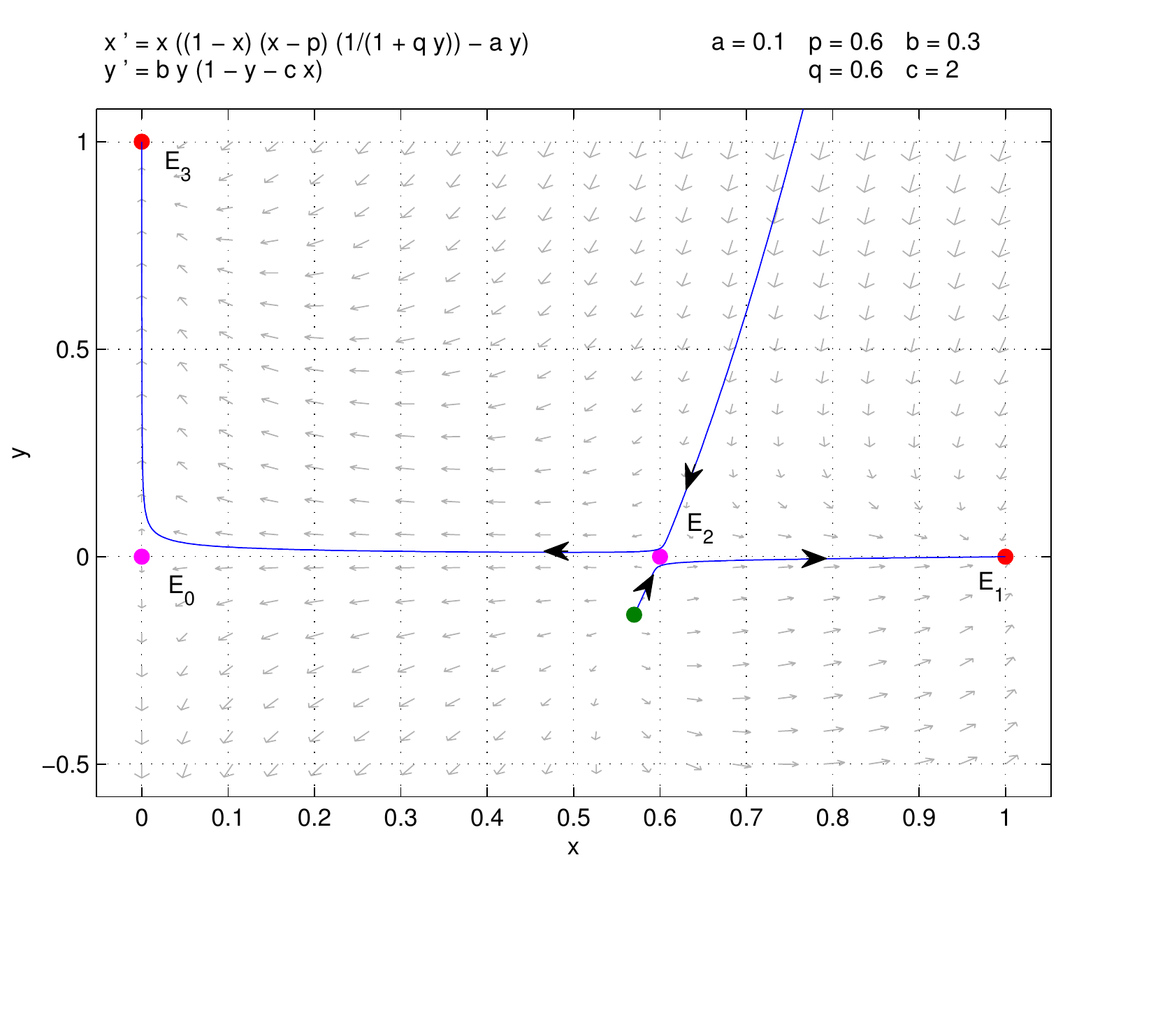}}}
			\subfigure[]
			{\scalebox{0.45}[0.45]{
					\includegraphics{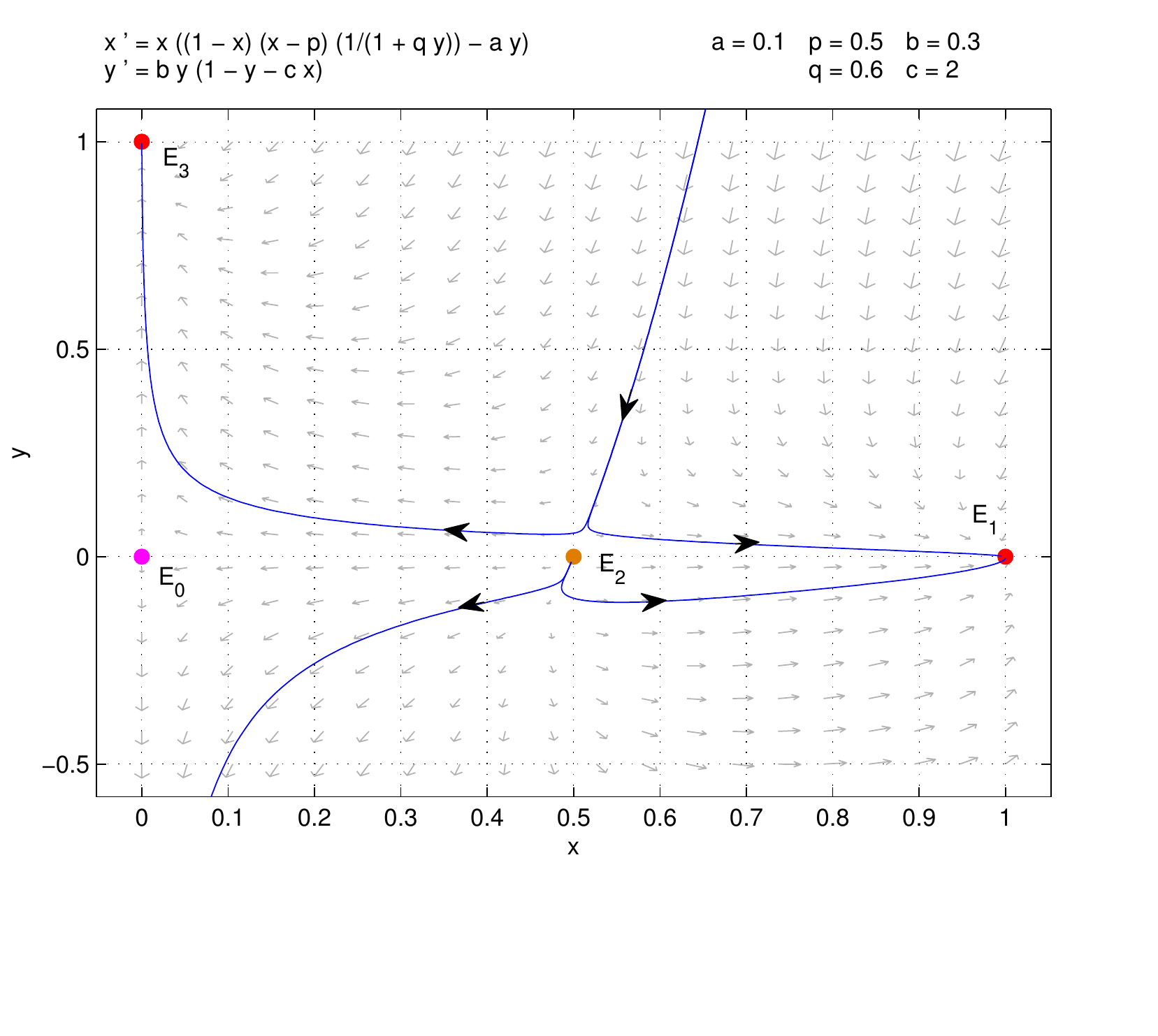}}}
			\caption{Red, green, pink, and orange points indicate stable node, unstable node (source), saddle, and saddle-node, respectively. (a) $0<p<\displaystyle\frac{1}{c}$. (b) $p>\displaystyle\frac{1}{c}$. (c) $p=\displaystyle\frac{1}{c}$.}
			\label{Fig2}
		\end{figure}
		In the next step, we make the following transformations to the above system
		$$
		\begin{bmatrix}
			X \\Y
			
		\end{bmatrix}=\begin{bmatrix}
			1 &\displaystyle \frac{a}{1-p} \\
			0 &1
		\end{bmatrix}\begin{bmatrix}
			X_1 \\Y_1
			
		\end{bmatrix},
		$$
		and letting $\tau=-p(p-1) t$, for which we will retain $t$ to denote $\tau$ for notational simplicity, we get
		\begin{equation}
			\left\{\begin{array}{l}
				\displaystyle\frac{\mathrm{d} X_1}{\mathrm{d} t} =X_1+c_{11}X_1Y_1+c_{20}X^2_1+c_{02}Y^2_1+c_{30}X_1^3+c_{03}Y_1^3+c_{21}X_1^2Y_1+c_{12}X_1Y_1^2+P_3(X_1,Y_1), \vspace{2ex}\\
				\displaystyle\frac{\mathrm{d} Y_1}{\mathrm{d} t} =d_{02}Y_1^2+d_{11}X_1Y_1,
			\end{array}\right.
			\label{9}
		\end{equation}
		where
		\begin{flalign}
			\begin{split}
				& c_{11}=\frac{-p^{4} q +2 p^{3} q -3 ap^{2}-p^{2} q +a b +p a}{p^{2} \left(p -1\right)^{2}}, \quad c_{20}=\frac{2 p -1}{p \left(p -1\right)}, \quad c_{02}=-\frac{a \left(-p^{4} q +2 p^{3} q -a \,p^{2}-b \,p^{2}-p^{2} q +a b +b p \right)}{\left(p -1\right)^{3} p^{2}},\\
				& c_{30}=\frac{1}{p \left(p -1\right)}, \quad c_{21}=-\frac{2 p^{2} q -3 p q +3 a +q}{p \left(p -1\right)^{2}}, \quad c_{12}=\frac{p^{4} q^{2}-3 p^{3} q^{2}+4 a \,p^{2} q +3 p^{2} q^{2}-6 a p q -p \,q^{2}+3 a^{2}+2 a q}{p \left(p -1\right)^{3}},\\
				& c_{03}=-\frac{a \left(p^{2} q -p q +a \right) \left(p^{2} q -2 p q +a +q \right)}{\left(p -1\right)^{4} p}, \quad d_{11}=\frac{b}{p^{2} \left(p -1\right)}, \quad d_{02}=-\frac{b \left(-p^{2}+a +p \right)}{p^{2} \left(p -1\right)^{2}}.\nonumber
			\end{split}&
		\end{flalign}
		Since $-p(p-1)>0$ and $d_{02}<0$, $E_2$ is a repelling saddle-node.
	\end{proof}
	\begin{remark}
		By analyzing the stability of system \eqref{3} boundary equilibria, we found that the Allee effect parameter $p$ significantly affects the extinction of the species. This conclusion is consistent with the ecological implications of the strong Allee effect. The above theoretical analysis has guiding impacts on the conservation of environmental diversity.
	\end{remark}
	
	\subsection{Stability of positive equilibria}
	Transform the Jacobian matrix of the positive equilibria $E_{i*}$ of system \eqref{3} into
	\begin{equation}
		\begin{bmatrix}
			x_*\displaystyle\frac{\partial f}{\partial x_*} & x_*\displaystyle\frac{\partial f}{\partial y_*} \vspace{2ex}\\
			y_*\displaystyle\frac{\partial g}{\partial x_*} & y_*\displaystyle\frac{\partial g}{\partial y_*}
		\end{bmatrix}\triangleq\begin{bmatrix}
			B_1& B_2\\
			B_3&B_4
		\end{bmatrix},
		\label{10}
	\end{equation}
	where
	$$
	B_1=\frac{x_* \left(-2 x_* +p +1\right)}{q y_* +1},
	$$
	$$
	B_2=x_* \left(-\frac{\left(1-x_* \right) \left(x_* -p \right) q}{\left(q y_* +1\right)^{2}}-a \right)<0,
	$$
	$$
	B_3=-bcy_*<0,
	$$
	$$
	B_4=-by_*<0.
	$$
	First we consider $\mathrm{det}(J(E_*)) $. From system \eqref{3}, we have
	$$
	y_*^{(f)}=\displaystyle\frac{\sqrt{4 \left(1-x_* \right) \left(x_* -p \right) q +a}-\sqrt{a}}{2 \sqrt{a}\, q},
	$$
	$$
	y_*^{(g)}=1-cx_*.
	$$
	Using the implicit function differentiability theorem, we get the following expression
	\begin{equation}
		\mathrm{det}(J(E_*)) =\left [ x_*y_*\frac{\partial f}{\partial y_*} \frac{\partial g}{\partial y_*}\left ( \frac{\mathrm{d} y_*^{(g)}}{\mathrm{d} x_*}-\frac{\mathrm{d} y_*^{(f)}}{\mathrm{d} x_*} \right ) \right ],
		\label{11}
	\end{equation}
	where $\displaystyle\frac{\mathrm{d} y_*^{(g)}}{\mathrm{d} x_*}$ and $\displaystyle\frac{\mathrm{d} y_*^{(f)}}{\mathrm{d} x_*}$ stand for the magnitude of the tangent slope at $E_*$ of two isoclines. In the following, we investigate the monotonicity of the function $y_*^{(f)}$. First, we can find
	$$
	\displaystyle\frac{\mathrm{d} y_*^{(f)}}{\mathrm{d} x_*}=\displaystyle\frac{p -2 x +1}{\sqrt{4 \left(1-x \right) \left(x -p \right) q +a}\, \sqrt{a}}.
	$$
	When the two isoclines have two intersections in the first quadrant, the function $y_*^{(f)}$ is monotonically increasing on the interval $(p,\displaystyle\frac{p+1}{2})$ and monotonically decreasing on the interval $(\displaystyle\frac{p+1}{2},1)$. And because when $x_* \in (p,1)$, $y_*^{(f)}$ is a continuous function, so combined with the image analysis(Fig. 3), we can get
	$$
	\left. \displaystyle\frac{\mathrm{d} y_*^{(g)}}{\mathrm{d} x_*}-\displaystyle\frac{\mathrm{d} y_*^{(f)}}{\mathrm{d} x_*}\right|_{(x_*,y_*)=(x_1,y_1)}<0,
	$$
	$$
	\left.\displaystyle\frac{\mathrm{d} y_*^{(g)}}{\mathrm{d} x_*}-\displaystyle\frac{\mathrm{d} y_*^{(f)}}{\mathrm{d} x_*}\right|_{(x_*,y_*)=(x_2,y_2)}>0.
	$$
	So we can determine that the positive equilibrium point $E_{1*}$ is a saddle. For $E_{2*}$, we have
	$$
	\begin{aligned}
		\mathrm{det}(J(E_{2*})) &=\left.\left [ x_*y_*\frac{\partial f}{\partial y_*} \frac{\partial g}{\partial y_*}\left ( \frac{\mathrm{d} y_*^{(g)}}{\mathrm{d} x_*}-\frac{\mathrm{d} y_*^{(f)}}{\mathrm{d} x_*} \right ) \right ]\right|_{(x_*,y_*)=(x_2,y_2)}\\
		&=\left.B_1B_4-B_2B_3\right|_{(x_*,y_*)=(x_2,y_2)}\\
		&>0.
	\end{aligned}
	$$
	Then we can conclude that $B_1B_4>B_2B_3$. Since the signs of $B_2$, $B_3$, and $B_4$ have been determined, we can thus know that $B_1<0$. Finally, we can determine that $\mathrm{tr}(J(E_{2*}))=B_1+B_4<0$ by the above analysis. From $\mathrm{det}(J(E_{2*}))>0$, $\mathrm{tr}(J(E_{2*}))<0$, we know that $E_{2*}$ is a stable node.
	
	According to Theorem 3.1, when the discriminant $\Delta = 0$ of \eqref{5}, the positive equilibria $E_{1*}$ and $E_{2*}$ will merge into a new point $E_{3*}$. In the following, we discuss the stability of $E_{3*}$.
	
	If $\Delta = 0$, we get
	$$
	q = \displaystyle\frac{a^{2} c^{2}+2 a c p +2 a c +p^{2}-4 a -2 p +1}{4 a \left(c^{2} p -c p -c +1\right)}\triangleq q_*.
	$$
	Substituting $q=q_*$, $x=x_{3}$, $y=1-cx_{3}$ into \eqref{6}, we find that $\mathrm{det}(J(E_{3*}))=0$. Thus the positive equilibrium point $E_{3*}$ is a degenerate equilibrium point, which we will analyze in more detail in the next step.
	
	We move equilibrium $E_{3*}$ to the origin by transforming $(X, Y)=(x-x_3, y-y_3)$ and make Taylor's expansion around the origin, then system \eqref{2} becomes
	\begin{equation}
		\left\{\begin{array}{l}
			\displaystyle\frac{\mathrm{d} X}{\mathrm{d} t} =e_{10}X+e_{01}Y+e_{11}XY+e_{20}X^2+e_{02}Y^2+e_{30}X^3+e_{03}Y^3+e_{21}X^2Y+e_{12}XY^2+P_4(X,Y), \vspace{2ex}\\
			\displaystyle\frac{\mathrm{d} Y}{\mathrm{d} t} =f_{10}X+f_{01}Y+f_{11}XY+f_{02}Y^2,
		\end{array}\right.
		\label{12}
	\end{equation}
	where $f_{11}=-bc$, $f_{02}=-b$, please see Appendix A for the rest of the parameters. We make the following transformations to system \eqref{12}
	$$
	\begin{bmatrix}
		X \\Y
		
	\end{bmatrix}=\begin{bmatrix}
		\displaystyle\frac{\left(c a +2 c p -p -1\right) a \left(c a p +c a +p^{2}-2 a -2 p +1\right)}{\left(a^{2} c^{2}-p^{2}+2 p -1\right) b \left(c^{2} p -c p -c +1\right)} &\displaystyle -\frac{1}{c} \\
		1 &1
	\end{bmatrix}\begin{bmatrix}
		X_1 \\Y_1
		
	\end{bmatrix},
	$$
	and letting $\tau=m_1 t$, where
	$$
	\begin{aligned}
		m_1= & \frac{2 a^{2} b \,c^{4} p +2 a^{2} \left(p +1\right) \left(2 p +a -b \right) c^{3}+\left(\left(4 a -2 b \right) p^{3}+\left(-8 a +4 b \right) p^{2}+\left(-16 a^{2}+4 a -2 b \right) p -4 a^{3}+2 b \,a^{2}\right) c^{2}}{\left(-2+ac^{2}+\left(p +1\right) c \right) \left(c a +p -1\right) \left(c a -p +1\right)}\\
		& +\frac{4 \left(p +1\right) \left(\left(-\frac{a}{2}+\frac{b}{2}\right) p^{2}+\left(a -b \right) p +a^{2}-\frac{a}{2}+\frac{b}{2}\right) c -2 b \left(p -1\right)^{2}}{\left(-2+ac^{2}+\left(p +1\right) c \right) \left(c a +p -1\right) \left(c a -p +1\right)},
	\end{aligned}
	$$
	for which we will retain $t$ to denote $\tau$ for notational simplicity. We get
	\begin{equation}
		\left\{\begin{array}{l}
			\displaystyle\frac{\mathrm{d} X_1}{\mathrm{d} t} =X_1+0\cdot Y_1+\cdots \cdots , \vspace{2ex}\\
			\displaystyle\frac{\mathrm{d} Y_1}{\mathrm{d} t} =0\cdot X_1+0\cdot Y_1+g_{02}Y_1^2+\cdots \cdots,
		\end{array}\right.
		\label{13}
	\end{equation}
	where
	$$
	g_{02}=\frac{\left(c a +2 c p -p -1\right) \left(a \,c^{2}+c p +c -2\right)^{3} a b \left(a^{2} c^{2}-p^{2}+2 p -1\right)}{4 \left(a^{2} b \,c^{4} p +a^{2} \left(p +1\right) \left(2 p +a -b \right) c^{3}+\left(\left(2 a -b \right) p^{3}+\left(-4 a +2 b \right) p^{2}+\left(-8 a^{2}+2 a -b \right) p -2 a^{3}+b \,a^{2}\right) c^{2}+N\right)^{2} c},
	$$
	$$
	N=2 \left(p +1\right) \left(\left(-\frac{a}{2}+\frac{b}{2}\right) p^{2}+\left(a -b \right) p +a^{2}-\frac{a}{2}+\frac{b}{2}\right) c -b \left(p -1\right)^{2}.
	$$
	According to Theorem 3.1, system \eqref{3} satisfies the conditions $\Delta=0$, $2A_1+A_2>0$ when the positive equilibrium point $E_3*$ exists. We substitute $q=q*$ for $2A_1+A_2>0$ and simplify to get $\displaystyle\frac{\left(a c -p +1\right) \left(a c^{2}+c p +c -2\right)}{2 c p -2}>0$. With $0<c<1$, we can determine that $a c^{2}+c p +c -2< 0$. \
	
	If $a^{2} c^{2}-p^{2}+2 p -1=0$, we can organize to get $a=\displaystyle\frac{1-p}{c}\triangleq a_*$. Substitute $a_*$ into $2A_1+A_2$ to simplify and get $2A_1+A_2=-2(p - 1)q(c - 1)\le 0$, i.e., $a^{2} c^{2}-p^{2}+2 p -1\ne 0$.
	
	\begin{figure}[]
		\centering
		\vspace{-0cm}
		\setlength{\abovecaptionskip}{-0pt}
		\subfigcapskip=-30pt
		\subfigure[]
		{\scalebox{0.45}[0.45]{
				\includegraphics{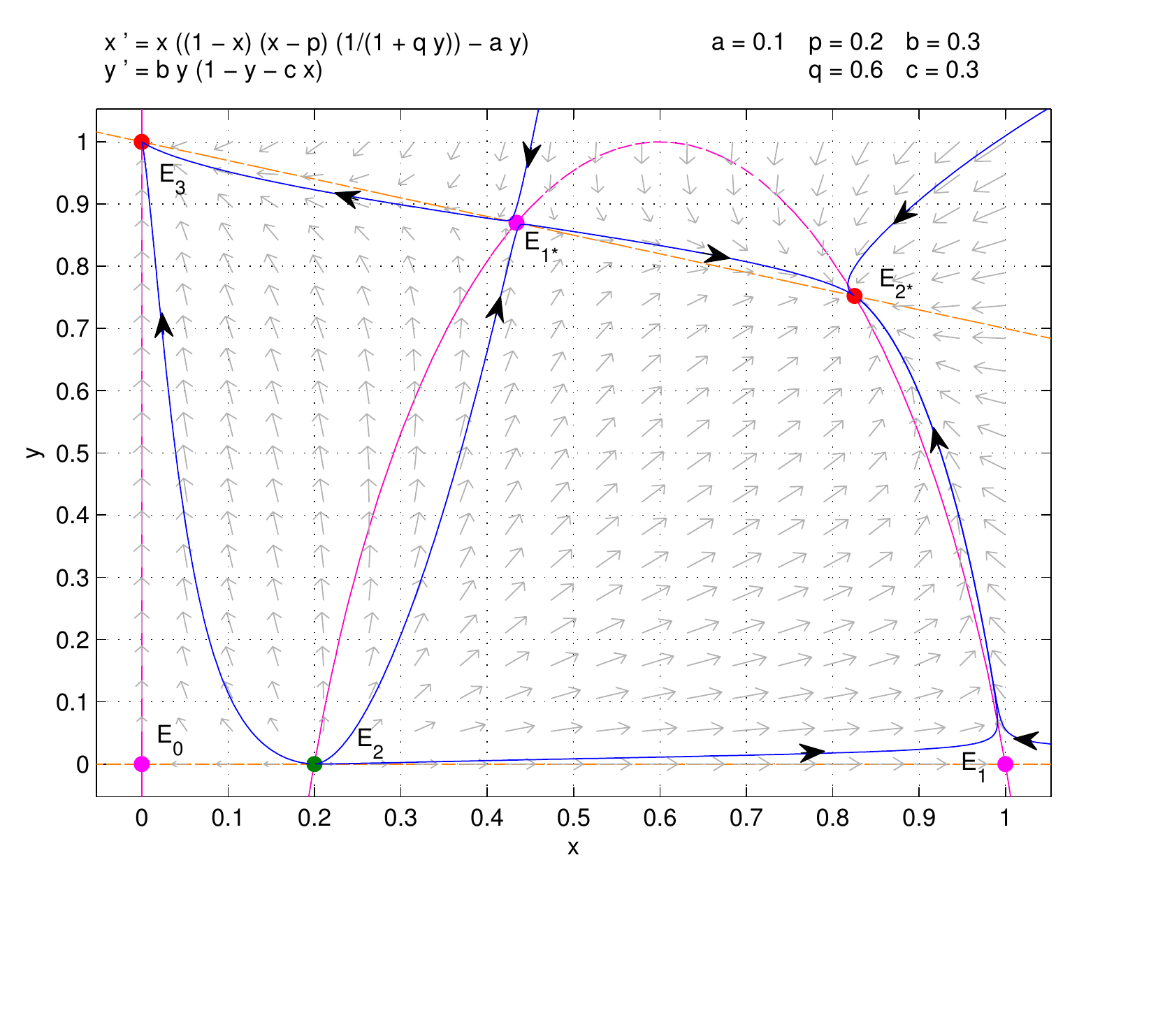}}}
		\subfigure[]
		{\scalebox{0.45}[0.45]{
				\includegraphics{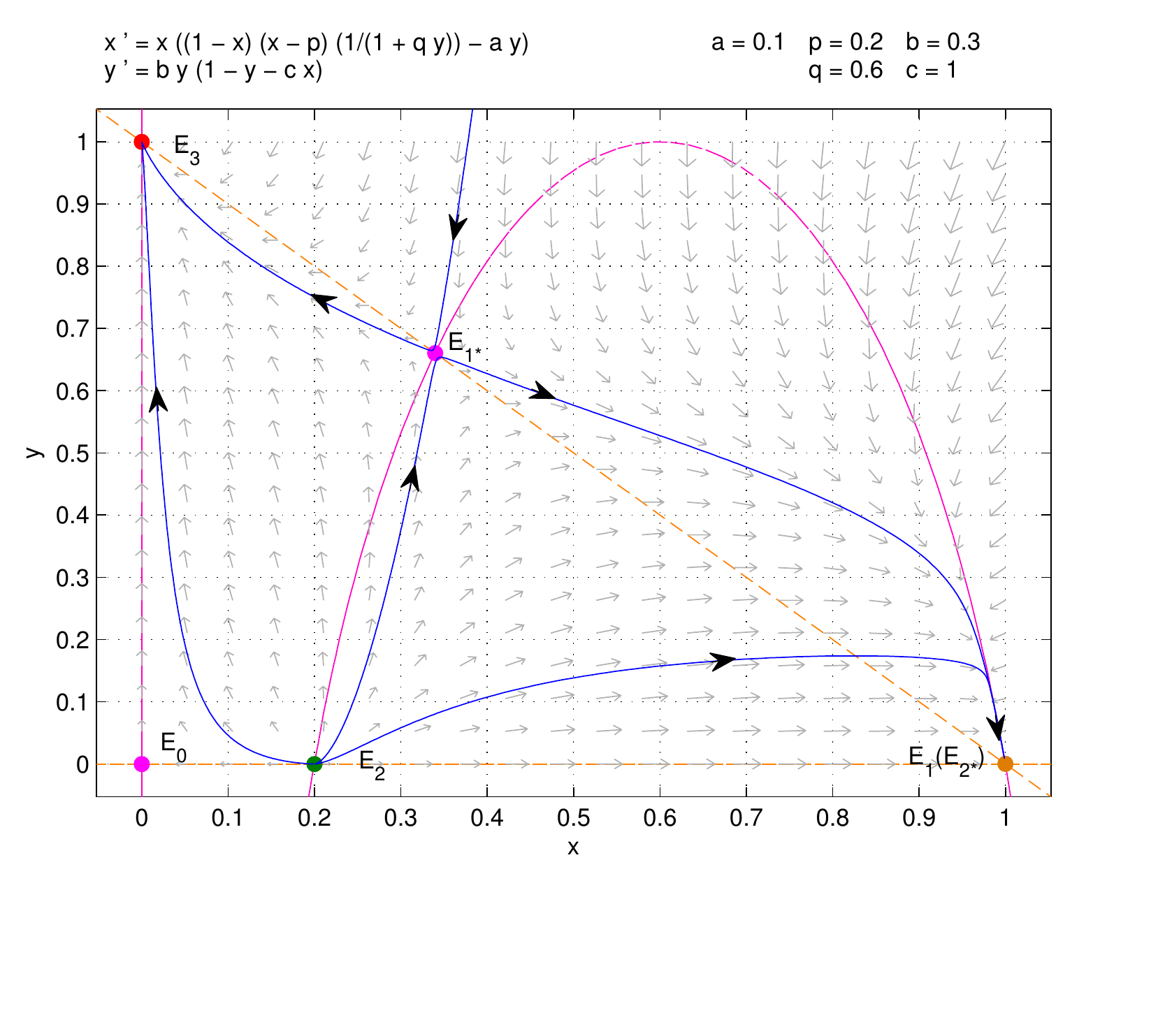}}}
		\subfigure[]
		{\scalebox{0.45}[0.45]{
				\includegraphics{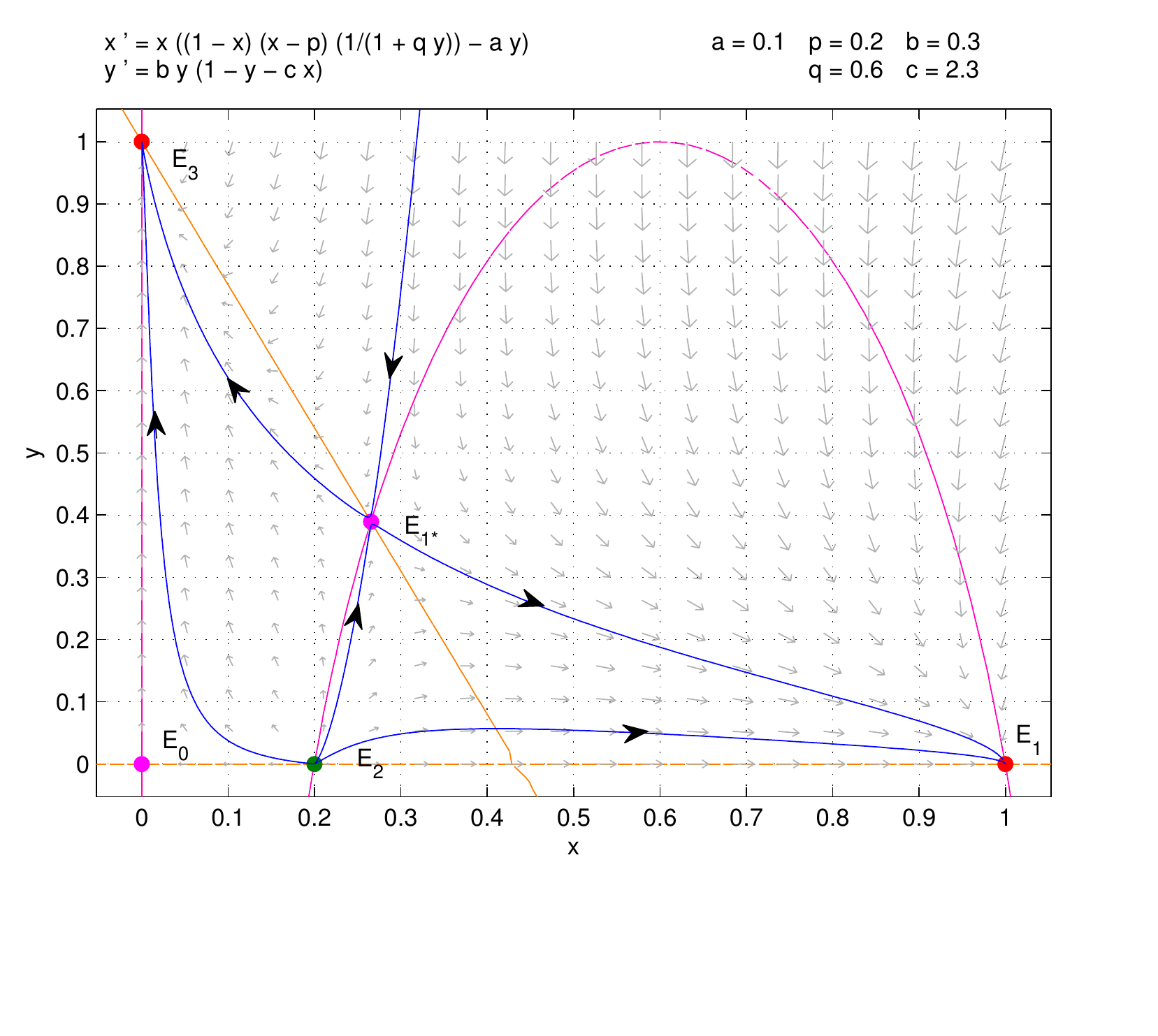}}}
		\subfigure[]
		{\scalebox{0.45}[0.45]{
				\includegraphics{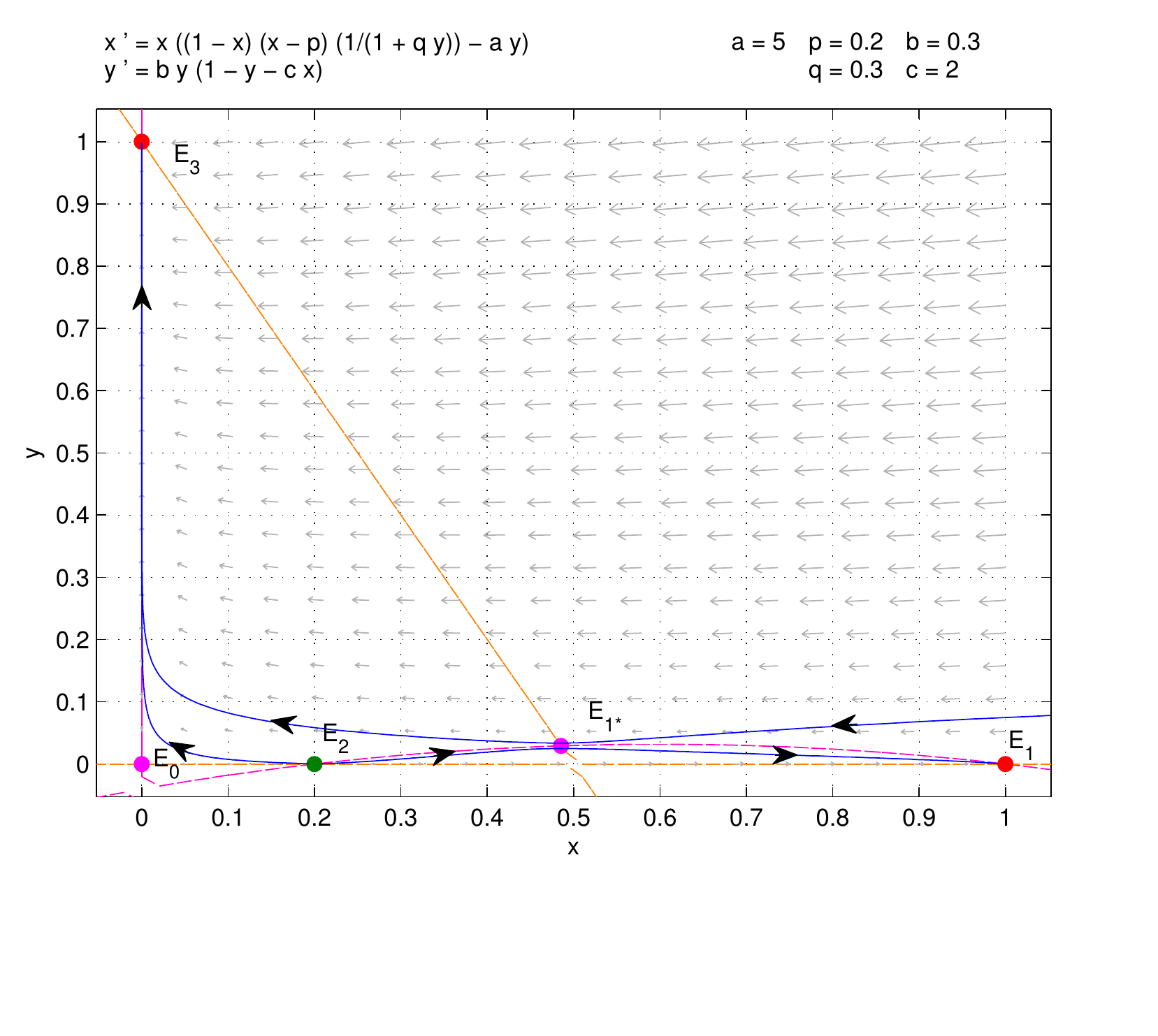}}}
		\subfigure[]
		{\scalebox{0.45}[0.45]{
				\includegraphics{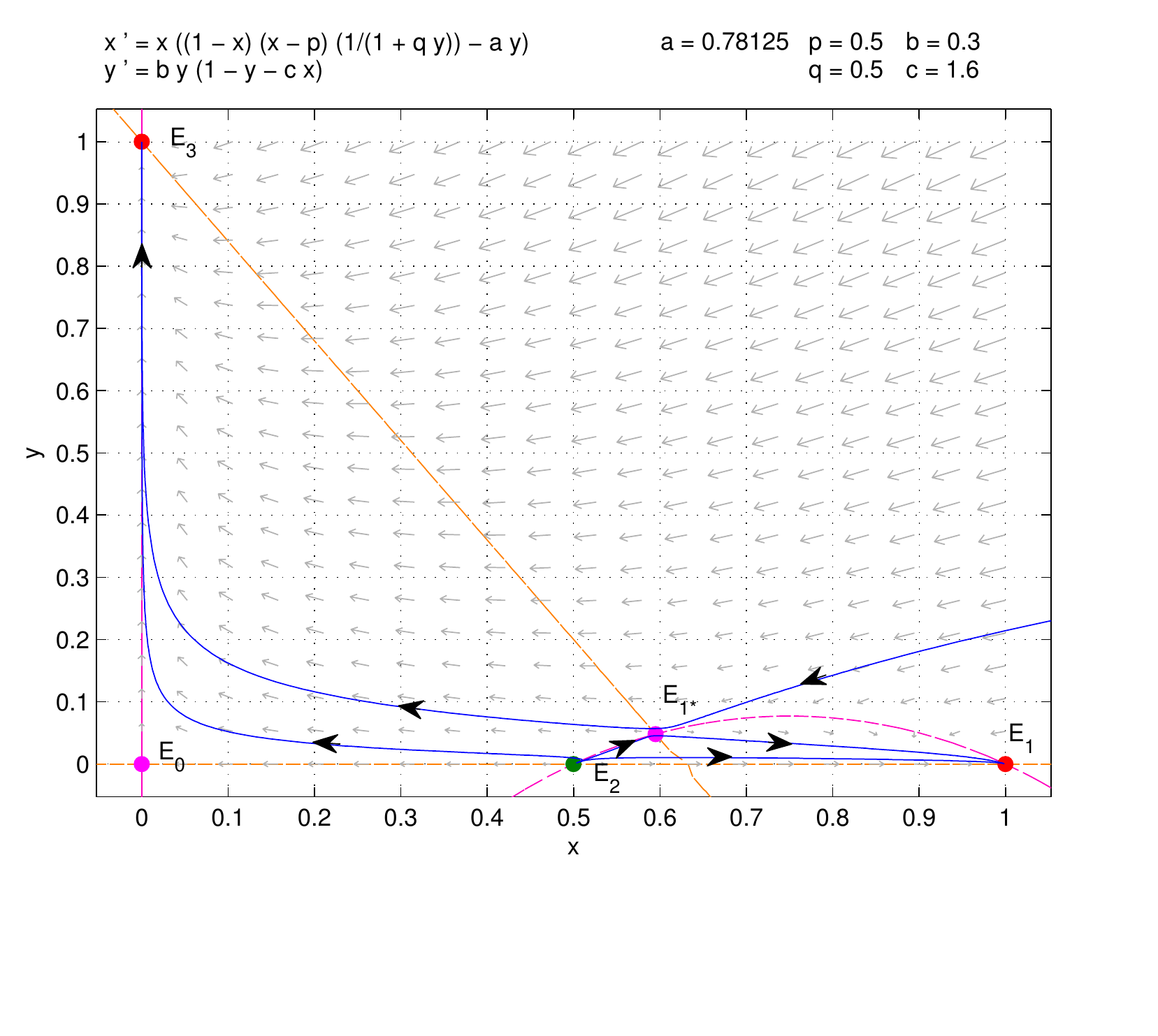}}}
		\subfigure[]
		{\scalebox{0.45}[0.45]{
				\includegraphics{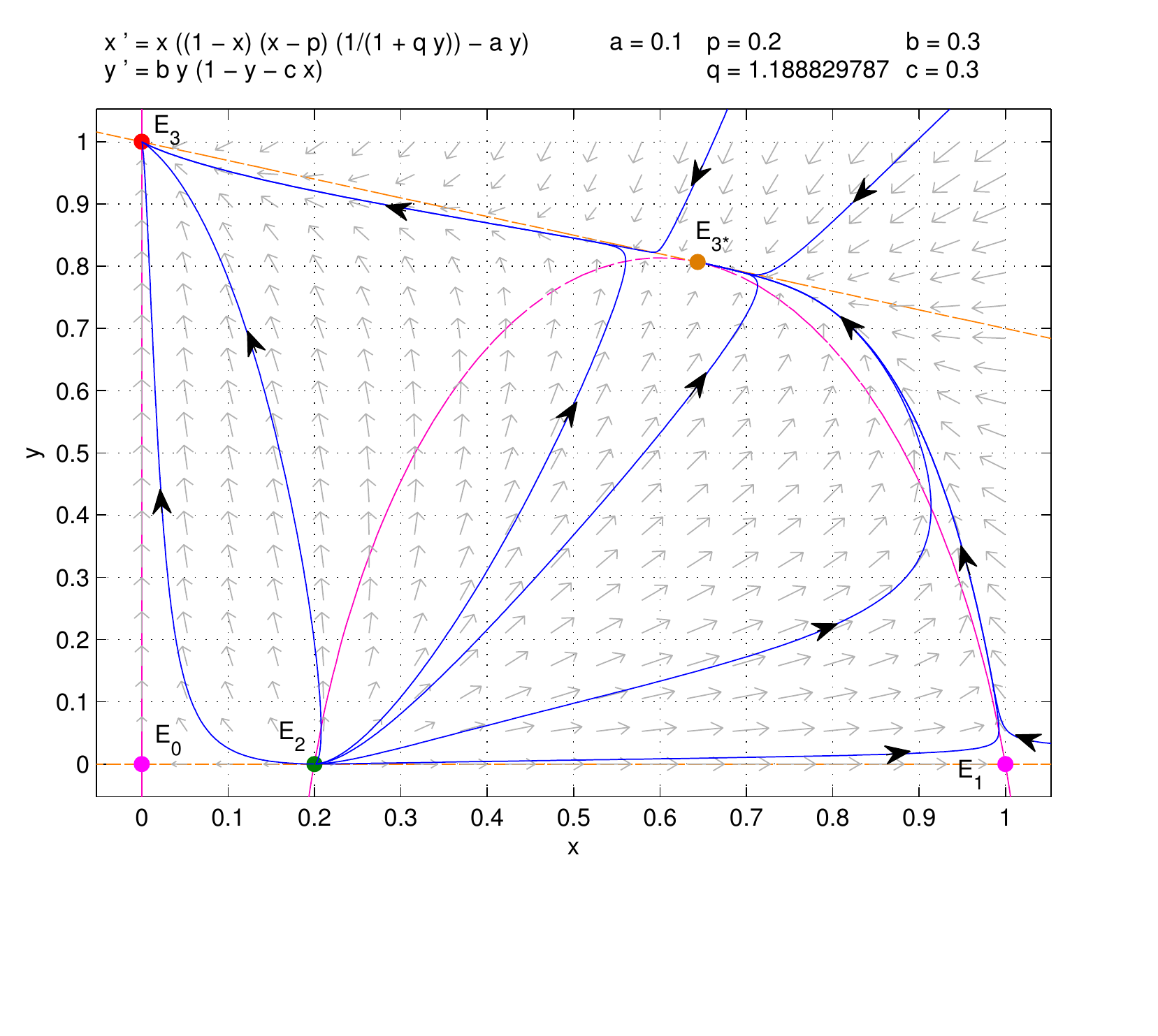}}}
		\caption{Red, green, pink, and orange points indicate stable node, unstable node (source), saddle, and saddle-node, respectively. The pink and orange lines represent isoclines $y^{(f)}$, $y^{(g)}$, respectively.}
		\label{Fig3}
	\end{figure}
	
	For $ca+2cp-p-1$, if $a = \displaystyle\frac{-2 c p +p +1}{c}\triangleq a_{**}$ is true, then substitute $q=q_*$ and $a=a_{**}$ into $A_1$, $A_2$, and $A_3$, respectively. It is tested that $A_1=(c - 1)(cp - 1)$, $A_2=0$, and $A_3=0$, which do not comply with the parameter assumptions in the previous section. From this, we know that $ca+2cp-p-1\ne 0$.
	
	According to the above analysis, we can determine the parameter $g02 \ne 0$. Hence by Theorem 7.1 in Chapter 2 in \cite{22}, $E_{3*}$ is a saddle-node. In summary, we derive the following theorem.
	\begin{theorem}\label{thm:stab_pos}
		The stability of the positive equilibria is shown below:
		\begin{enumerate}
			\item $E_{1*}$ is a saddle.
			\item $E_{2*}$ is a stable node.
			\item $E_{3*}$ is a saddle-node.
		\end{enumerate}
	\end{theorem}
	
	\section{Bifurcation Analysis}
	\subsection{Transcritical bifurcation}
	In proving Theorem 3.1 and 4.1, we found an interesting phenomenon: when $c=1$, the positive equilibrium point $E_{2*}$ will merge with the boundary equilibrium point $E_1$. Also, the stability of the boundary equilibrium point $E_1$ will change when the parameter $c$ is in different intervals $(0,1)$ and $(1,+\infty )$, respectively. From this, we conjecture that system \eqref{3} experiences a transcritical bifurcation around $E_1$ while noting $c$ as a bifurcation parameter(Fig. 4). We proceed to a rigorous proof below.
	
	\begin{theorem}
		System \eqref{3} undergoes a transcritical bifurcation around $E_1$ at the bifurcation parameter threshold $c_{TR} = 1$ when $p\ne 1-a$.
	\end{theorem}
	\begin{proof}
		From Theorem 4.1, we know that the eigenvalues of $J(E_1)$ are $\lambda_{11}=-1+p$, $\lambda_{21}=0$ if $c=c_{TR}=1$. Now, let $\mathbf{V_1} = (v_1, v_2)^T$ and $\mathbf{W_1} = (w_1, w_2)^T$ be the eigenvectors of $J(E_1)$ and $J^T(E_1)$ corresponding to $\lambda_{21}=0$, respectively. By calculating, we obtain
		\begin{equation}
			\mathbf{V_1}=\begin{bmatrix}
				v_1\\v_2
				
			\end{bmatrix} =\begin{bmatrix}
				\displaystyle\frac{a}{p-1} \\1
				
			\end{bmatrix},\mathbf{W_1}=\begin{bmatrix}
				w_1\\w_2
				
			\end{bmatrix} =\begin{bmatrix}
				0 \\1
			\end{bmatrix}.
			\label{14}
		\end{equation}
		We assume that
		$$
		Q(x,y)=\begin{bmatrix}
			F(x,y) \\G(x,y)
			
		\end{bmatrix}=\begin{bmatrix}
			x \left [(1-x)(x-p)\displaystyle\frac{1}{1+qy}-ay \right ]\vspace{2ex}\\by\left (1-y-cx \right )
		\end{bmatrix}.
		$$
		Furthermore,
		$$
		Q_c(E_1;c_{TR})=\begin{bmatrix}
			\displaystyle\frac{\partial F}{\partial c} \vspace{2ex}\\\displaystyle\frac{\partial G}{\partial c}
			
		\end{bmatrix}=\begin{bmatrix}
			0\\0
			
		\end{bmatrix},
		$$
		\begin{figure}[H]
			\centering
			\vspace{-0cm}
			\setlength{\abovecaptionskip}{-0pt}
			\subfigcapskip=-30pt
			\subfigure[]
			{\scalebox{0.45}[0.45]{
					\includegraphics{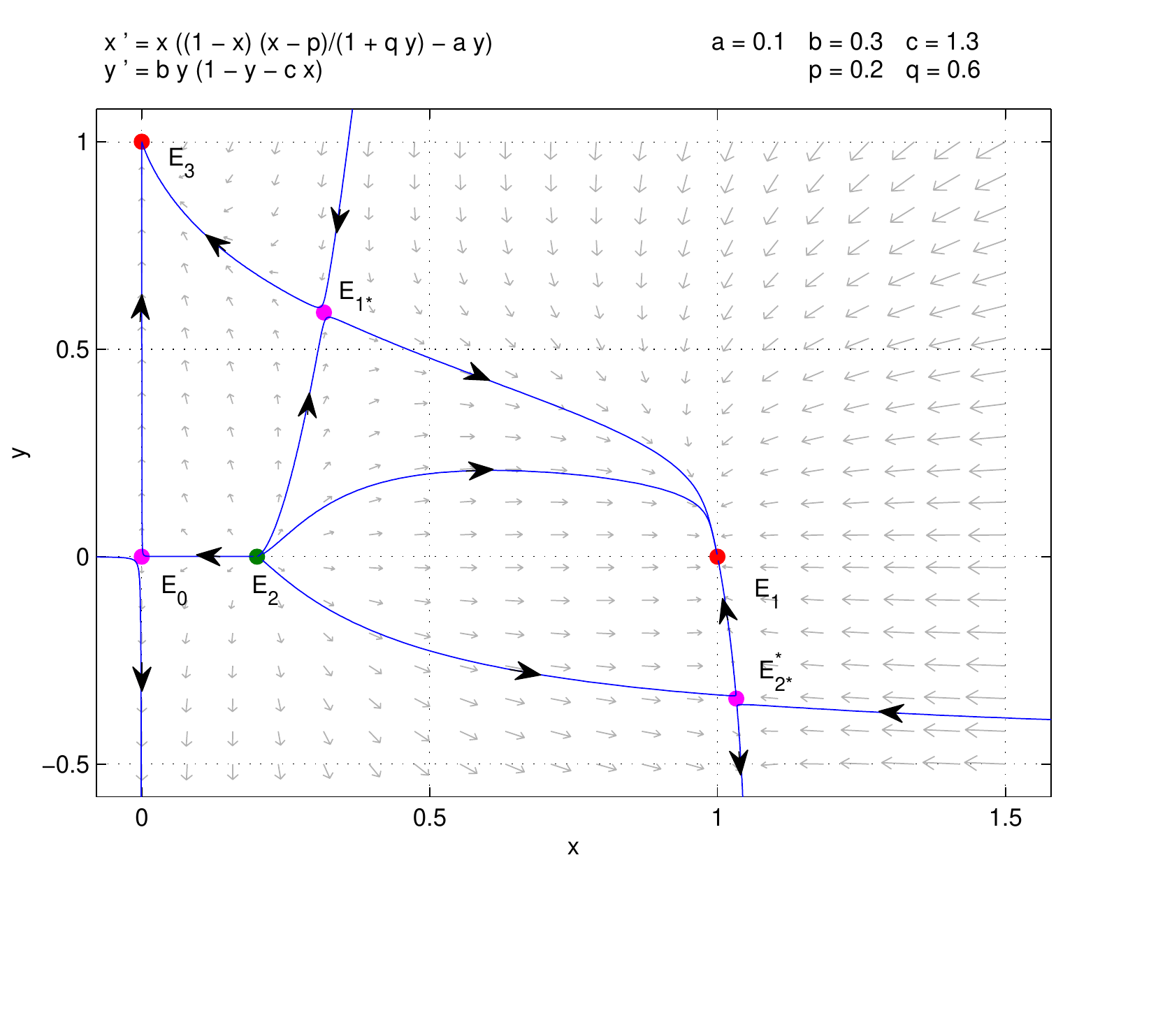}}}
			\subfigure[]
			{\scalebox{0.45}[0.45]{
					\includegraphics{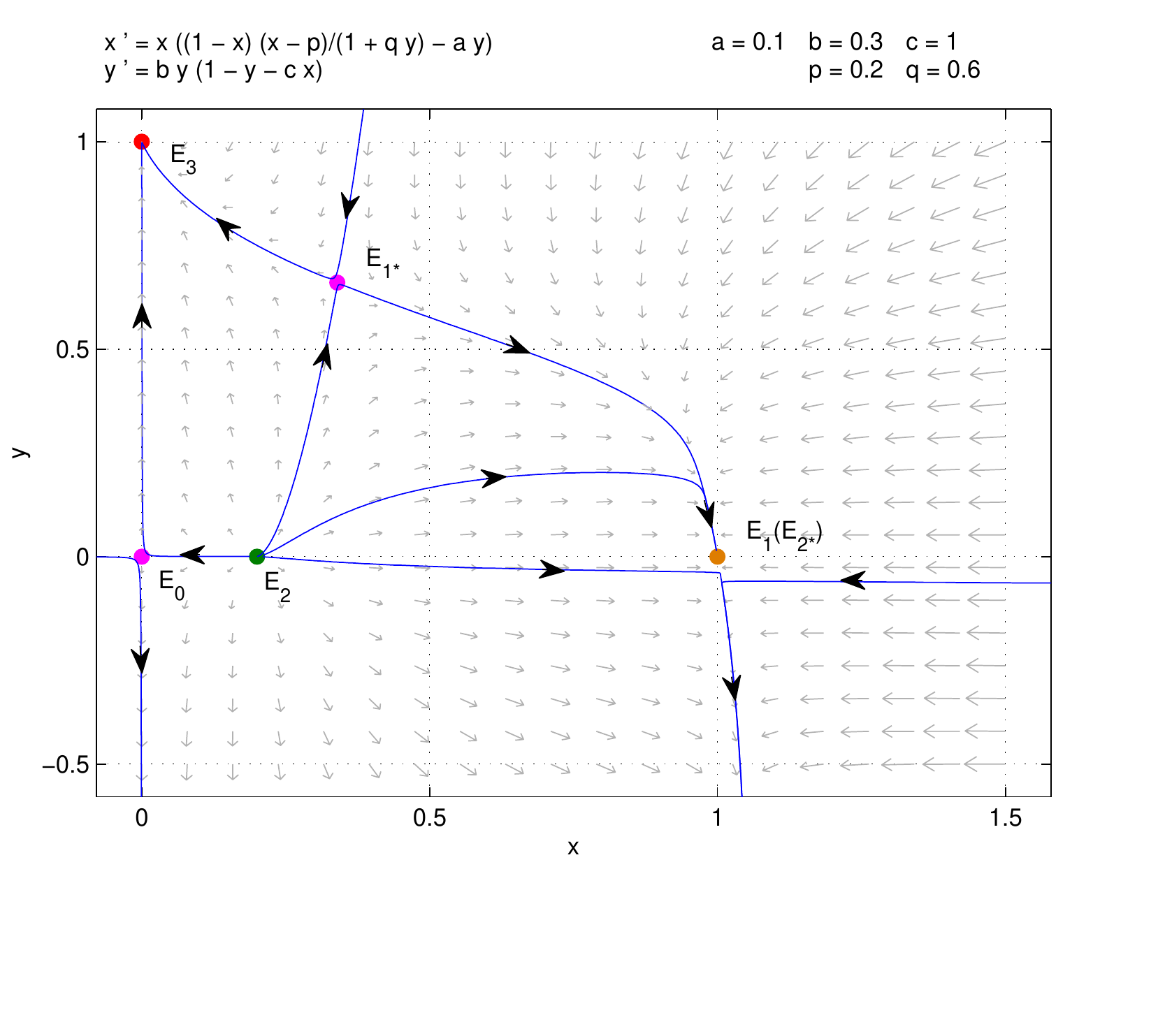}}}
			\subfigure[]
			{\scalebox{0.45}[0.45]{
					\includegraphics{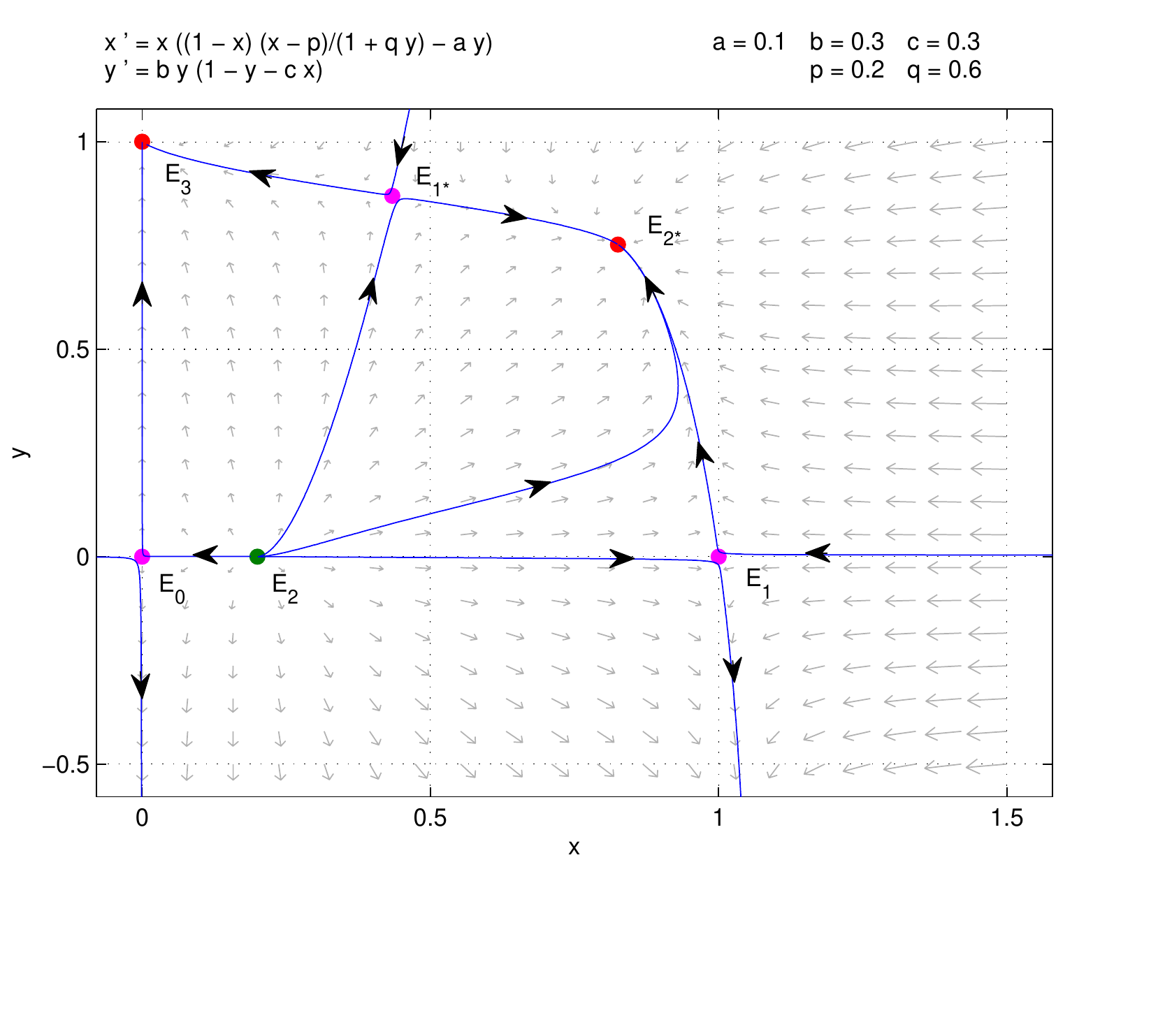}}}
			\caption{Red, green, pink, and orange points indicate stable node, unstable node (source), saddle, and saddle-node, respectively. System \eqref{3} undergoes a transcritical bifurcation around $E_1$.}
			\label{Fig4}
		\end{figure}
		$$
		\left. DQ_c(E_1;c_{TR})\mathbf{V_1}=\left[\begin{array}{cc}
			0 & 0
			\\
			-b y & -b x
		\end{array}\right] \right |_{(E_1;c_{TR})}\begin{bmatrix}
			\displaystyle\frac{a}{p-1} \\1
			
		\end{bmatrix}=\begin{bmatrix}
			0 \\-b
		\end{bmatrix},
		$$
		$$
		\left.D^2Q(E_1;c_{TR})(\mathbf{V_1}, \mathbf{V_1})=\begin{bmatrix}
			\displaystyle\frac{\partial^2F}{\partial x^2}v_1^2+ 2\displaystyle\frac{\partial^2F}{\partial x\partial y}v_1v_2+ \displaystyle\frac{\partial^2F}{\partial y^2}v^2_2\vspace{2ex}\\
			\displaystyle\frac{\partial^2G}{\partial x^2}v_1^2+ 2\displaystyle\frac{\partial^2G}{\partial x\partial y}v_1v_2+ \displaystyle\frac{\partial^2G}{\partial y^2}v^2_2
		\end{bmatrix}\right|_{(E_1;c_{TR})}=\begin{bmatrix}
			-\displaystyle \frac{2 \left(\left(p -1\right)^{2} q +a \right) a}{\left(p -1\right)^{2}} \vspace{2ex}\\\displaystyle\frac{2 b \left(a +p -1\right)}{1-p}
		\end{bmatrix}.
		$$
		Thus, we have
		$$
		\mathbf{W_1}^TQ_c(E_1;c_{TR})=0,
		$$
		$$
		\mathbf{W_1}^T\left[ DQ_c(E_1;c_{TR})\mathbf{V_1} \right]=-b\ne0,
		$$
		$$
		\mathbf{W_1}^T\left[ D^2Q(E_1;c_{TR})(\mathbf{V_1}, \mathbf{V_1}) \right]=\displaystyle\frac{2 b \left(a +p -1\right)}{1-p}\ne0.
		$$
		According to {\it Sotomayor's Theorem} \cite{23}, all the transversality conditions for system \eqref{3} to experience a transcritical bifurcation are satisfied, so system \eqref{3} undergoes a transcritical bifurcation around $E_1$ at the bifurcation parameter threshold $c_{TR} = 1$.
	\end{proof}
	
	\begin{figure}[H]
		\centering
		\vspace{-0cm}
		\setlength{\abovecaptionskip}{-0pt}
		\subfigcapskip=-30pt
		\subfigure[]
		{\scalebox{0.45}[0.45]{
				\includegraphics{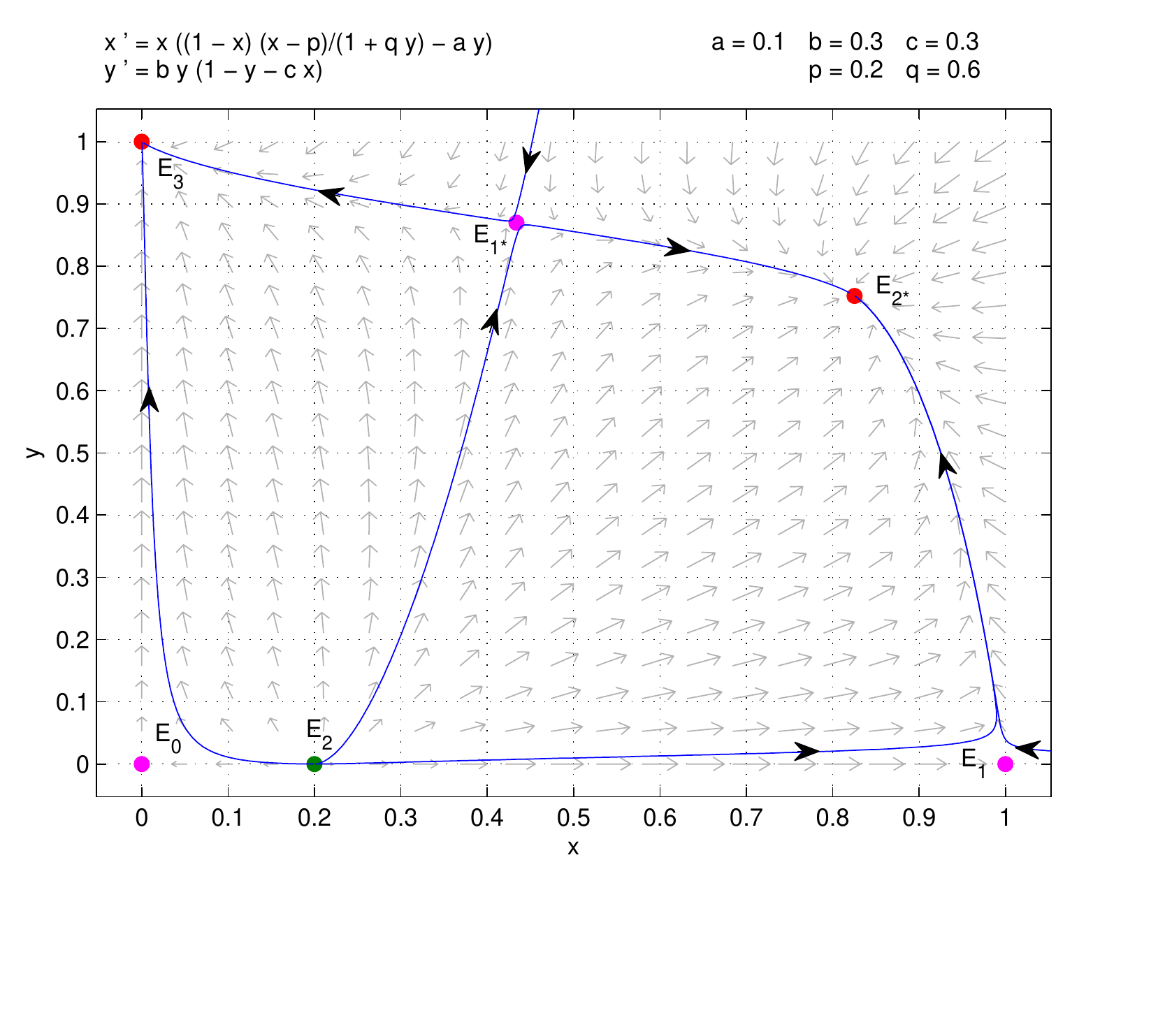}}}
		\subfigure[]
		{\scalebox{0.45}[0.45]{
				\includegraphics{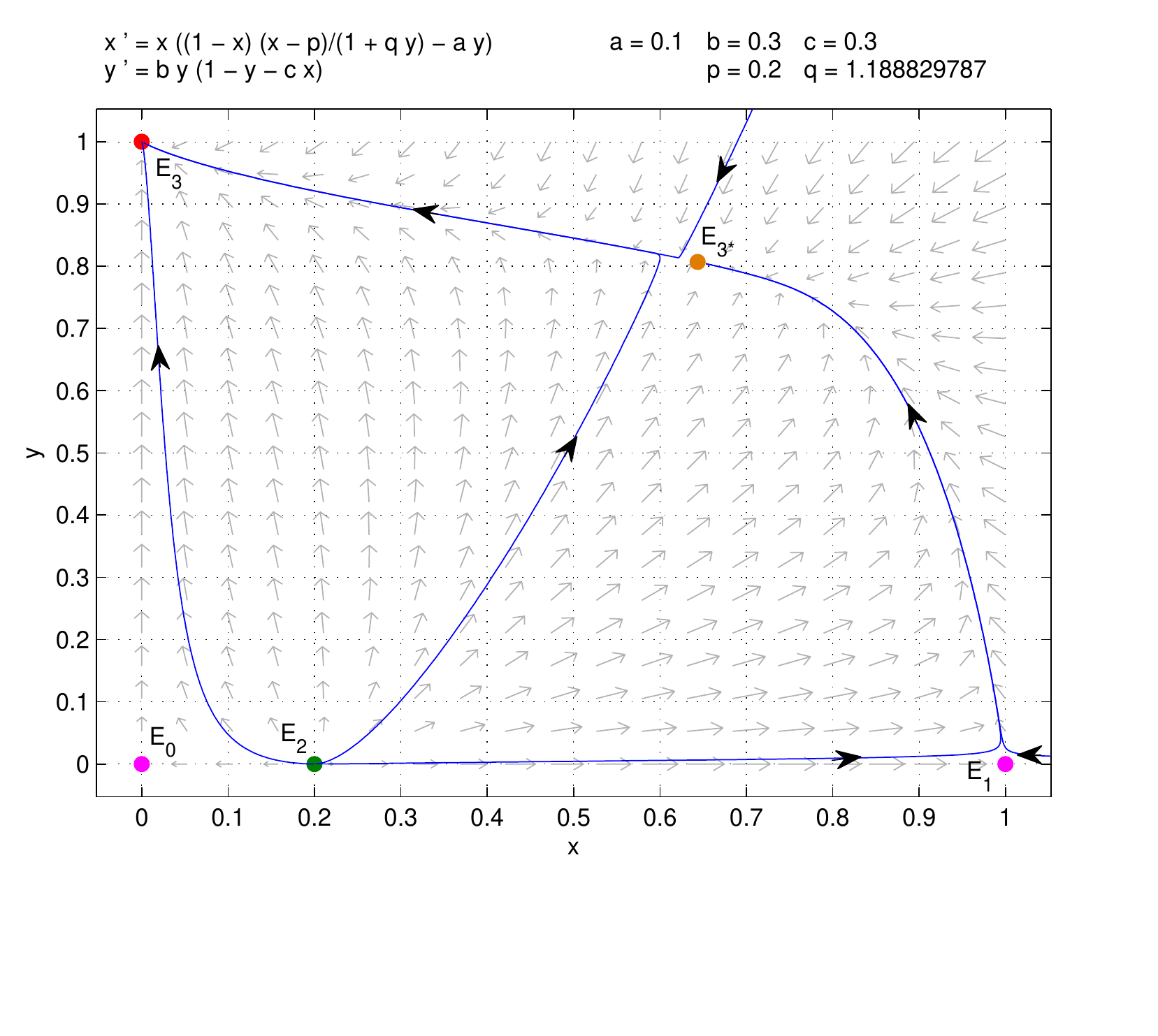}}}
		\subfigure[]
		{\scalebox{0.45}[0.45]{
				\includegraphics{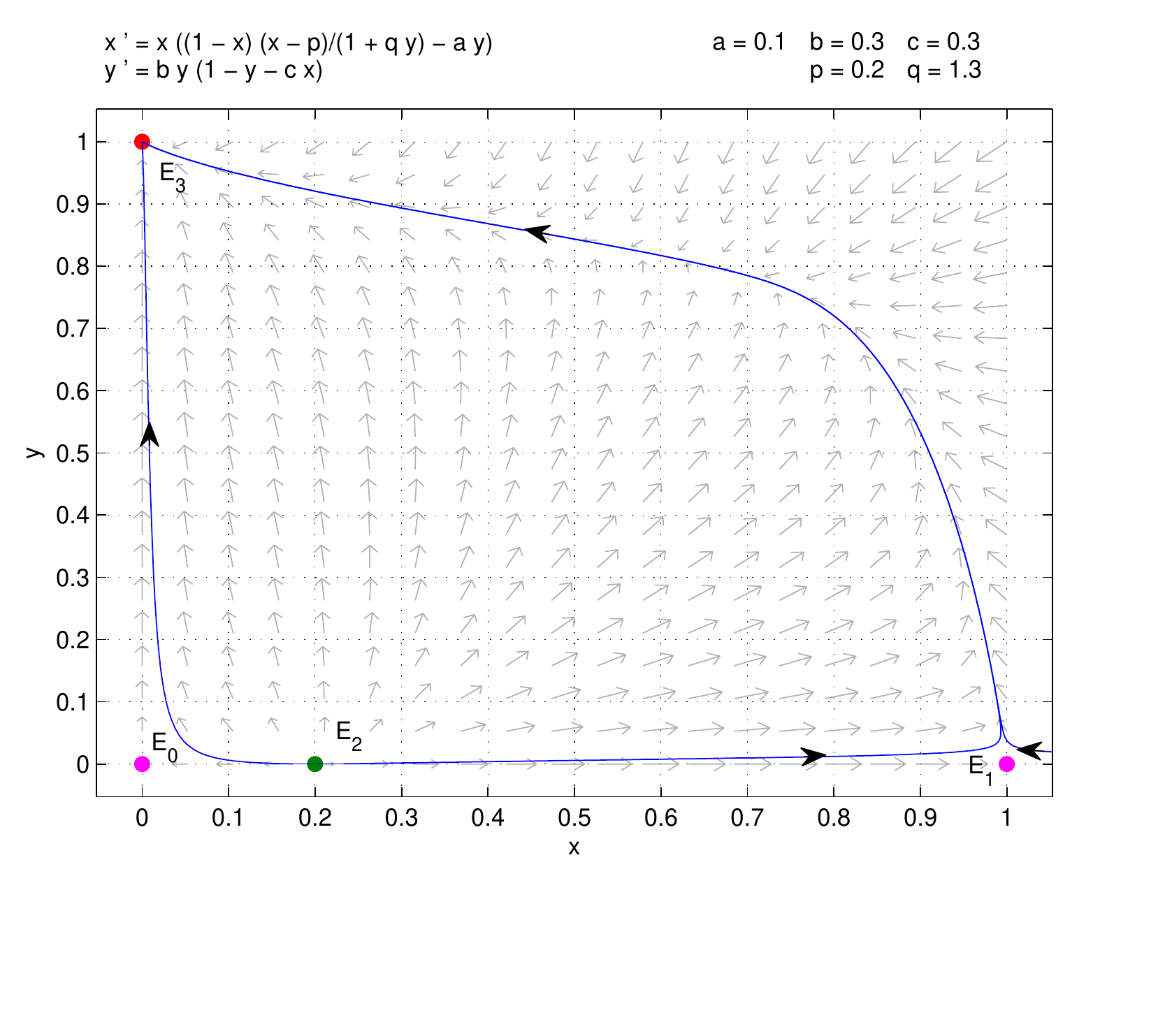}}}
		\caption{Red, green, pink, and orange points indicate stable node, unstable node (source), saddle, and saddle-node, respectively. System \eqref{3} undergoes a saddle-node bifurcation around $E_{3*}$.}
		\label{Fig4}
	\end{figure}
	
	\subsection{Saddle-node bifurcation}
	Under the condition $2A_1+A_2>0$ and $0<c<1$, we note that when $\Delta<0$, $\Delta=0$, and $\Delta>0$, system \eqref{3} has 0, 1, and 2 positive equilibria, respectively. Therefore we consider system \eqref{3} undergoing a saddle-node bifurcation around the positive equilibrium point $E_{3*}$. We selected the fear effect parameter $q$ as the bifurcation parameter. By calculating $\Delta=0$, we obtain the bifurcation parameter threshold $q=\displaystyle\frac{a^{2} c^{2}+2 a c p +2 a c +p^{2}-4 a -2 p +1}{4 a \left(c^{2} p -c p -c +1\right)}\triangleq q_*$. Next, we use {\it Sotomayor's Theorem} \cite{23} to verify that the transversality conditions for saddle-node bifurcation are satisfied.
	\begin{theorem}
		System \eqref{3} undergoes a saddle-node bifurcation around $E_{3*}$ at the bifurcation parameter threshold $q_{SN} = q_*$ when $2A_1+A_2>0$ and $0<c<1$.
	\end{theorem}
	\begin{proof}
		By \eqref{10}, we know that the Jacobi matrix of the positive equilibrium point $E_{3*}$ can be expressed in the following form, and one of the eigenvalues is $\lambda = 0$.
		$$
		\left.J(E_{3*})=\left[\begin{array}{cc}
			x \left(-\displaystyle\frac{x -p}{q y +1}+\frac{1-x}{q y +1}\right) & x \left(-\displaystyle\frac{\left(1-x \right) \left(x -p \right) q}{\left(q y +1\right)^{2}}-a \right)
			\vspace{2ex}\\
			-b y c & -b y
		\end{array}\right]\right|_{(E_{3*};q_{SN})}
		$$
		Now, let $\mathbf{V_2} = (v_3, v_4)^T$ and $\mathbf{W_2} = (w_3, w_4)^T$ be the eigenvectors of $J(E_{3*})$ and $J^T(E_{3*})$ corresponding to $\lambda=0$, respectively. By calculating, we obtain
		\begin{equation}
			\mathbf{V_2}=\begin{bmatrix}
				v_3\\v_4
				
			\end{bmatrix} =\begin{bmatrix}
				1 \\-c
				
			\end{bmatrix},\mathbf{W_2}=\begin{bmatrix}
				w_3\\w_4
				
			\end{bmatrix} =\begin{bmatrix}
				1\vspace{2ex} \\\displaystyle\frac{x \left(-2 x +p +1\right)}{\left(c q x -q -1\right) b c \left(c x -1\right)}
			\end{bmatrix}.
			\label{15}
		\end{equation}
		Furthermore,
		$$
		Q_q(E_{3*};q_{SN})=\begin{bmatrix}
			\displaystyle\frac{\partial F}{\partial q} \vspace{2ex}\\\displaystyle\frac{\partial G}{\partial q}
			
		\end{bmatrix}=\begin{bmatrix}
			-\displaystyle\frac{x_3 \left(1-x_3 \right) \left(x_3 -p \right) y_3}{\left(q_* y_3 +1\right)^{2}}\vspace{2ex}\\0
		\end{bmatrix},
		$$
		$$
		\left.D^2Q(E_{3*};q_{SN})(\mathbf{V_2}, \mathbf{V_2})=\begin{bmatrix}
			\displaystyle\frac{\partial^2F}{\partial x^2}v_3^2+ 2\displaystyle\frac{\partial^2F}{\partial x\partial y}v_3v_4+ \displaystyle\frac{\partial^2F}{\partial y^2}v^2_4\vspace{2ex}\\
			\displaystyle\frac{\partial^2G}{\partial x^2}v_3^2+ 2\displaystyle\frac{\partial^2G}{\partial x\partial y}v_3v_4+ \displaystyle\frac{\partial^2G}{\partial y^2}v^2_4
		\end{bmatrix}\right|_{(E_{3*};q_{SN})}=\begin{bmatrix}
			H \vspace{2ex}\\0
		\end{bmatrix},
		$$
		where
		$$
		H=\displaystyle\frac{-6 x_3 +2 p +2}{q_* y_3 +1}+\displaystyle\frac{2 c \left(a \,q_*^{2} y^{2}+\left(2 a y_3 -3 x_3^{2}+\left(2 p +2\right) x_3 -p \right) q_* +a \right)}{\left(q_* y_3 +1\right)^{2}}+\displaystyle\frac{2 c^{2} x_3 \left(1-x_3 \right) \left(x_3 -p \right) q_*^{2}}{\left(q_* y_3 +1\right)^{3}}.
		$$
		Thus, we have
		$$
		\mathbf{W_2}^TQ_q(E_{3*};q_{SN})= -\displaystyle\frac{x_3 \left(1-x_3 \right) \left(x_3 -p \right) y_3}{\left(q_* y_3 +1\right)^{2}}\ne0,
		$$
		$$
		\mathbf{W_2}^T\left[ D^2Q(E_{3*};q_{SN})(\mathbf{V_2}, \mathbf{V_2}) \right]=H\ne0.
		$$
		According to {\it Sotomayor's Theorem} \cite{23}, all the transversality conditions for system \eqref{3} to experience a saddle-node bifurcation are satisfied, so system \eqref{3} undergoes a saddle-node bifurcation around $E_{3*}$ at the bifurcation parameter threshold $q_{SN} = q_*$.
	\end{proof}
	
	\begin{remark}
		This section discusses all possible bifurcations of system \eqref{3}. The above analysis demonstrates that the value of the Allee effect parameter $p$ determines whether a transcritical bifurcation of system \eqref{3} is possible. Moreover, the variation of the fear effect parameter $q$ causes a saddle-node bifurcation of system \eqref{3}. Thus we can determine that both Allee and fear effect lead to complex dynamics in the classical Lotka-Volterra competition model.
	\end{remark}

	\section{The PDE Case}
	
	\subsubsection{Notations and preliminary observations}
	\label{notat}
	We go over several preliminaries that will enable us to prove global existence of solutions to \eqref{PDEmodel}. 
	To this end it suffices to derive uniform estimate on the $\mathbb{L}^{p}$ norms of the R.H.S. of \eqref{PDEmodel}, for some $p > \frac{n}{2}$. 
	Classical theory will then yield global existence, \cite{henry}.
	The usual norms in spaces $\mathbb{L}^{p}(\Omega )$, $\mathbb{L}^{\infty
	}(\Omega ) $ and $\mathbb{C}\left( \overline{\Omega }\right) $ are
	respectively denoted by
	
	\begin{equation}
		\label{(2.2)}
		\left\| u\right\| _{p}^{p} = 
		\int_{\Omega }\left| u(x)\right|^{p}dx, \ \left\| u\right\| _{\infty }\text{=}\underset{x\in \Omega }{ess \sup}\left|
		u(x)\right| .  
	\end{equation}
	
	To this end, we use standard techniques \cite{ Morgan89}. We first recall classical results guaranteeing non-negativity of solutions, local and global existence \cite{P10, Morgan89}:
	
	\begin{lemma}\label{lem:class1}
		Let us consider the following $m\times m$ - reaction diffusion system: for all $i=1,...,m,$ 
		\begin{equation}
			\label{eq:class1}
			\partial_t u_i-d_i\Delta u_i=f_i(u_1,...,u_m)~in~ \mathbb{R}_+\times \Omega,~ \partial_\nu u_i=0~ \text{on}~ \partial \Omega, u_i(0)=u_{i0},
		\end{equation}
		where $d_i \in(0,+\infty)$, $f=(f_1,...,f_m):\mathbb{R}^m \rightarrow \mathbb{R}^m$ is $C^1(\Omega)$ and $u_{i0}\in L^{\infty}(\Omega)$. Then there exists a $T>0$ and a unique classical solution of \textcolor{blue}{(\ref{eq:class1})} on $[0,T).$ If $T^*$ denotes the greatest of these $T's$, then 
		\begin{equation*}
			\Bigg[\sup_{t \in [0,T^*),1\leq i\leq m} ||u_i(t)||_{L^{\infty}(\Omega)} < +\infty \Bigg] \implies [T^*=+\infty].
		\end{equation*}
		If the nonlinearity $(f_i)_{1\leq i\leq m}$ is moreover quasi-positive, which means 
		$$\forall i=1,..., m,~~\forall u_1,..., u_m \geq 0,~~f_i(u_1,...,u_{i-1}, 0, u_{i+1}, ..., u_m)\geq 0,$$
		then $$[\forall i=1,..., m, u_{i0}\geq 0]\implies [\forall i=1,...,m,~ \forall t\in [0,T^*), u_i(t)\geq 0].$$
	\end{lemma}
	
	\begin{lemma}\label{lem:class2}
		Using the same notations and hypotheses as in Lemma \ref{lem:class1}, suppose moreover that $f$ has at most polynomial growth and that there exists $\mathbf{b}\in \mathbb{R}^m$ and a lower triangular invertible matrix $P$ with nonnegative entries such that  $$\forall r \in [0,+\infty)^m,~~~Pf(r)\leq \Bigg[1+ \sum_{i=1}^{m} r_i \Bigg]\mathbf{b}.$$
		Then, for $u_0 \in L^{\infty}(\Omega, \mathbb{R}_+^m),$ the system (\ref{eq:class1}) has a strong global solution.
	\end{lemma}
	
	Under these assumptions, the following local existence result is well known, see D. Henry \cite{henry}.
	
	\begin{theorem}
		\label{thm:class3}
		The system (\ref{eq:class1}) admits a unique, classical solution $(u,v)$ on $%
		[0,T_{\max }]\times \Omega $. If $T_{\max }<\infty $ then 
		\begin{equation}
			\underset{t\nearrow T_{\max }}{\lim }\Big\{ \left\Vert u(t,.)\right\Vert
			_{\infty }+\left\Vert v(t,.)\right\Vert _{\infty } \Big\} =\infty ,  
		\end{equation}%
		where $T_{\max }$ denotes the eventual blow-up time in $\mathbb{L}^{\infty }(\Omega ).$
	\end{theorem}

	The next result follows from the application of standard theory \cite{kish88}.
	
	\begin{theorem}
		\label{thm:km1}
		Consider the reaction diffusion system (\ref{eq:class1}) . For spatially homogenous initial data $u_{0} \equiv c, v_{0} \equiv d$, with $c,d>0$, then the dynamics of (\ref{eq:class1}) and its resulting kinetic (ODE) system, when $d_{1}=d_{2}=0$ in (\ref{eq:class1}), are equivalent.
	\end{theorem}

	Our objective now is to consider the case of a fear function that may be heterogeneous in space. A motivation for this comes from several ecological and sociological settings. For example it is very common for prey to be highly fearful closer to a predators lair, but less fearful in a region of refuge \cite{Zhang19}, or in regions of high density due to group defense \cite{Samsal20}. To these ends, it is conceivable that the fear coefficient $q$ is not a constant, but actually varies in the spatial domain $\Omega$, so $q=q(x)$, which could take different forms depending on the application at hand. This is also in line with the LOF concept \cite{Brown99}. Thus we consider the following spatially explicit version of \eqref{1}, with heterogeneous fear function $q(x)$, as well as the Allee effect, resulting in the following reaction diffusion system,
	
	\begin{equation}\label{PDEmodel}
		\left\{\begin{array}{l}
			u_t =d_1 \Delta u + u \left [(1-u)(u-p) \dfrac{1}{1+q(x) v}-av \right ], \quad x \in  \Omega, \vspace{2ex}\\
			v_t =d_2 \Delta v + \left (1-v-cu \right ) \quad x \in  \Omega, \vspace{2ex}\\
			\dfrac{\partial u}{\partial \nu} = \dfrac{\partial v}{\partial \nu} =0, \quad \text{on} \quad \partial \Omega. \vspace{2ex}\\
			u(x,0)=u_0(x)\equiv c >0, \quad v(x,0)=v_0(x) \equiv d>0,
		\end{array}\right.
	\end{equation}
	whee $\Omega \subset \mathbb{R}^n.$

	Furthermore, we impose the following restrictions on the fear function $q(x)$,
	
	\begin{align}\label{eq:as1}
		\begin{split}
			&(i) \quad q(x)  \in C^{1}(\Omega),
			\\
			&(ii) \quad q(x) \geq 0,
			\\
			& (iii)\quad  \mbox{If} \  q(x) \equiv 0 \ \mbox{on}  \ \Omega_{1} \subset \Omega, \ \mbox{then} \ |\Omega_{1}| = 0.
			\\
			& (iv)\quad  \mbox{If} \  q(x) \equiv 0 \ \mbox{on}  \ \cup^{n}_{i=1}\Omega_{i} \subset \Omega, \ \mbox{then} \ \Sigma^{n}_{i=1}|\Omega_{i}| = 0.
		\end{split}
	\end{align}
	
	\begin{remark}
		If $q(x) \equiv 0$ on $\Omega_{1} \subset \Omega$, with $|\Omega_{1}| > \delta > 0$, or $q(x) \equiv 0$ on $\cup^{n}_{i=1}\Omega_{i} \subset \Omega$, with $\Sigma^{n}_{i=1}|\Omega_{i}| > \delta > 0$, that is, on non-trivial parts of the domain, the analysis is notoriously difficult, as one now is dealing with a \emph{degenerate} problem. See \cite{Du02a, Du02b} for results on this problem. This case is not in the scope of the current manuscript. 
	\end{remark}

	Since the nonlinear right hand side of (\ref{eq:class1}) is continuously
	differentiable on $\mathbb{R}^{+}\times $ $\mathbb{R}^{+}$, then for any
	initial data in $\mathbb{C}\left( \overline{\Omega }\right) $ or $\mathbb{L}%
	^{p}(\Omega ),\;p\in \left( 1,+\infty \right) $, it is standard to 
	estimate the $\mathbb{L}^{p}-$norms of the solutions and thus deduce global existence. Standard theory will apply even in the case of a bonafide fear function $k(x)$, because due to our assumptions on the form of $k$,  standard comparison arguments will apply \cite{Gil77}. Thus applying the classical methods above, via Theorem \ref{thm:class3}, and Lemmas \ref{lem:class1}-\ref{lem:class2}, we can state the following lemmas:

	\begin{lemma}
		\label{lem:pos1}
		Consider the reaction diffusion system \eqref{PDEmodel}, for $q(x)$ such that the assumptions via \eqref{eq:as1} hold. Then solutions to \eqref{PDEmodel} are non-negative, as long as they initiate from positive initial conditions.
	\end{lemma}

	\begin{lemma}
		\label{lem:cl1}
		Consider the reaction diffusion system \eqref{PDEmodel}. For $q(x)$ such that the assumptions via \eqref{eq:as1} hold. The solutions to \eqref{PDEmodel} are classical. That is for $(u_{0},v_{0}) \in \mathbb{L}^{\infty }(\Omega )$,  $(u,v) \in C^{1}(0,T; C^{2}(\Omega))$, $\forall T$.
	\end{lemma}
	
	Our goal in this section is to investigate the dynamics of \eqref{PDEmodel}. Herein we will use the comparison technique, and compare to the ODE cases of classical competition, or the constant fear function case, where the dynamics are well known. 
	
	\begin{remark}
		The analysis in this section are primarily focused on the choice of spatially homogenous (flat) initial data.
	\end{remark}

	Let's define some PDEs system,
	
	\begin{align}\label{eq:lv_model}
		\begin{split}
			\overline{u}_t &= d_1 \overline{u}_{xx} + \overline{u} \Big[ (1-\overline{u})(\overline{u}-p)  -a\overline{v} \Big] , \\
			\overline{v}_t &= d_2 \overline{v}_{xx} + \left (1-\overline{v}-c\overline{u} \right),
		\end{split}
	\end{align}
	
	\begin{align}\label{eq:upper}
		\begin{split}
			\widehat{u}_t &= d_1 \widehat{u}_{xx} + \widehat{u} \Big[ (1-\widehat{u} )(\widehat{u}-p)\dfrac{1}{1+\mathbf{\widehat{q}} \widehat{v}}  -a\widehat{v} \Big] , \\
			\widehat{v}_t &= d_2 \widehat{v}_{xx} + \left (1-\widehat{v}-c\widehat{u} \right),
		\end{split}
	\end{align}
	
	\begin{align}\label{eq:lower}
		\begin{split}
			\widetilde{u}_t &= d_1 \widetilde{u}_{xx} + \widetilde{u} \Big[ (1-\widetilde{u} )(\widetilde{u}-p)\dfrac{1}{1+ \mathbf{\widetilde{q}} \widetilde{v}}  -a\widetilde{v} \Big] , \\
			\widetilde{v}_t &= d_2 \widetilde{v}_{xx} + \left (1-\widetilde{v}-c\widetilde{u} \right),
		\end{split}
	\end{align}
	
	\begin{align}\label{eq:lowest}
		\begin{split}
			\tilde{u}_t &= d_1 \tilde{u}_{xx} + \tilde{u} \Big[ (1-\tilde{u} )(\tilde{u}-p)\dfrac{1}{1+\mathbf{\widetilde{q}}}  -a\tilde{v} \Big] , \\
			\tilde{v}_t &= d_2 \tilde{v}_{xx} + \left (1-\tilde{v}-c\tilde{u} \right),
		\end{split}
	\end{align}

	where
	\begin{align}\label{eq:lowest1}
		\mathbf{\widehat{q}} = \min_{x\in \Omega} q(x), \quad \, \quad   \mathbf{\widetilde{q} }= \max_{x\in \Omega} q(x).
	\end{align}
	
	We assume Neumann boundary conditions for all of the reaction diffusion systems \eqref{eq:lv_model}--\eqref{eq:lowest}. Also in each of the systems we prescribe spatially homogenous (flat) initial conditions
	$u(x,0)=u_0 (x) \equiv c>~0, \quad v(x,0)=v_0(x) \equiv d > 0.$
	
	\begin{theorem}\label{lem:com1}
		For the reaction diffusion system (\ref{PDEmodel}) for the Allee effect with a fear function $q(x)$, as well as the reaction diffusion  systems (\ref{eq:lv_model})-(\ref{eq:lowest}). Then the following point wise comparison holds,
		\[ \tilde{u} \le \widetilde{u} \le u \le \widehat{u} \le \overline{u}.\]
	\end{theorem}
	
	\begin{proof}
		From the positivity of the solutions to reaction diffusion systems \eqref{eq:upper}-\eqref{eq:lowest} and via comparison of \eqref{PDEmodel} to logistic equation to get upper bound for second species, i.e., $v \le 1$. Hence, we have
		\[ \dfrac{1}{1+\mathbf{\widetilde{q}}} \le \dfrac{1}{1+ \mathbf{\widetilde{q}} \widetilde{v}} \le  \dfrac{1}{1+ q(x) v} \le \dfrac{1}{1+\mathbf{\widehat{q}} \widehat{v}} \le 1, \quad x \in \Omega. \]
		Hence, the result follows from the standard comparison theory \cite{gilbarg1977elliptic}.
	\end{proof}
	
	Let's recall the notations: Denote the discriminant of \eqref{5} as 
	\begin{align}\label{eq:discrim}
		\Delta=A_2^2-4A_1A_3,
	\end{align}
	where 
	\[ A_1=ac^2q+1, \quad A_2=-(2acq+ac+p+1), \quad A_3=a+aq+p. \]

	\begin{theorem}\label{thm:ce1}
		For the reaction diffusion system (\ref{PDEmodel}) for the Allee effect with a fear function $q(x)$ that satisfies the parametric restriction
		\begin{equation}\label{eq:ce_pde}
			\Big( 2 a c \mathbf{q} + ac + p + 1 \Big)^2 < 4 \Big(ac^2 \mathbf{q}  +1 \Big) \Big(a+ a \mathbf{q} +p \Big)
		\end{equation}
		for $\mathbf{q}=\mathbf{\widehat{q}},\mathbf{\widetilde{q} }$, then the solution $(u,v)$ to (\ref{PDEmodel})
		converges uniformly to the spatially homogeneous state $(0,1)$ as $t \to \infty$ for any choice of flat initial data.
	\end{theorem}
	
	\begin{proof}
		
		Under the given parametric restriction for for $\mathbf{q}=\mathbf{\widehat{q}},\mathbf{\widetilde{q} }$, it is evident that discriminat $\Delta<0$, hence there is no real interior equilibrium. Moreover, from the geometry of the nullclines, we have $\dfrac{\mathrm{d} u}{\mathrm{d} t} \le \dfrac{\mathrm{d} v}{\mathrm{d} t}$. Hence, we have
		\[ (\widehat{u},\widehat{v}) \to (0,1) \quad \& \quad (\widetilde{u},\widetilde{v}) \to (0,1).\]
		Moreover, on using Lemma \ref{lem:com1} we have,
		\[ \widetilde{v}	 \leq   v    \leq \overline{v},\]
		which entails,
		\begin{align*}
			\lim_{t \rightarrow \infty}(\widetilde{u}, \widetilde{v})	 \leq  \lim_{t \rightarrow \infty} (u,v)    \leq \lim_{t \rightarrow \infty} (\widehat{u},\widehat{v}),
		\end{align*}
		subsequently,
		\begin{align*}
			\left(0,1\right)	 \leq  \lim_{t \rightarrow \infty}   (u,v)    \leq 	(0,1).
		\end{align*}
		Now using a squeezing argument, in the limit that $t \rightarrow \infty$, we have uniform convergence of solutions of \eqref{PDEmodel}, i.e.,
		\[ (u,v) \to (0,1)\] as $t \rightarrow \infty$.
	\end{proof}
	
	\begin{figure}[h]
		\subfigure[]
		{\scalebox{0.45}[0.45]{
				\includegraphics[width=\linewidth,height=5in]{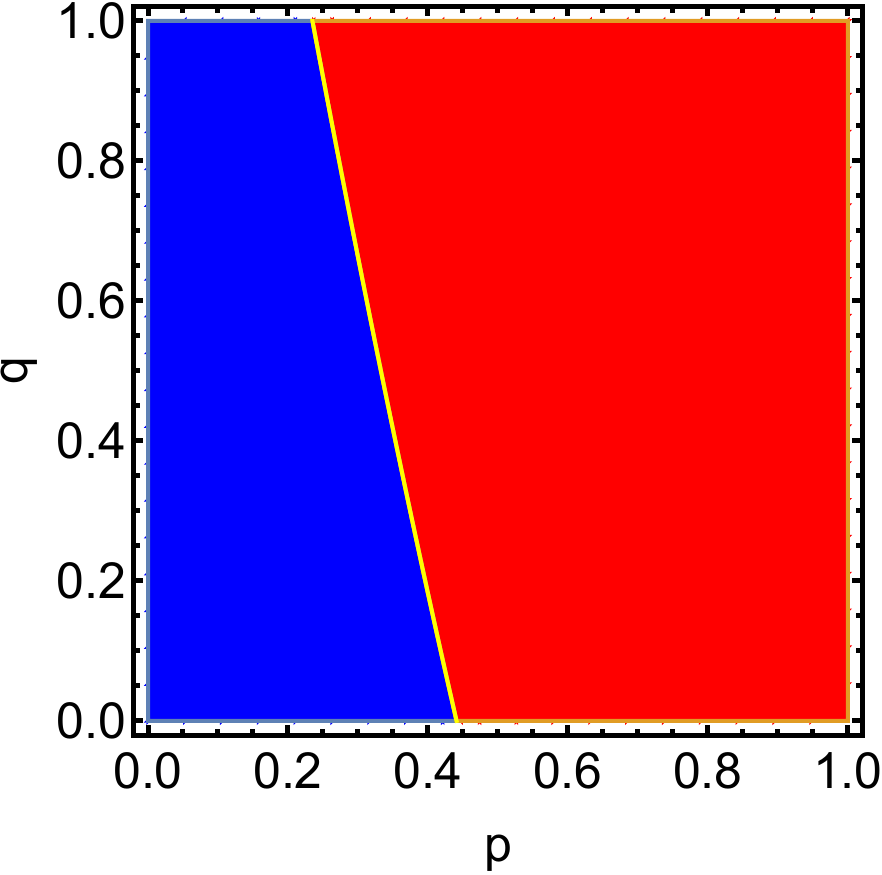}}}
		\subfigure[]
		{\scalebox{0.45}[0.45]{
				\includegraphics[width=\linewidth,height=4.9in]{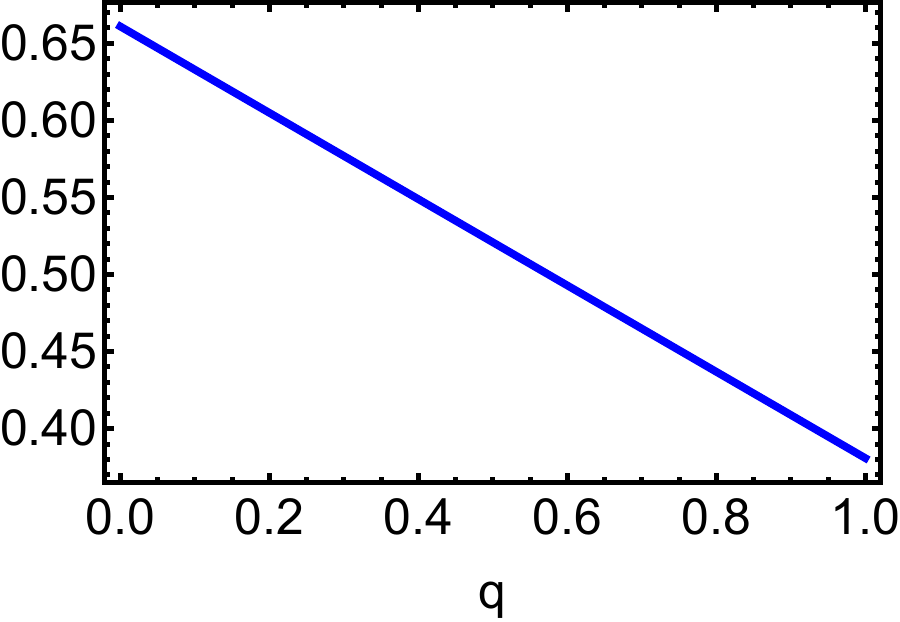}}}
		\caption{A graphical representation of $\Delta(p,q)$ defined by equation \eqref{eq:discrim} is presented through a region plot and a line plot. These plots illustrate the relationship between the parameters $p$ and $q$ that measure the Allee and fear for the reaction-diffusion system described in equation \eqref{PDEmodel}. The values of $a$, $b$, and $c$ are set to $0.1$, $0.7$, and $0.3$, respectively. In the region plot, the red region indicates that the competitive exclusion state $(0,1)$ is globally stable for any initial data, while the blue region shows that the dynamics depend on the choice of initial data. The line plot specifically represents the case where the Allee effect is weak ($p=0$).}
		\label{fig_region}
	\end{figure}

	\begin{remark}
		The global stability of the competition state $(0,1)$ holds true for both the case of strong and weak Allee effects, which entails Theorem~\ref{thm:ce1} also stands true for both types of Allee effects (See Figs~[\ref{fig_ce_st_al} ,\ref{fig_ce_wk_al}]). 
	\end{remark}
	
	\begin{figure}[H]
		\subfigure[]
		{\scalebox{0.45}[0.45]{
				\includegraphics[width=\linewidth,height=4.5in]{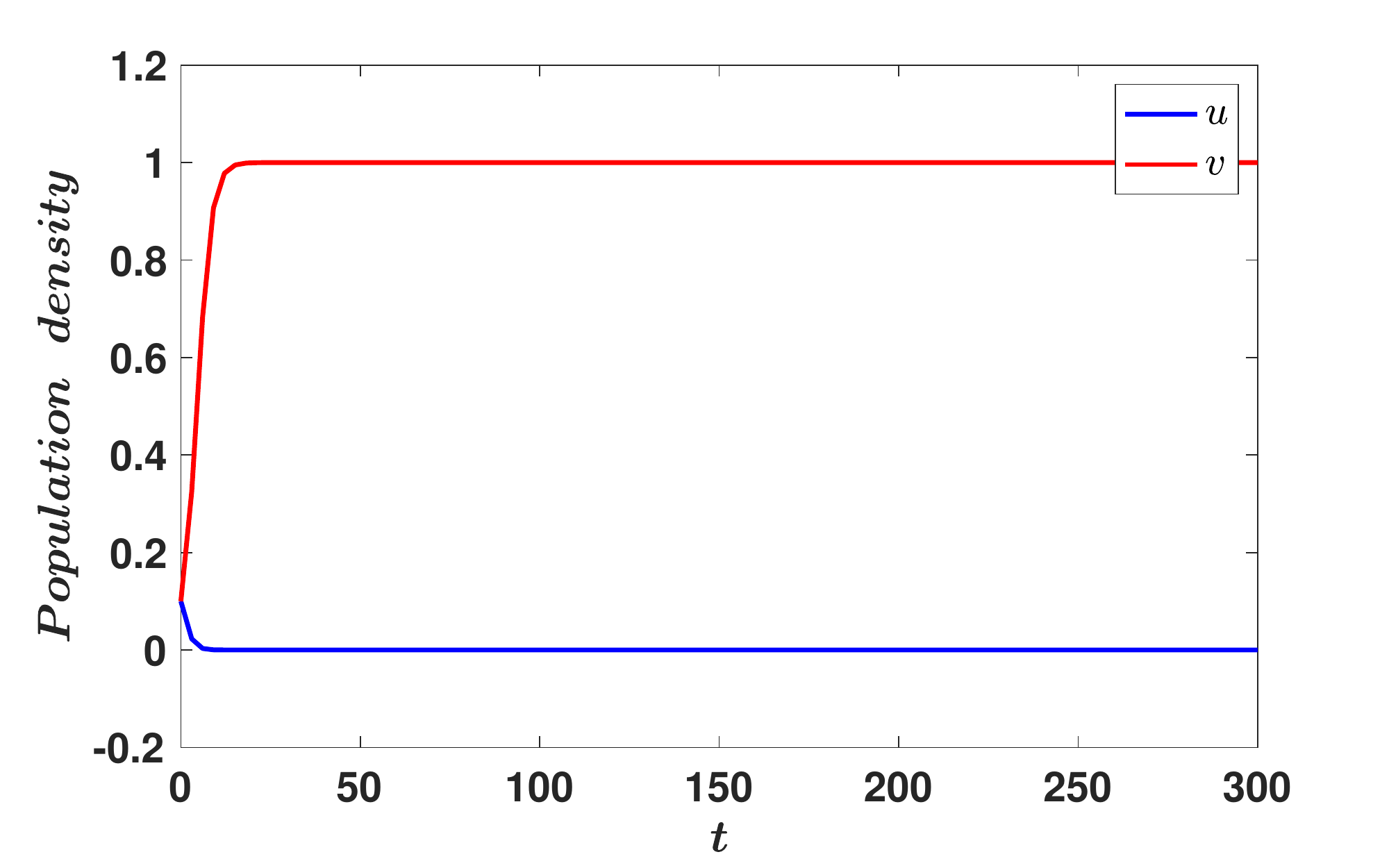}}}
		\subfigure[]
		{\scalebox{0.45}[0.45]{
				\includegraphics[width=\linewidth,height=4.5in]{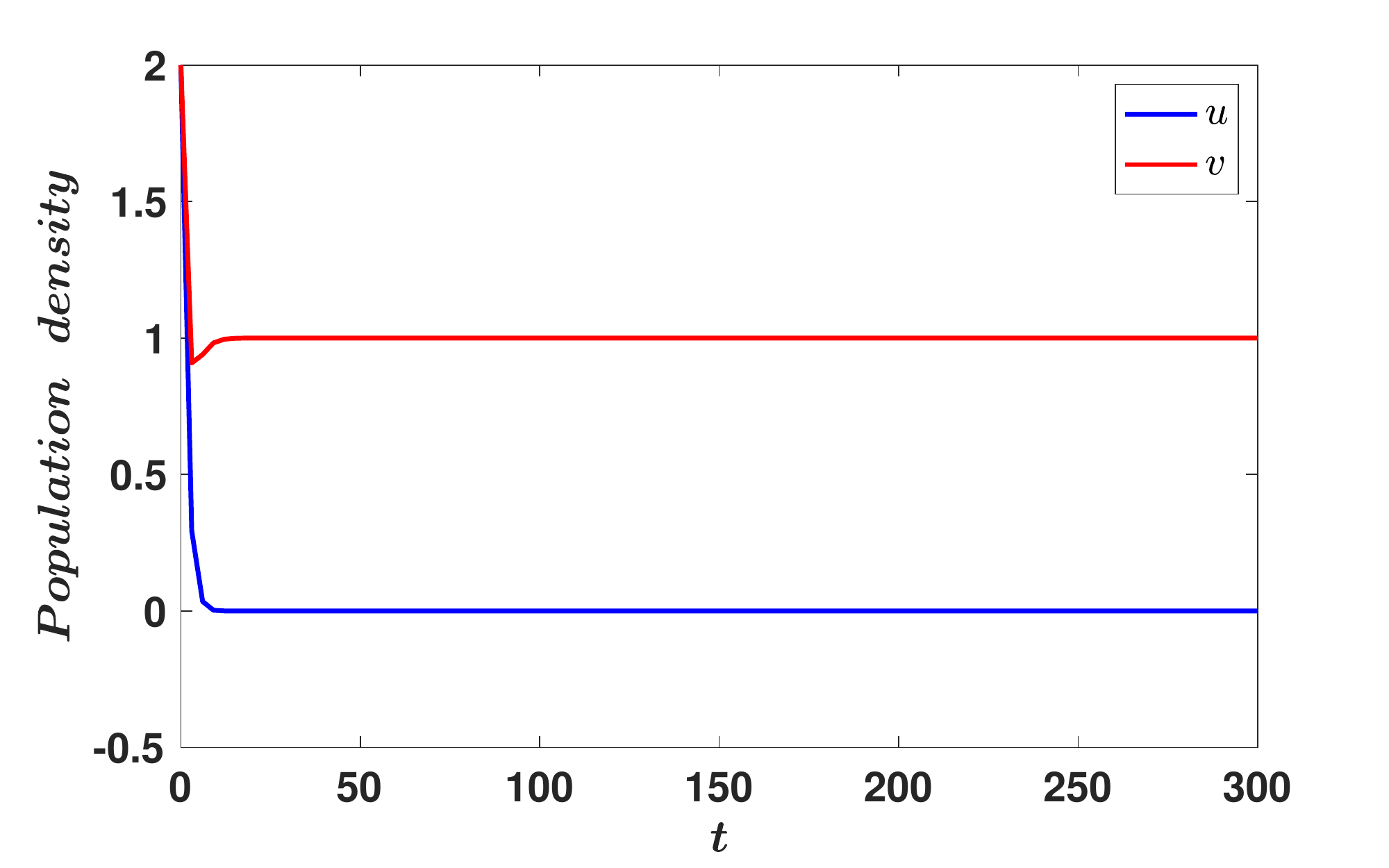}}}
		\caption{The parameters specified for the simulation of the reaction-diffusion system given in (\ref{PDEmodel}) on the spatial domain $\Omega=[0,1]$ for the case of competition exclusion with fear function $q(x)=\sin^2 (4x)$ and strong Allee effect are $d_1=1,d_2=1,a=0.5,b=0.5,c=0.5$ and $p=0.5$ Two different sets of initial data is used: (a) $[u_0,v_0]=[0.1,0.1]$ and (b) $[u_0,v_0]=[2,2]$. It should be noted that these parameters satisfy the parametric constraints specified in the Theorem~\ref{thm:ce1}}.
		\label{fig_ce_st_al}
	\end{figure}

	\begin{figure}[h]
		\subfigure[]
		{\scalebox{0.45}[0.45]{
				\includegraphics[width=\linewidth,height=4.5in]{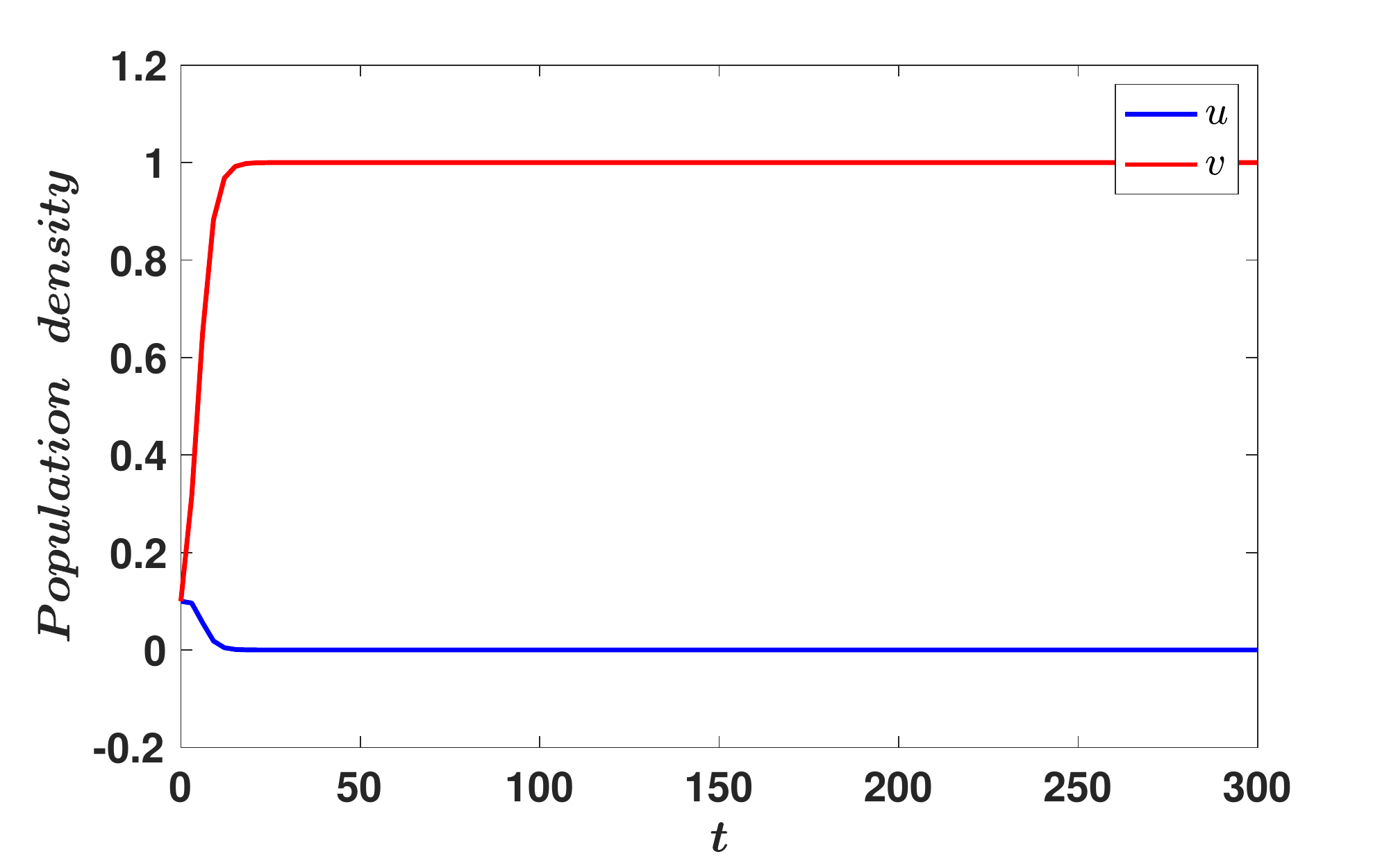}}}
		\subfigure[]
		{\scalebox{0.45}[0.45]{
				\includegraphics[width=\linewidth,height=4.5in]{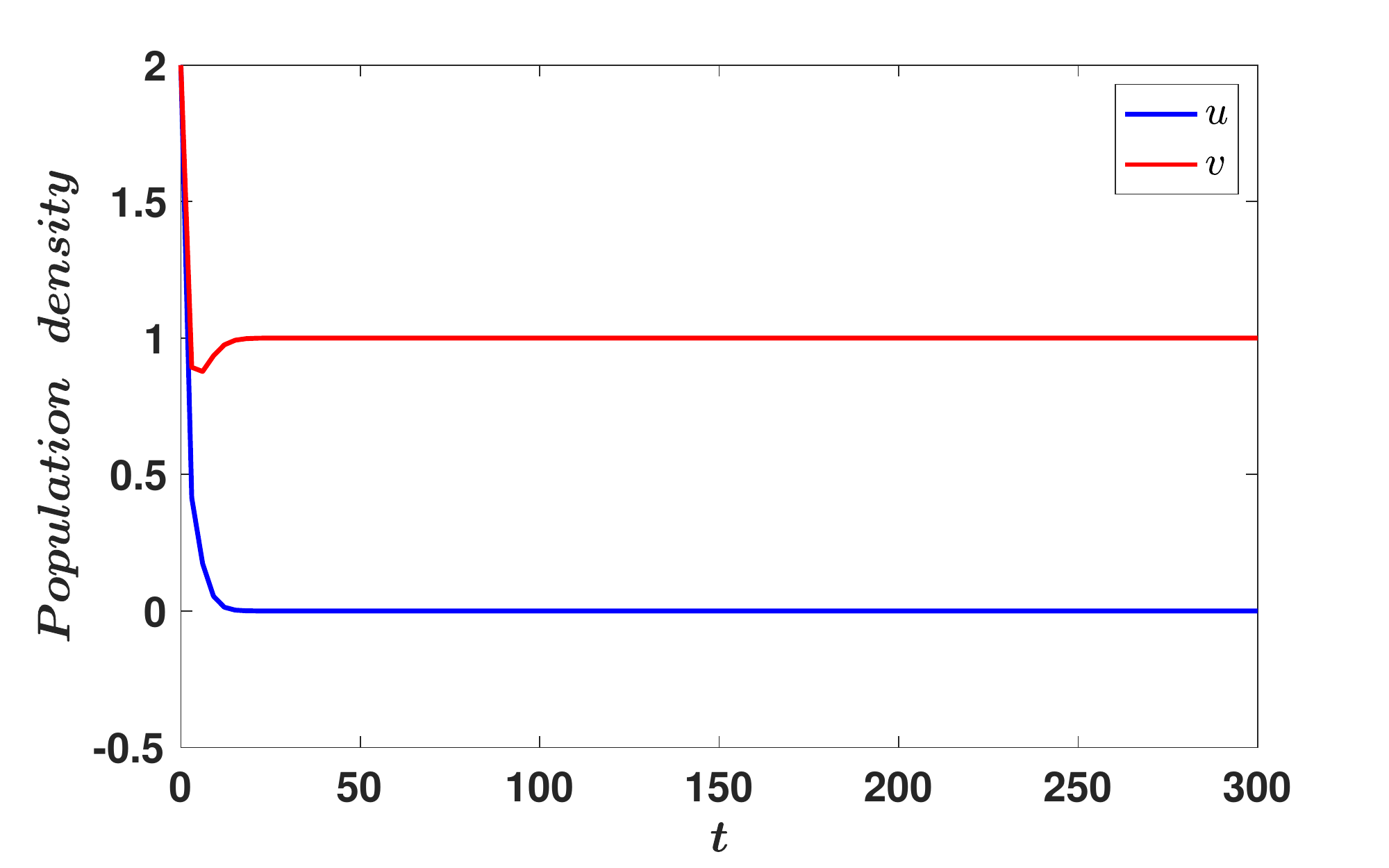}}}
		\caption{The parameters specified for the simulation of the reaction-diffusion system given in (\ref{PDEmodel}) on the spatial domain $\Omega=[0,1]$ for the case of competition exclusion with fear function $q(x)=\sin^2 (4x)$ and weak Allee effect are $d_1=1,d_2=1,a=0.5,b=0.5,c=0.5$ and $p=0$. Two different sets of initial data is used: (a) $[u_0,v_0]=[0.2,0.6]$and (b) $[u_0,v_0]=[1,0.1]$. It should be noted that these parameters satisfy the parametric constraints specified in the Theorem~\ref{thm:ce1}}.
		\label{fig_ce_wk_al}
	\end{figure}
	
	We now provide a lemma that entails a result such as Theorem \ref{thm:ce1}, but for any initial data,

	\begin{corollary}\label{cor:ce1n}
		For the reaction diffusion system (\ref{PDEmodel}) with Allee effect and a fear function $q(x)$ that satisfies the parametric restriction
		\begin{equation}\label{eq:ce_pde}
			\Big( 2 a c \mathbf{q} + ac + p + 1 \Big)^2 < 4 \Big(ac^2 \mathbf{q} +1 \Big) \Big(a+ a \mathbf{q} +p \Big),
		\end{equation}
		for $\mathbf{q}=\mathbf{\widehat{q}},\mathbf{\widetilde{q} }$, the solution $(u,v)$ to (\ref{PDEmodel})
		converges uniformly to the spatially homogeneous state $(0,1)$ as $t \to \infty$ for initial data, $(u_{0}(x), v_{0}(x)) \in L^{\infty}(\Omega)$.
	\end{corollary}
	
	\begin{proof}
		The proof follows the proof of Theorem \ref{thm:ce1}. We can show that the solution to (\ref{PDEmodel}) via comparison is bounded by the solution to \eqref{eq:lv_model}. Via the geometry of nullclines and parametric restrictions, the solution to \eqref{eq:lv_model} is bounded by the solution to the classical Lokta-Volterra model, with the same parametric restrictions, only without the Allee term.
		Next we can construct a Lyapunov function,
		
		\begin{equation}\label{eq:ce_pde}
			E(u,v) = \int_{\Omega}\left(|u|^{2} + |v-1|^{2} \right)dx,
		\end{equation}
		for the classical Lokta-Volterra model. Via standard methods \cite{Ni2011} we see that for the parametric restrictions considered, $\frac{d}{dt}E(u,v) \leq 0$, so $(u,v) \rightarrow (0,1)$ as $t \rightarrow \infty$.
		This is the classical case of competitive exclusion. Now the geometry of nullclines and parametric restrictions, enable via direct comparison to show $(\tilde{u},\tilde{v}) \rightarrow (0,1)$ as $t \rightarrow \infty$, where $(\tilde{u},\tilde{v})$, is the solution to (\ref{PDEmodel}).
		
	\end{proof}
	
	\begin{figure}[h]
		\subfigure[]
		{\scalebox{0.45}[0.45]{
				\includegraphics[width=\linewidth,height=5in]{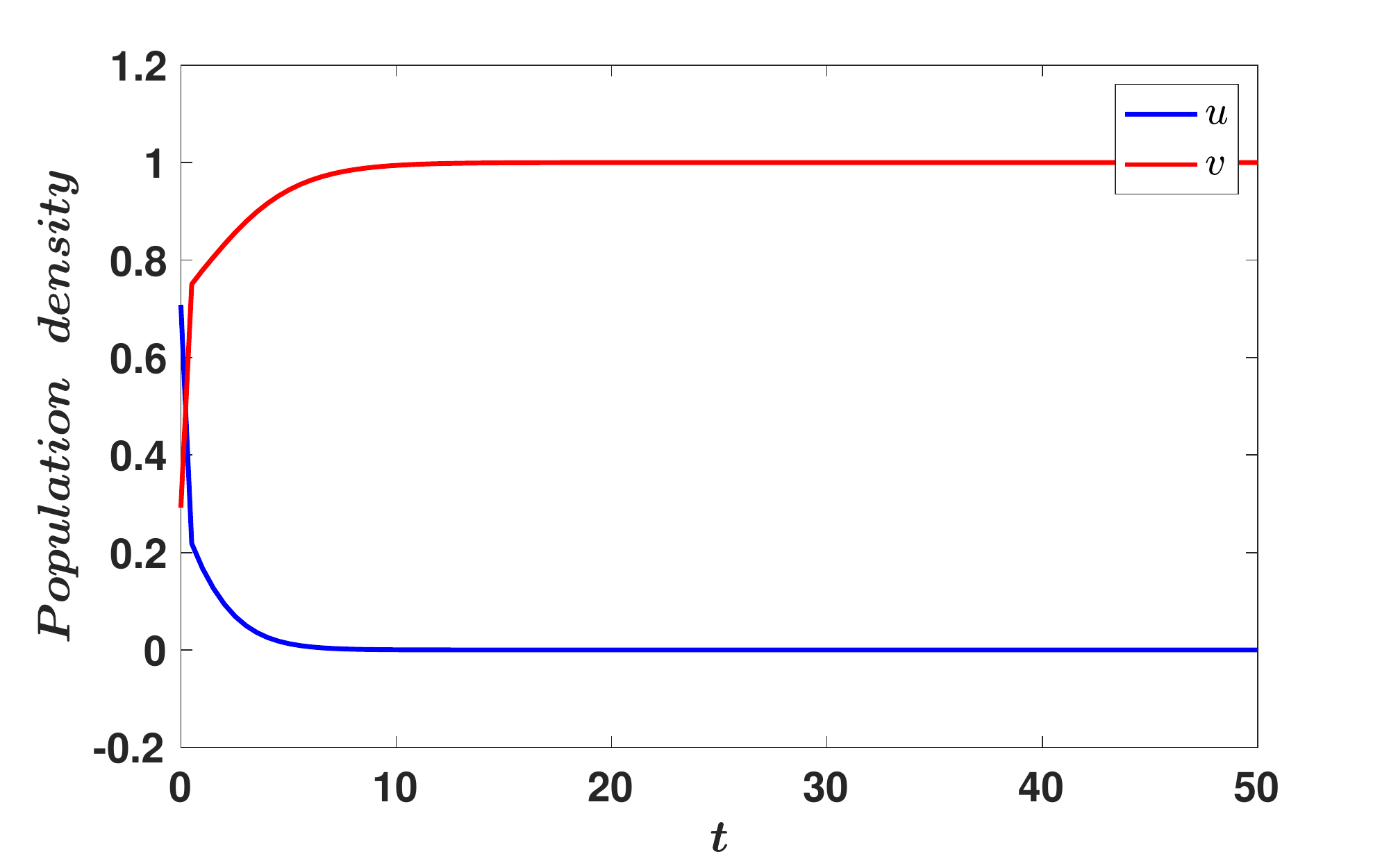}}}
		\subfigure[]
		{\scalebox{0.45}[0.45]{
				\includegraphics[width=\linewidth,height=5in]{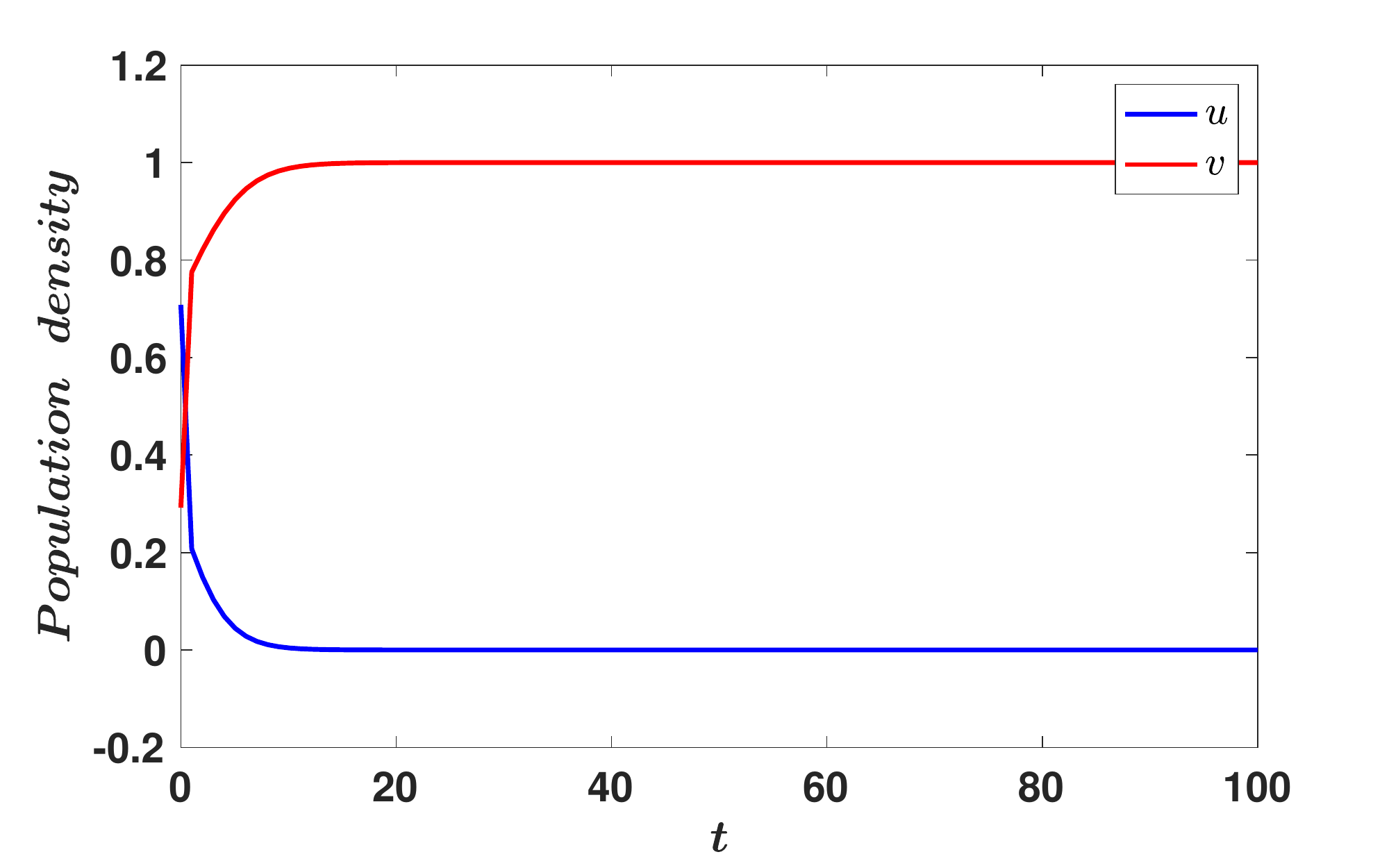}}}
		\caption{The parameters specified for the simulation of the reaction-diffusion system given in (\ref{PDEmodel}) on the spatial domain $\Omega=[0,1]$ for the case of competition exclusion with the fear function $q(x)=0.8 + \sin^2 (4x)$ are $d_1=1,d_2=1,a=0.5,b=0.5,c=0.5$ and $[u_0(x),v_0(x)]=[sin^2 (x),cos^2 (x)]$. (a) Strong Allee $p=0.5$ and (b) Weak Allee $p=0$. It should be noted that these parameters satisfy the parametric constraints specified in the Corollary~\ref{cor:ce1n}}.
		\label{fig_ce_cor}
	\end{figure}

	\begin{theorem}\label{thm:str_str}
		For the reaction diffusion system (\ref{PDEmodel}) for the Strong Allee effect with a fear function $q(x)$ that satisfies the parametric restriction
		\begin{equation}\label{eq:st_pde_str_allee}
			\Delta >0, \quad ,\quad 2 A_1 + A_2 >0, \quad  1 \le c \le \dfrac{1}{\mathbf{q}} + 1,\quad \text{and} \quad 0< p < \dfrac{1}{c}
		\end{equation}
		for $\mathbf{q}=\mathbf{\widehat{q}},\mathbf{\widetilde{q} }$. Then there exists sufficiently small initial data $[u_0(x),v_0(x)]$ such that the solution $(u,v)$ to (\ref{PDEmodel}) converges uniformly to the spatially homogeneous state $(1,0)$ as $t \to \infty$, while there exits also sufficiently large intial data $[u_1 (x),v_1 (x)]$, for which the solution $(u,v)$ to (\ref{PDEmodel}) converges uniformly to the spatially homogeneous state $(0,1)$ as $t \to \infty$.
	\end{theorem}

	\begin{proof}
		Consider the reaction diffusion system \eqref{eq:upper}. Since the $\widehat{q}$ satisfies the parametric restriction, from Theorem~\ref{thm:exist}(c) and Theorem~\ref{thm:stab_pos}, there exists a interior saddle equilibrium $E_{1*}$ to the kinetic (ODE) system \eqref{eq:upper}. On making use of the stable manifold theorem \cite{23}, i.e., $\exists \hspace{0.1in} W^1_{s}(E_{1*}) \in \mathcal{C}^{1}$ separatrix, such that for initial data $(\widehat{u}_0,\widehat{v}_0)$ chosen above $W^1_{s}(E_{1*})$ the solution $(\widehat{u},\widehat{v}) \to (0,1)$ and for initial data chosen below $W^1_{s}(E_{1*})$, $(\widehat{u},\widehat{v}) \to (1,0)$.

		Moreover, notice that $\dfrac{1}{1 + \mathbf{\widetilde{q}} v} \le \dfrac{1}{1 + \mathbf{\widehat{q}} v}$, we have that for the kinetic (ODE) system \eqref{eq:lower}, we still remain in the strong competition case, and via standard theory again, $\exists \hspace{0.1in} W_{s}(E_{1**}) \in \mathcal{C}^{1}$ separatrix, such that for initial data $(\widetilde{u}_0,\widetilde{v}_0)$ chosen above $W_{s}(E^{1**})$ the solution $(\widetilde{u},\widetilde{v}) \to (0,1)$ and for initial data chosen below $W_{s}(E_{1**})$, $(\widetilde{u},\widetilde{v})\to (1,0)$. Here $E^{1**}$ is the interior saddle equilibrium to the kinetic (ODE) system for \eqref{eq:lower}.
		
		Now since $\dfrac{1}{1 + \mathbf{\widetilde{q}} v} \le \dfrac{1}{1 + \mathbf{\widehat{q}} v}$, the $u$ component of $E_{1**}$ is more than the $u$ component of $E_{1*}$. Now using the $\mathcal{C}^{1}$ property of the separatricies $W^{1}_{s}(E_{1*}) , W_{s}(E_{1**})$, we have the existence of a wedge $\mathbb{V}$ emanating from $E_{1*}$, s.t within $\mathbb{V}$ we have $W^{1}_{s}(E_{1*}) \leq W_{s}(E_{1**})$. Note via Lemma \ref{lem:com1} we have $ \widetilde{u} \leq u \leq \widehat{u}$.	Let us consider positive initial data $(u_0,v_0)$ chosen small enough, within $\mathbb{V}$ s.t. $  (u_0,v_0) < W^{1}_{s}(E_{1*})  \le W_{s}(E_{1**})$, we will have 
		\begin{align*}
			\Big\{  (1,0) \Big\} \le \Big\{ (u,v) \Big\} \le \Big\{ (1,0) \Big\}.
		\end{align*}
		On the other hand,  for sufficiently large initial data $(u_1,v_1)$ via an analogous construction we will have 
		\begin{align*}
			\Big\{  (0,1)\Big\} \le \Big\{ (u,v) \Big\} \le \Big\{ (0,1) \Big\}.
		\end{align*}
		This proves the theorem.
	\end{proof}
	
	\begin{figure}[H]\label{str_sepratrix}
		\centering
		\includegraphics[scale=0.5]{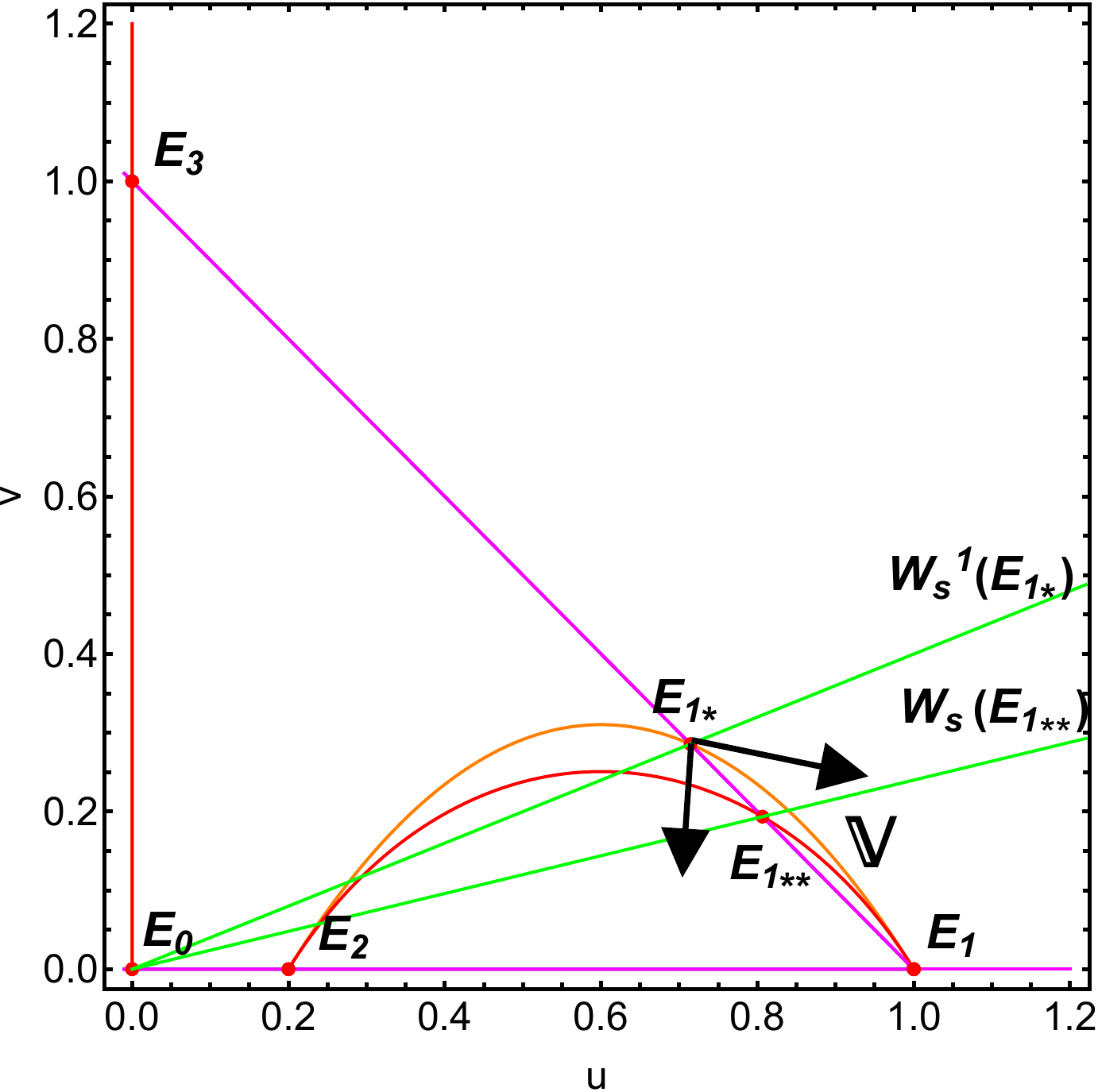}
		\caption{Numerical simulation for the Allee competition PDE case when species $v$ is fearing species $u$ in $\Omega=[0,1]$. The parameters are chosen as $a=0.5,b=0.5,c=1$ and $p=0.2$. Equilibria: $E_{1*}=(0.7142,0.2857),E_{1**}=(0.8064,0.1935), E_2=(0.2,0), E_1=(1,0),E_3=(0,1)$ and $E_0=(0,0)$. $W_{s}^{1}(E_{1*})$ ($q=0.1$) and $W_{s} (E_{1**)}$ ($q=1.1$) are two sepratrices passing through $E_{1*}$ and  $E_{1**}$ respectively. The $C^{1}$ property of the separatrices $W^{1}_{s}(E^{1*} , W_{s}(E_{1**})$, shows a wedge $V$ emanating from $E_{1*}$, s.t within $V$ we have $W_{s}(E^{1**} \geq W^{1}_{s}(E_{1*})$. The u-nullcline is in orange for $q=0.1$  and red for $q=1.1$. The v-nullcline is in magenta for $q=0.1$ and $q=1.1$.}
	\end{figure}
	
	The above theorem can also be  stated for the Weak Allee model and the proof follows the same argument as the Theorem~\ref{thm:str_str}.
	
	\begin{theorem}\label{thm:str_weak}
		For the reaction diffusion system (\ref{PDEmodel}) for the Strong Allee effect with a fear function $q(x)$ that satisfies the parametric restriction
		\begin{equation}\label{eq:st_pde_weak_allee}
			\Delta >0, \quad ,\quad 2 A_1 + A_2 >0,\quad \text{and}  \quad c=1
		\end{equation}
		for $\mathbf{q}=\mathbf{\widehat{q}},\mathbf{\widetilde{q} }$. Then there exists sufficiently large initial data $[u_0(x),v_0(x)]$ such that the solution $(u,v)$ to (\ref{PDEmodel}) converges uniformly to the spatially homogeneous state $(0,1)$ as $t \to \infty$, while there exits also sufficiently small data $[u_1 (x),v_1 (x)]$, for which the solution $(u,v)$ to (\ref{PDEmodel}) converges uniformly to the spatially homogeneous state $(1,0)$ as $t \to \infty$.
	\end{theorem}

	\begin{figure}[h]
		\subfigure[]
		{\scalebox{0.45}[0.45]{
				\includegraphics[width=\linewidth,height=4.5in]{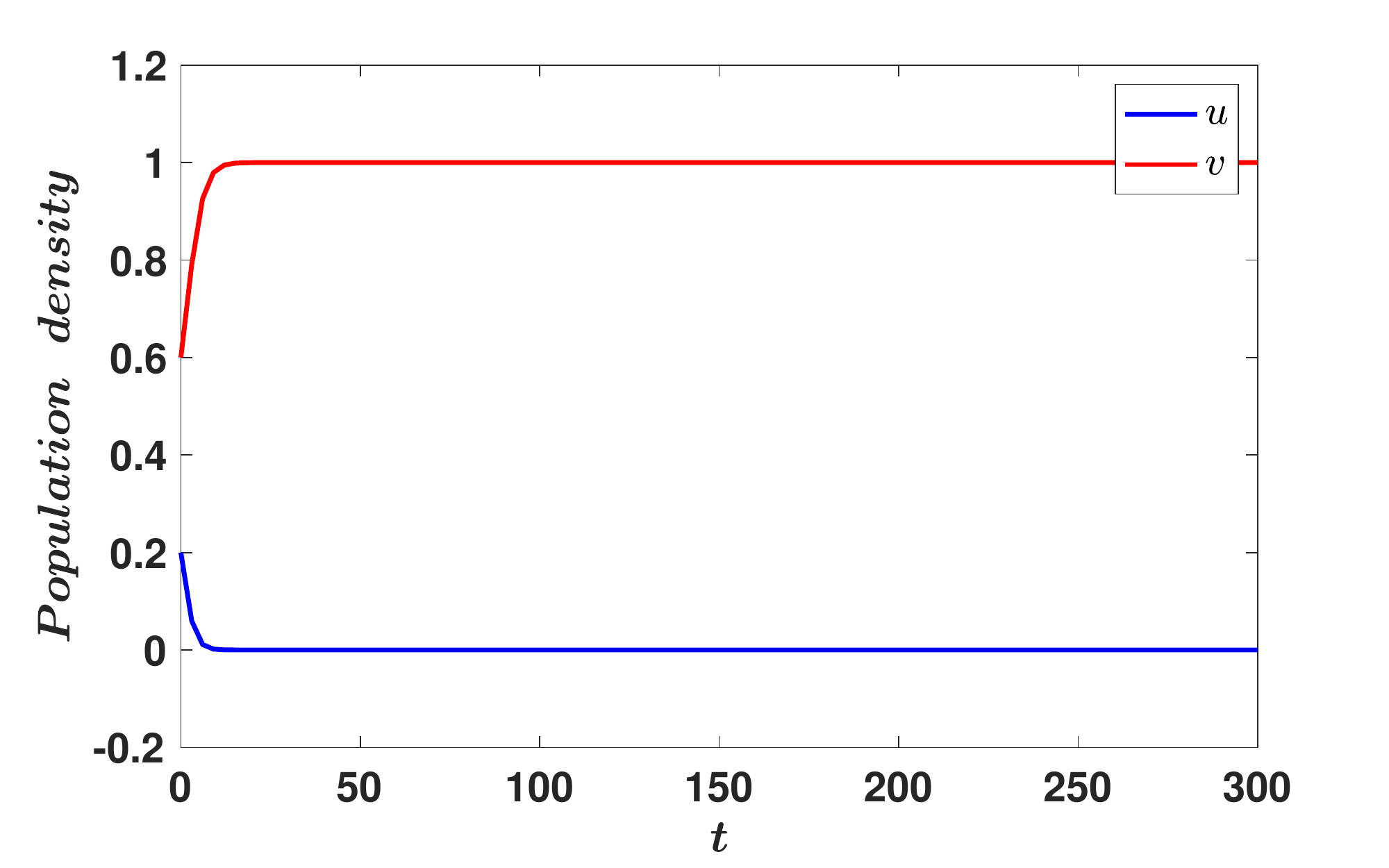}}}
		\subfigure[]
		{\scalebox{0.45}[0.45]{
				\includegraphics[width=\linewidth,height=4.5in]{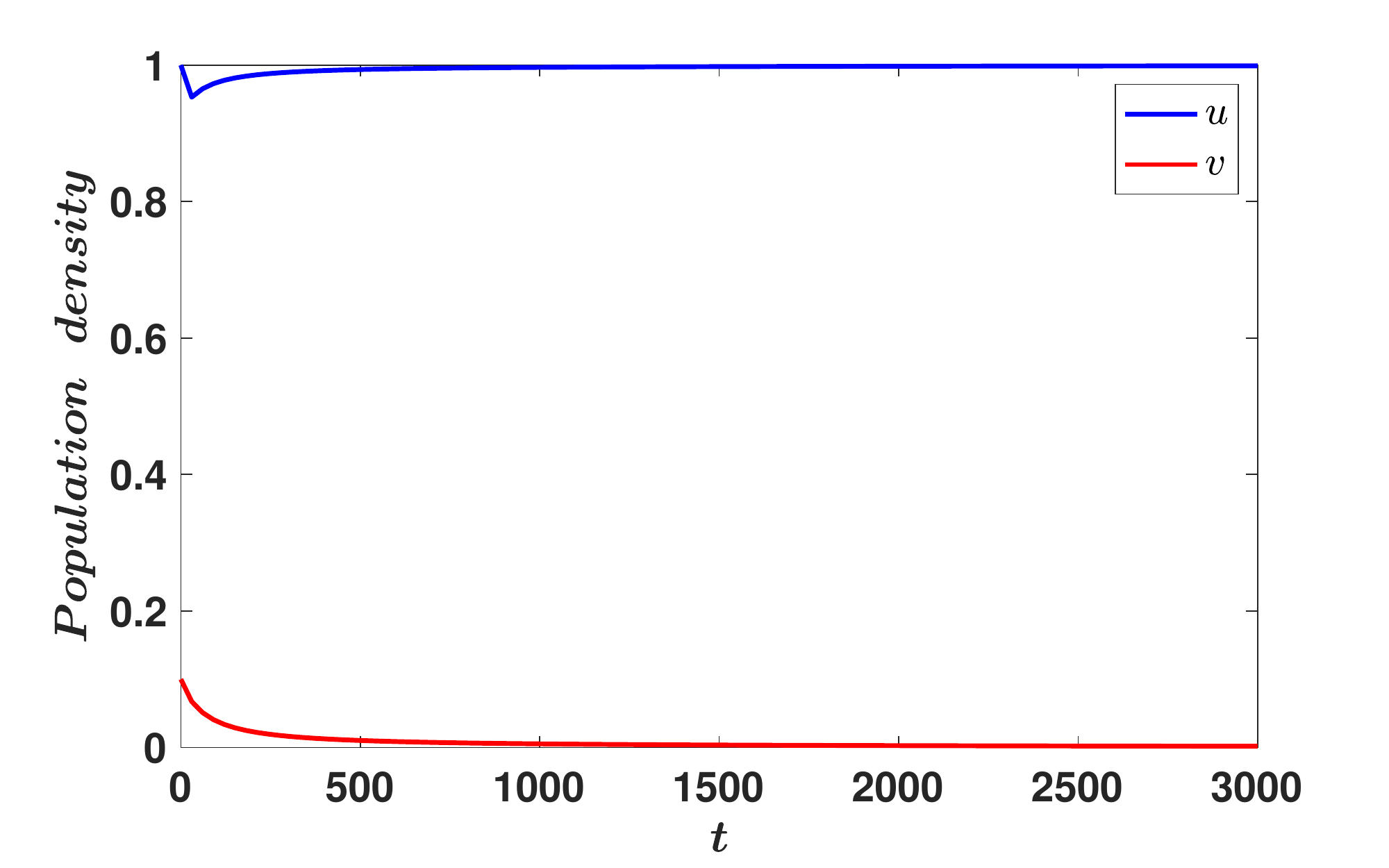}}}
		\caption{The parameters specified for the simulation of the reaction-diffusion system given in (\ref{PDEmodel}) on the spatial domain $\Omega=[0,1]$ for strong competition with fear function $q(x)=0.1+\sin^2 (4x)$ and strong Allee effect are $d_1=1$, $d_2=1$, $a=0.5$, $b=0.5$, $c=1$, and $p=0.2$. Two different sets of initial data is used: (a) $[u_0,v_0]=[0.2,0.6]$ and (b) $[u_0,v_0]=[1,0.1]$. It should be noted that these parameters satisfy the parametric constraints specified in the Theorem~\ref{thm:str_str}.}
		\label{fig_st_st_al}
	\end{figure}

	\begin{figure}[h]
		\subfigure[]
		{\scalebox{0.45}[0.45]{
				\includegraphics[width=\linewidth,height=4.5in]{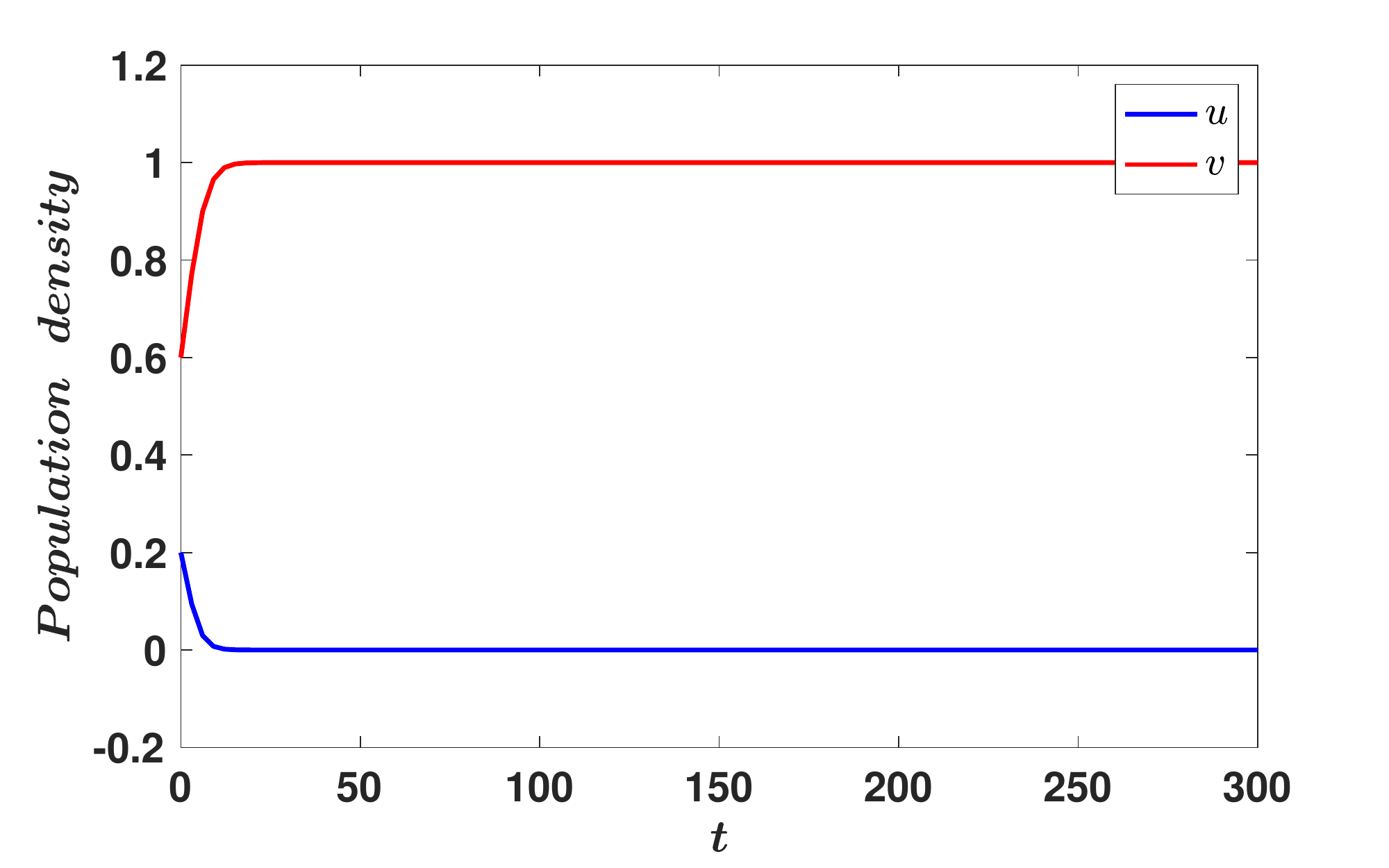}}}
		\subfigure[]
		{\scalebox{0.45}[0.45]{
				\includegraphics[width=\linewidth,height=4.5in]{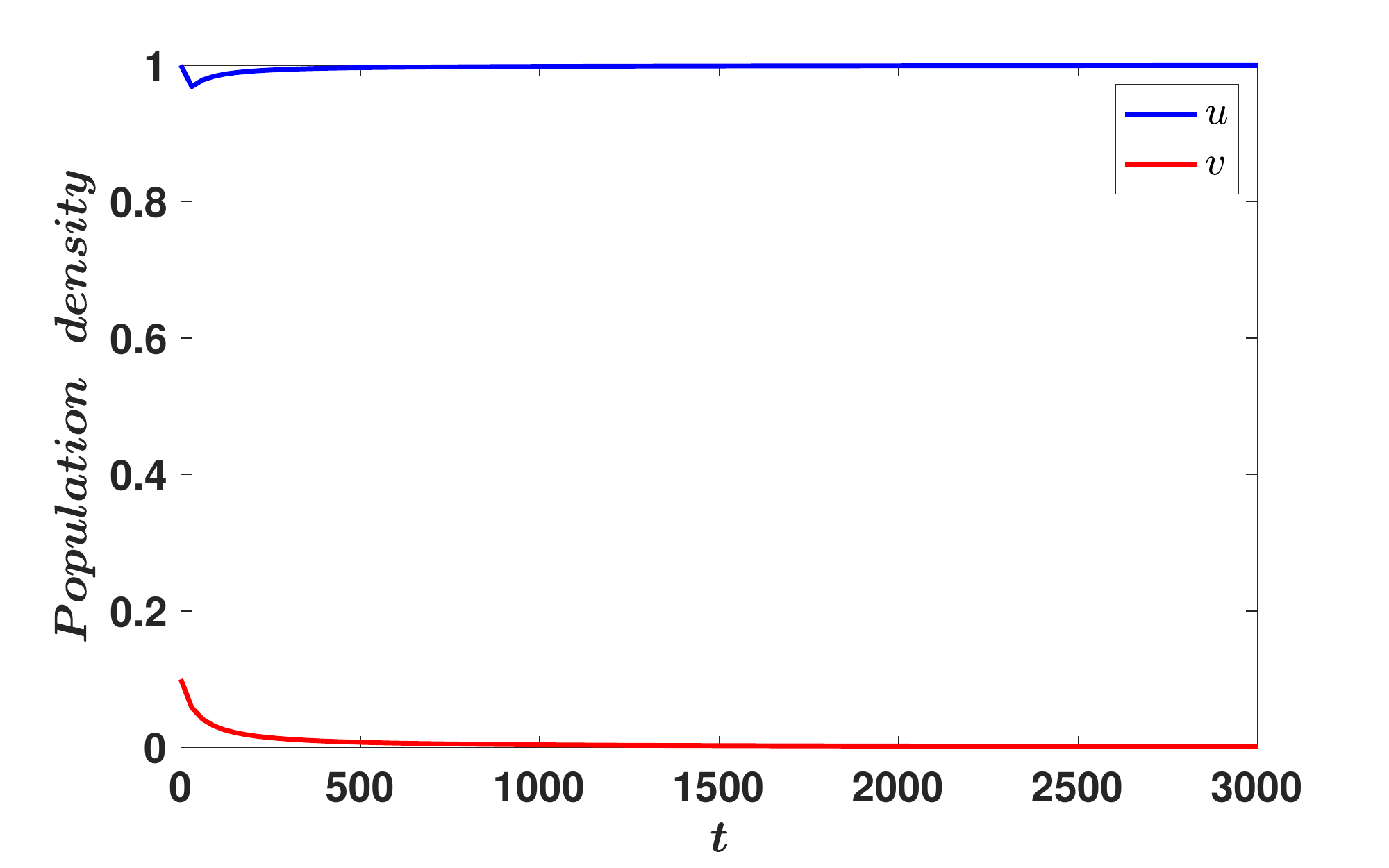}}}
		\caption{The parameters specified for the simulation of the reaction-diffusion system given in (\ref{PDEmodel}) on the spatial domain $\Omega=[0,1]$ for strong competition with fear function $q(x)=0.1+\sin^2 (4x)$ and weak Allee effect are $d_1=1$, $d_2=1$, $a=0.5$, $b=0.5$, $c=1$, and $p=0$. Two different sets of initial data is used: (a) $[u_0,v_0]=[0.2,0.6]$ and (b) $[u_0,v_0]=[1,0.1]$. It should be noted that these parameters satisfy the parametric constraints specified in the Theorem~\ref{thm:str_weak}}.
		\label{fig_st_wk_al}
	\end{figure}

	\begin{remark}
		For the reaction diffusion system (\ref{PDEmodel}) for the Strong Allee effect with a fear function $q(x)$ that satisfies the parametric restriction
		\begin{equation}\label{eq:st_pde_str_allee}
			\Delta >0 \quad \& \quad c>1
		\end{equation}
		holds true for $\mathbf{q}=\mathbf{\widehat{q}},\mathbf{\widetilde{q} }$. Then there exists sufficiently small data $[u_0 (x),v_0 (x)]$, for which the solution $(u,v)$ to (\ref{PDEmodel}) converges uniformly to the spatially homogeneous state $(1,0)$ as $t \to \infty$.
	\end{remark}

	Denote $q_c$, critical amount of fear required to transit from the globaly stable competition exclusion state $(0,v^*)$ to the other dynamical case, i.e for any $c \in (0,1)$ and $\mathbf{q} \in [0,q_c]$, we have 
	\[ \Delta (\mathbf{q}) > 0 \quad \& \quad 2 A_1 + A_2>0.\]
	
	\begin{theorem}\label{thm:two_post}
		For the reaction diffusion system (\ref{PDEmodel}) for the Allee effect (both strong and weak) with a fear function $q_{\epsilon}(x)$, where $(0 < \epsilon \ll 1)$ such that the parametric restriction
		\begin{equation}\label{eq:st_pde_str_allee}
			\Delta >0, \quad \quad 2 A_1 + A_2 >0,\quad \text{and} \quad 0< c < 1
		\end{equation}
		holds true for $\mathbf{\widetilde{q_\epsilon}}$ and $q(x)\equiv 0$. Then for some initial data $[u_0(x),v_0(x)]$ such that the solution $(u,v)$ to (\ref{PDEmodel}) converges uniformly to  $(u^*,v^*)$ as $t \to \infty$, while for some choice of initial data $[u_1 (x),v_1 (x)]$, for which the solution $(u,v)$ to (\ref{PDEmodel}) converges uniformly to the spatially homogeneous state $(0,1)$ as $t \to \infty$.
	\end{theorem}
	
	\begin{proof}
		Given $\epsilon$ such that $0 < \epsilon \ll 1$, we can always construct a fear function $q_\epsilon (x)$ such that $q_c -\epsilon \le \mathbf{\overline{q_\epsilon}}$ and $\mathbf{\widetilde{q}_\epsilon} \le q_c +\epsilon$. Hence from \ref{lem:com1}, we have $\widetilde{u} \le u_\epsilon \le \overline{u}$.
		
		Now, consider the reaction diffusion system \eqref{eq:lv_model}. Since the $q(x)\equiv 0$ satisfies the parametric restriction, from Theorems~\ref{thm:exist}(a) and Theorem~\ref{thm:stab_pos}, there exists a interior saddle equilibrium $E_{1*}$ and interior stable equilibrium $E_{2*}$ to the kinetic (ODE) system \eqref{eq:lv_model}. On making use of the stable manifold theorem, i.e., $\exists \hspace{0.1in} W^1_{s}(E_{1*}) \in \mathcal{C}^{1}$ separatrix, such that for initial data $(\overline{u}_0,\overline{v}_0)$ chosen above $W^1_{s}(E_{1*})$ the solution $(\overline{u},\overline{v}) \to (0,1)$ and for initial data chosen below $W^1_{s}(E_{1*})$, $(\overline{u},\overline{v}) \to (u^*,v^*)$ as $t\to \infty$.

		Also, for the reaction diffusion system \eqref{eq:lower}, since the $q_{\epsilon}$ satisfies the parametric restriction, from Theorems~\ref{thm:exist}(a) and Theorem~\ref{thm:stab_pos}, there exists a interior saddle equilibrium $E_{1**}$ and interior stable equilibrium $E_{2**}$ to the kinetic (ODE) system. Via standard theory again, $\exists \hspace{0.1in} W_{s}(E_{1**}) \in \mathcal{C}^{1}$ separatrix, such that for initial data $(\widetilde{u}_0,\widetilde{v}_0)$ chosen above $W_{s}(E^{1**})$ the solution $(\widetilde{u},\widetilde{v}) \to (0,1)$ and for initial data chosen below $W_{s}(E_{1**})$, $(\widetilde{u},\widetilde{v})\to (u^{**},v^{**})$ as $t \to \infty.$

		Now, since $\dfrac{1}{1 + \mathbf{\widetilde{q}} v} \le \dfrac{1}{1 + \mathbf{\widehat{q}} v}$, the $u$ component of $E_{1**}$ is more than the $u$ component of $E_{1*}$. Now using the $\mathcal{C}^{1}$ property of the separatricies $W^{1}_{s}(E_{1*}) , W_{s}(E_{1**})$, we have the existence of a wedge $\mathbb{V}_1$ emanating from $E_{1*}$, s.t within $\mathbb{V}_1$ we have $W^{1}_{s}(E_{1*}) \leq W_{s}(E_{1**})$. Similarly, since the $v$ componet of  $E_{1*}$ is higher than $v$ component of $E_{1**}$, we have the existence of wedge $\mathbb{V}_2$ emanating from $E_{1*}$, s.t within $\mathbb{V}_2$ we have $W^{1}_{s}(E_{1*}) \geq W_{s}(E_{1**})$. Let us consider positive initial data $(u_0,v_0)$ chosen small enough, within $\mathbb{V}_1$ s.t. $  (u_0,v_0) < W^{1}_{s}(E_{1*})  < W_{s}(E_{1**})$, we will have $(\overline{u},\overline{v}) \to (u^*,v^*)$ and $(\widetilde{u},\widetilde{v})\to (u^{**},v^{**})$. Note the spatially homogeneous solutions may be different. Hence, from squeezing argument, we can take $\epsilon \to 0$, to yield the uniform convergence of solutions, i.e.,
		\[ \lim_{\epsilon \to 0} \lim_{t \to \infty} (u,v) \to (u^*,v^*).\]
		
		On the other hand, let us consider sufficiently large positive initial data $(u_1,v_1)$, within $\mathbb{V}_2$ s.t. $  (u_1,v_1)>W^{1}_{s}(E_{1*}) \geq W_{s}(E_{1**})$, we will have $(\overline{u},\overline{v}) \to (0,1)$ and $(\widetilde{u},\widetilde{v})\to (0,1)$, i.e.,
		\begin{align*}
			\Big\{  (0,1)\Big\} \le \Big\{ (u,v) \Big\} \le \Big\{ (0,1) \Big\}.
		\end{align*}
		This proves the theorem.	
	\end{proof}
	
	\begin{figure}[H]\label{two_sepratrix}
		\centering
		\includegraphics[scale=0.4]{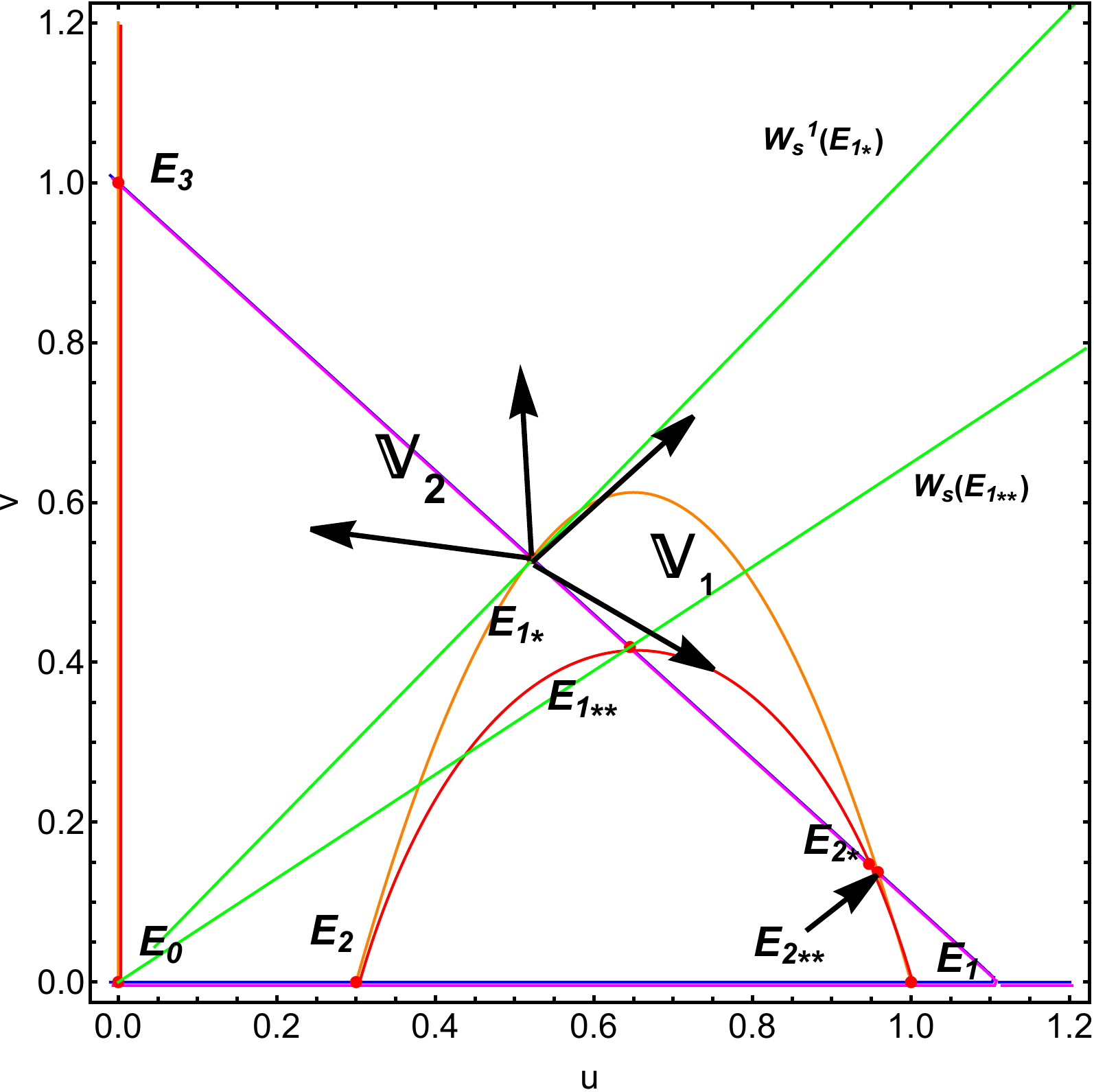}
		\caption{Numerical simulation for the Allee competition PDE case when species $v$ is fearing species $u$ in $\Omega=[0,1]$. The parameters are chosen as $a=0.2,b=0.9,c=0.9$ and $p=0.3$. Equilibria: $E_{1*}=(0.5218,0.5303),E_{1**}=(0.6453,0.4191), E_{2*}=(0.9581,0.1376),E_{2**}=(0.9468,0.1478), E_2=(0.3,0), E_1=(1,0),E_3=(0,1)$ and $E_0=(0,0)$. $W_{s}^{1}(E_{1*})$ ($q=0$) and $W_{s} (E_{1**)}$ ($q=1.1$) are two sepratrices passing through $E_{1*}$ and  $E_{1**}$ respectively. The $C^{1}$ property of the separatrices $W^{1}_{s}(E^{1*} , W_{s}(E_{1**})$, show a wedge $\mathbb{V}_1$ and $\mathbb{V}_2$ both emanating from $E_{1*}$, s.t in $\mathbb{V}_1$ we have $W_{s}(E^{1**} \geq W^{1}_{s}(E_{1*})$ and in $\mathbb{V}_2$ we have $W_{s}(E^{1**} \leq W^{1}_{s}(E_{1*})$. The u-nullcline is in orange for $q=0$  and red for $q=1.1$. The v-nullcline is in magenta for $q=0$ and $q=1.1$.}
	\end{figure}

	The numerical simulations (See Figure~\eqref{fig_conj_st_al} and \eqref{fig_conj_wk_al}), motivates us to formulate a conjecture:
	\begin{conjecture}\label{conj:two_post}
		For the reaction diffusion system (\ref{PDEmodel}) for the Allee effect (both strong and weak) with a fear function $q(x)$ such that the parametric restriction
		\begin{equation}\label{eq:st_pde_str_allee}
			\Delta >0, \quad \quad 2 A_1 + A_2 >0,\quad \text{and} \quad 0< c < 1
		\end{equation}
		holds true for $\mathbf{\widehat{q}}$ and $\mathbf{\widetilde{q}}$. Then there some initial data $[u_0(x),v_0(x)]$ such that the solution $(u,v)$ to (\ref{PDEmodel}) converges uniformly to  $(u^*,v^*)$ as $t \to \infty$, while for some choice of initial data $[u_1 (x),v_1 (x)]$, for which the solution $(u,v)$ to (\ref{PDEmodel}) converges uniformly to the spatially homogeneous state $(0,1)$ as $t \to \infty$.
	\end{conjecture}

	\begin{figure}[h]
		\subfigure[]
		{\scalebox{0.45}[0.45]{
				\includegraphics[width=\linewidth,height=4.5in]{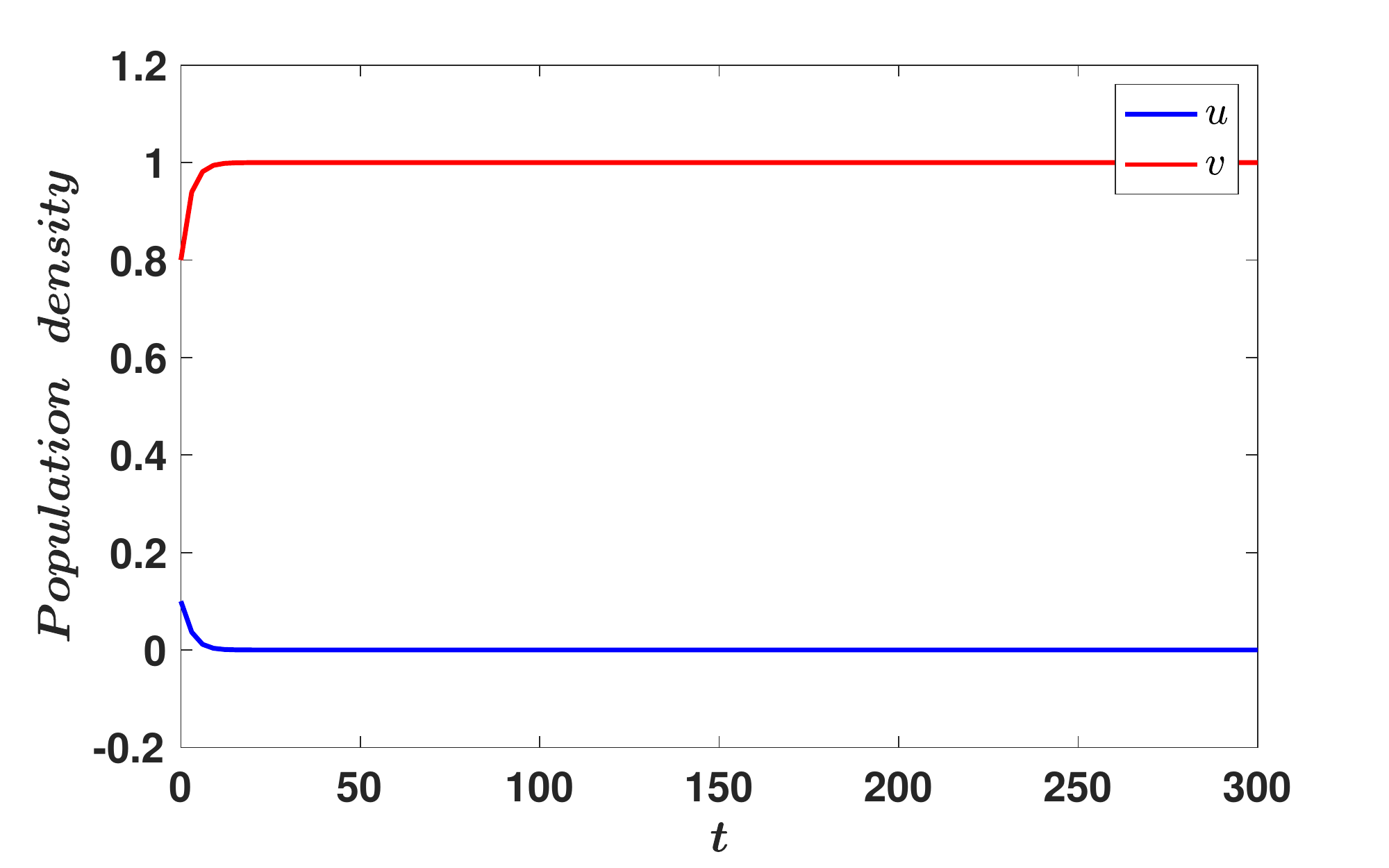}}}
		\subfigure[]
		{\scalebox{0.45}[0.45]{
				\includegraphics[width=\linewidth,height=4.5in]{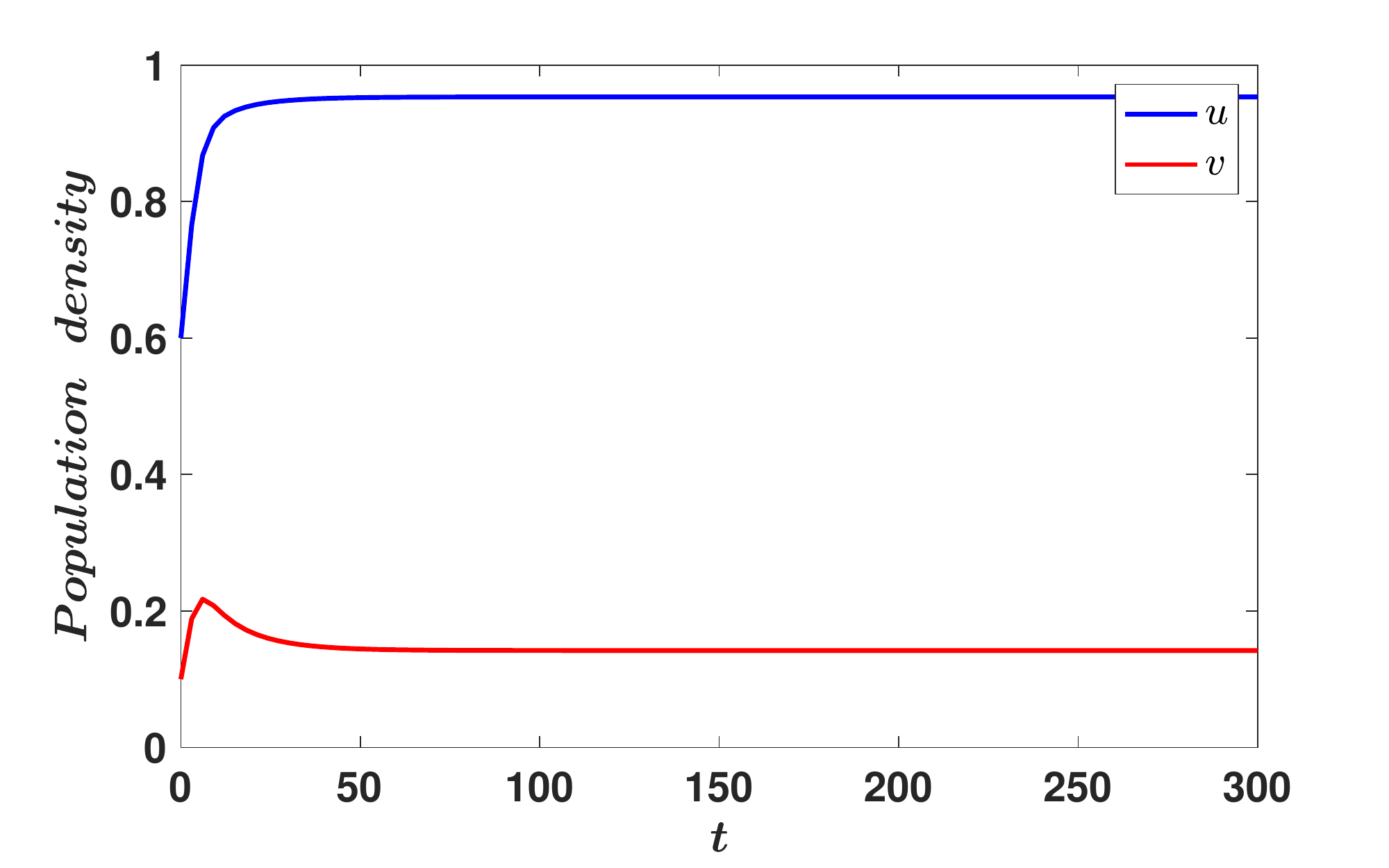}}}
		\caption{The parameters specified for the simulation of the reaction-diffusion system given in (\ref{PDEmodel}) on the spatial domain $\Omega=[0,1]$ for case of two positive interior equilibrium with fear function $q(x)=0.1+\sin^2 (4x)$ and strong Allee effect are $d_1=1,d_2=1,a=0.2,b=0.9,c=0.9$ and $p=0.3$. Two different sets of initial data is used: (a) $[u_0,v_0]=[0.1,0.8]$ and (b) $[u_0,v_0]=[0.6,0.2]$. It should be noted that these parameters satisfy the parametric constraints specified in the Theorem~\ref{conj:two_post}}.
		\label{fig_conj_st_al}
	\end{figure}

	\begin{figure}[h]
		\subfigure[]
		{\scalebox{0.45}[0.45]{
				\includegraphics[width=\linewidth,height=4.5in]{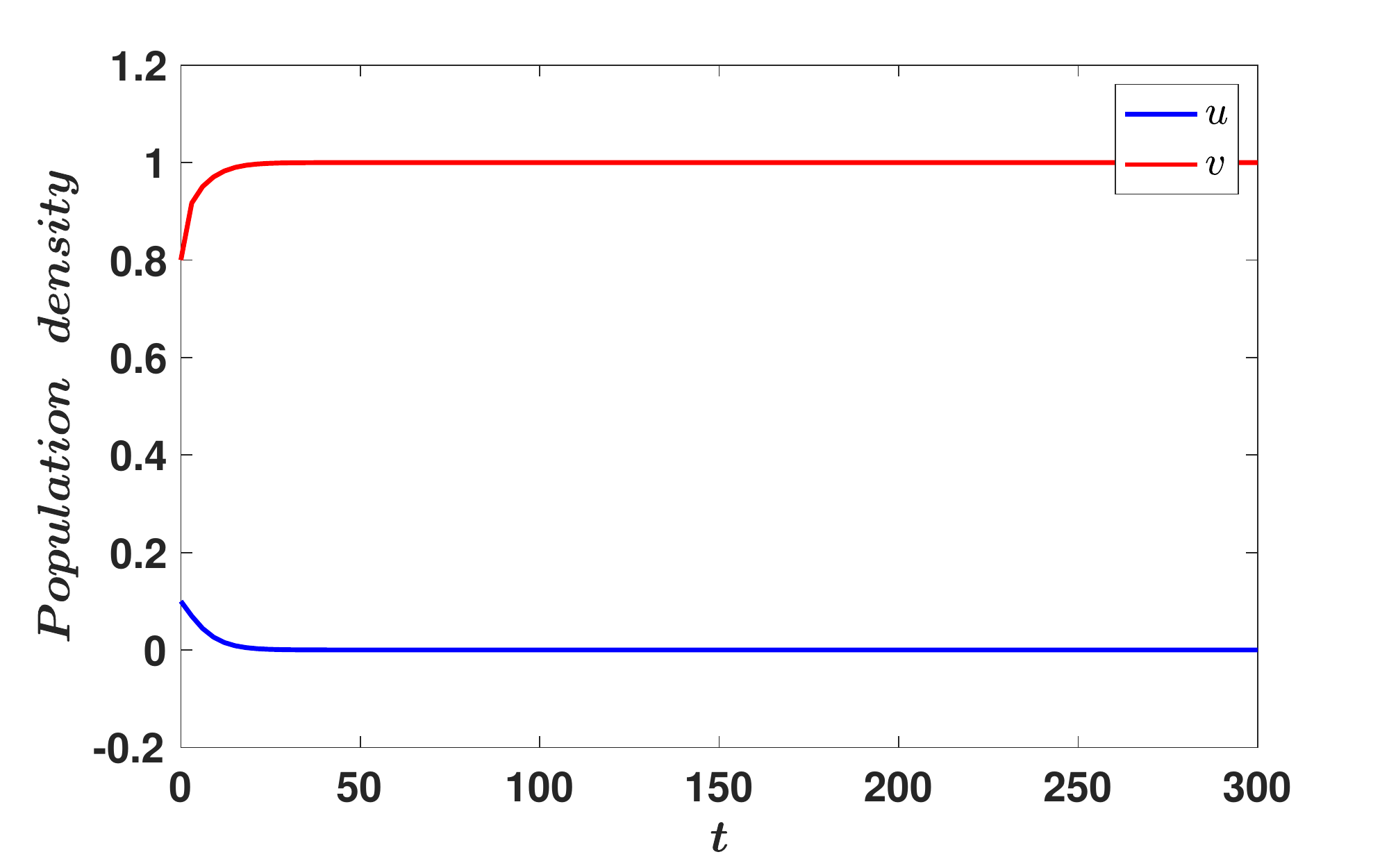}}}
		\subfigure[]
		{\scalebox{0.45}[0.45]{
				\includegraphics[width=\linewidth,height=4.5in]{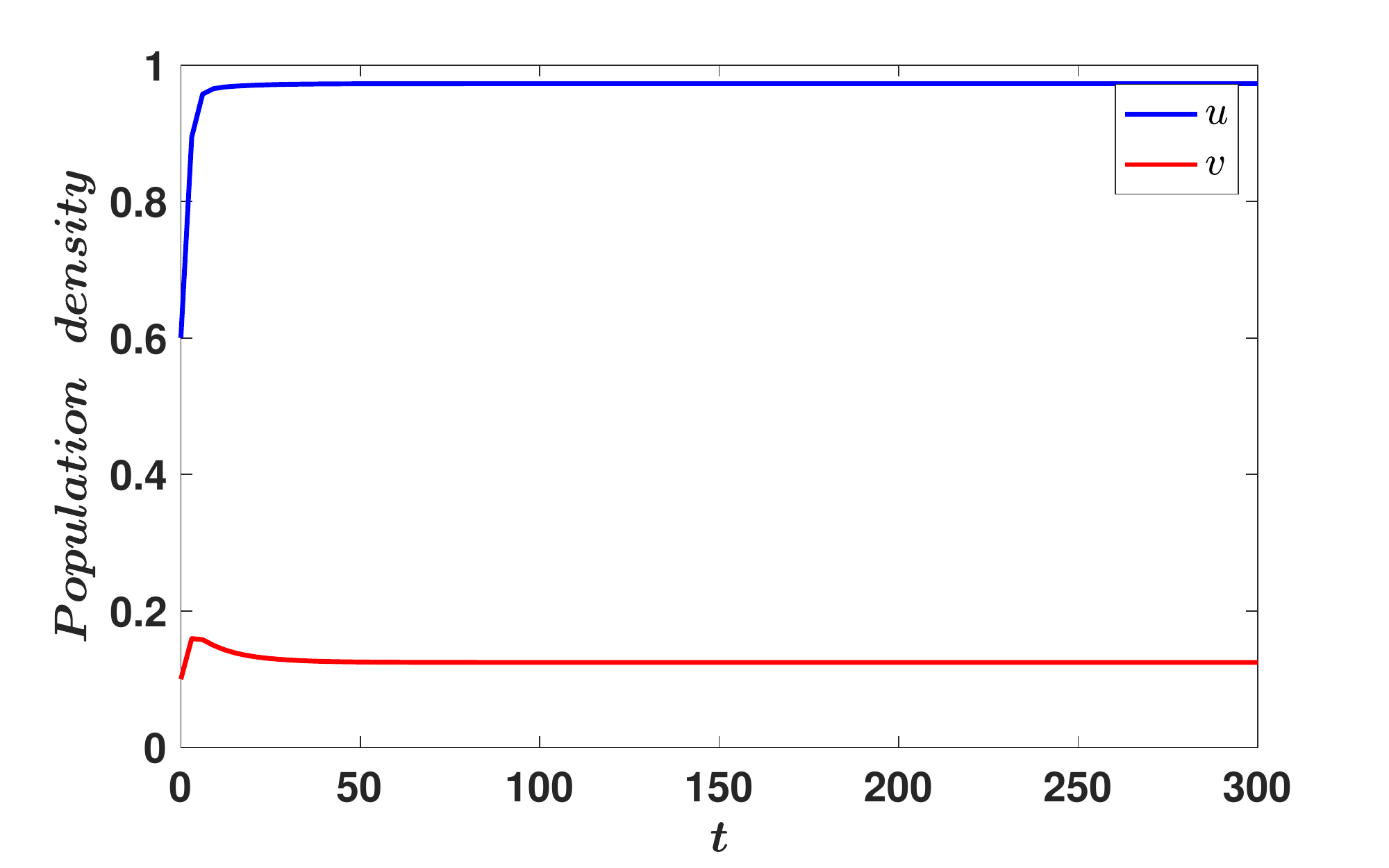}}}
		\caption{The parameters specified for the simulation of the reaction-diffusion system given in (\ref{PDEmodel}) on the spatial domain $\Omega=[0,1]$ for case of two positive interior equilibrium with fear function $q(x)=0.1+\sin^2 (4x)$ and weak Allee effect are $d_1=1,d_2=1,a=0.2,b=0.9,c=0.9$ and $p=0$. Two different sets of initial data is used: (a) $[u_0,v_0]=[0.1,0.8]$ and (b) $[u_0,v_0]=[0.6,0.2]$. It should be noted that these parameters satisfy the parametric constraints specified in the Theorem~\ref{conj:two_post}}.
		\label{fig_conj_wk_al}
	\end{figure}

	Let's make some observations about the spatially Homogeneous case. This is stated next,
	
	\begin{theorem}\label{lem:t11}
		For the reaction-diffusion system \eqref{PDEmodel}, for the Allee effect with a fear function $q(x) \equiv q >0$, where $q$ is a pure constant. Then the system does not possess diffusion-driven instability for any range of parameters or diffusion coefficients.
	\end{theorem}

	\begin{proof}
		The Jacobian matrix of the positive interior equilibrium, denoted as $E_{i*}$, for the reaction-diffusion system \eqref{PDEmodel} can be expressed as follows:
		\begin{equation*}
			J(u^*,v^*)  \nonumber
			= \begin{bmatrix}
				u^*\displaystyle\frac{\partial f}{\partial u_*} & u^*\displaystyle\frac{\partial f}{\partial v_*} \vspace{2ex}\\
				v_*\displaystyle\frac{\partial g}{\partial u_*} & v_*\displaystyle\frac{\partial g}{\partial v_*}
			\end{bmatrix}\triangleq\begin{bmatrix}
				B_1& B_2\\
				B_3&B_4
			\end{bmatrix},
			\label{10}
		\end{equation*}
		where
		\[ B_1=\frac{u_* \left(-2 u_* +p +1\right)}{q v_* +1}, \quad B_2=u_* \left(-\frac{\left(1-u_* \right) \left(u_* -p \right) q}{\left(q v_* +1\right)^{2}}-a \right)<0, \]
		\[B_3=-bcv_*<0, \quad \& \quad B_4=-bv_*<0.\]
		
		Clearly, the sign of $B_i$ is negative for $i\in\{2,3,4\}.$ Note that $(u^*,v^*)$ is dynamically stable, and
		
		\[ \begin{aligned}
			\mathrm{det}(J(E_{i*})) &=\left.\left [ u_*v_*\frac{\partial f}{\partial v_*} \frac{\partial g}{\partial v_*}\left ( \frac{\mathrm{d} v_*^{(g)}}{\mathrm{d} u_*}-\frac{\mathrm{d} v_*^{(f)}}{\mathrm{d} u_*} \right ) \right ]\right|_{(u_*,v_*)}\\
			&=\left.B_1B_4-B_2B_3\right|_{(u_*,v_*)}\\
			&>0.
		\end{aligned}
		\]
		Hence $B_1$ is negative, and it follows that $B_1 B_4> 0$. This violates the necessary criterion for Turing instability. Consequently, standard theory \cite{Winter13} implies that diffusion-driven instability cannot exist for any range of parameters or diffusion coefficients.
		
	\end{proof}

	\section{Numerical Simulations}

	The \verb|MATLAB| R2021b software was utilized to perform a PDE simulation for the reaction-diffusion system featuring Allee effect, with consideration for spatially heterogeneous fear function $q(x)$. The pdepe function was used for solving 1-D initial boundary value problems in one spatial dimension. The simulation was executed on an 8-core CPU, Apple M1 pro-chip-based workstation, and lasted between 5-7 seconds when the unit interval [0,1] was taken as a spatial domain, partitioned into 1000 sub-intervals.
	
	Our theoretical results and conjecture for the spatially explicit setting were validated numerically via a time series analysis over an extended period using simulations with parameters that adhered to theorems' parametric constraints. We utilized the standard comparison theory in the spatially explicit setting to determine point-wise restrictions on the fear function $q(x)$, which resulted in competitive exclusion, strong competition, and multiple equilibria type dynamics for the reaction-diffusion system with spatially heterogeneous fear function and Allee effect (weak or strong). Theorems \ref{thm:ce1}, \ref{thm:str_str}, \ref{thm:str_weak}, and \ref{thm:two_post}, and Conjecture \ref{conj:two_post} demonstrated these outcomes, while Figs [\ref{fig_ce_st_al},\ref{fig_ce_wk_al},\ref{fig_st_st_al},\ref{fig_st_wk_al},\ref{fig_conj_st_al},\ref{fig_conj_wk_al}] were used to validate the numerical results.
	
	Theoretical results were validated numerically using various heterogeneous fear functions, with captions specifying all parameters used for the numerical simulations and the relevant theorems. Based on the model and comparison to the logistics equation, it is evident that the population of any species cannot exceed one, so all parameters were chosen within the range [0,1].

	\begin{figure}[h]
		\subfigure[]
		{\scalebox{0.33}[0.45]{
				\includegraphics[width=\linewidth,height=5in]{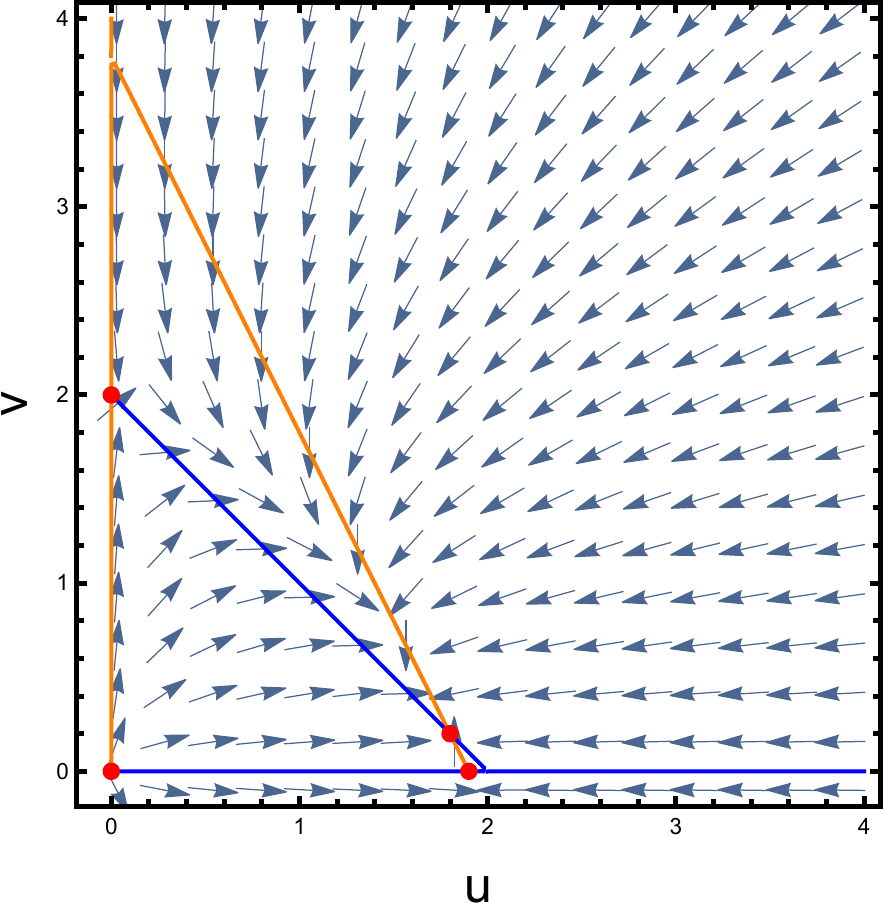}}}
		\subfigure[]
		{\scalebox{0.33}[0.45]{
				\includegraphics[width=\linewidth,height=5in]{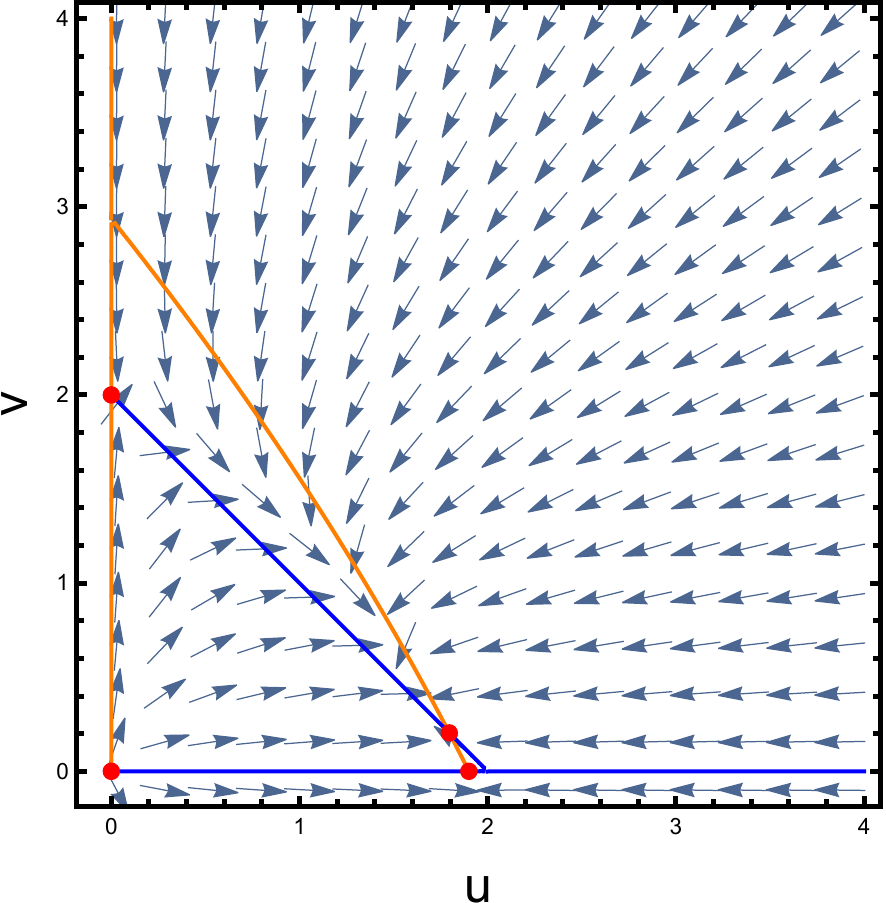}}}
		\subfigure[]
		{\scalebox{0.33}[0.45]{
				\includegraphics[width=\linewidth,height=5in]{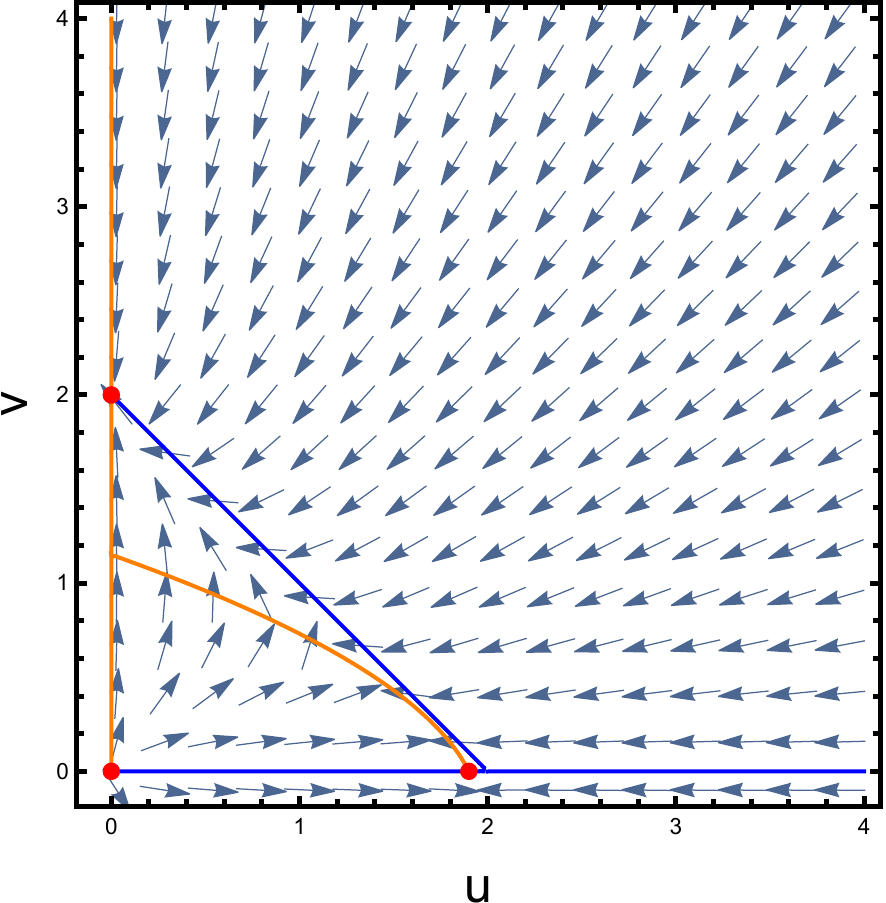}}}
		\newline
		\subfigure[]
		{\scalebox{0.33}[0.45]{
				\includegraphics[width=\linewidth,height=5in]{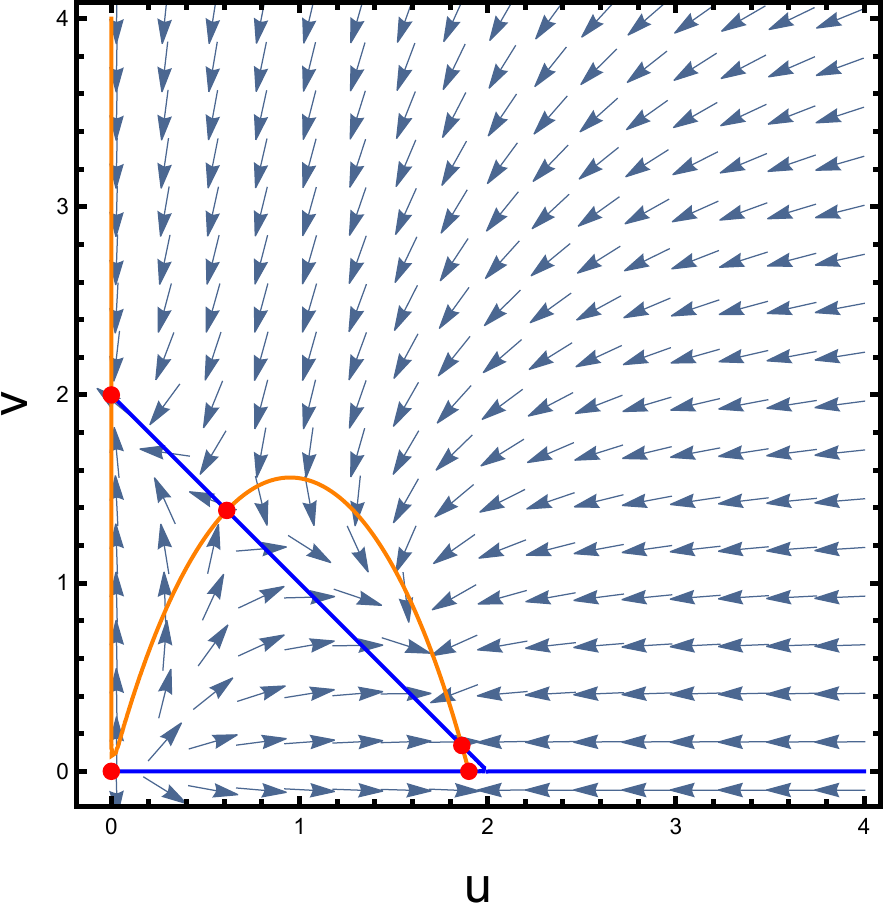}}}
		\subfigure[]
		{\scalebox{0.33}[0.45]{
				\includegraphics[width=\linewidth,height=5in]{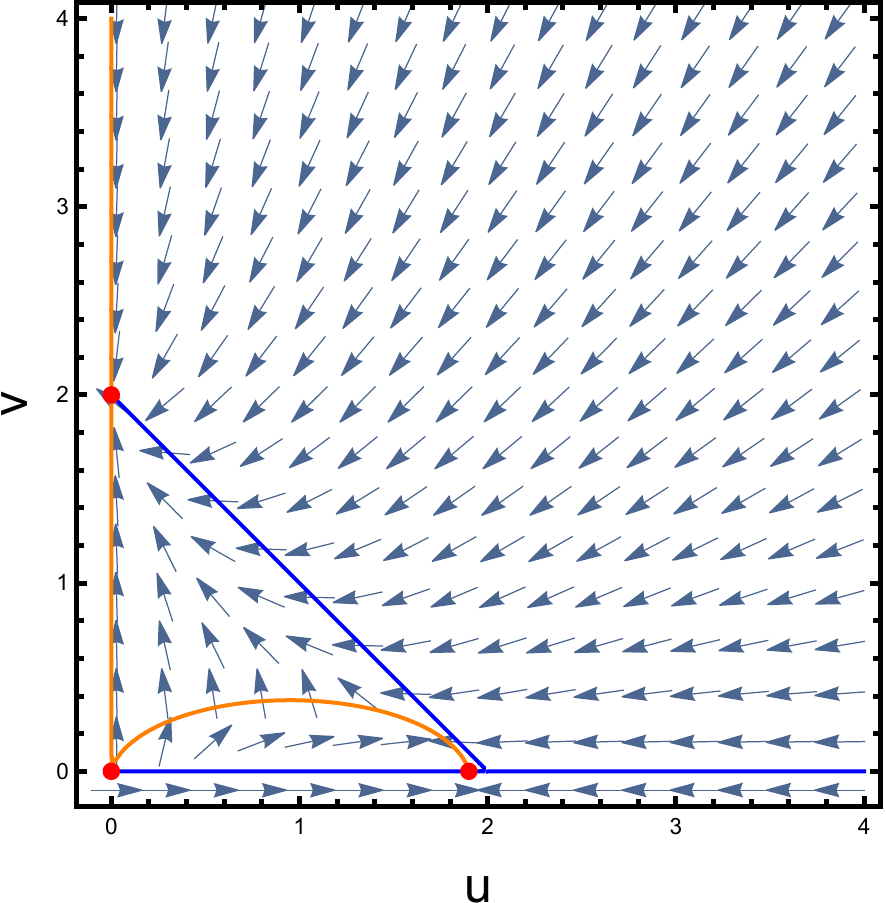}}}
		\subfigure[]
		{\scalebox{0.33}[0.45]{
				\includegraphics[width=\linewidth,height=5in]{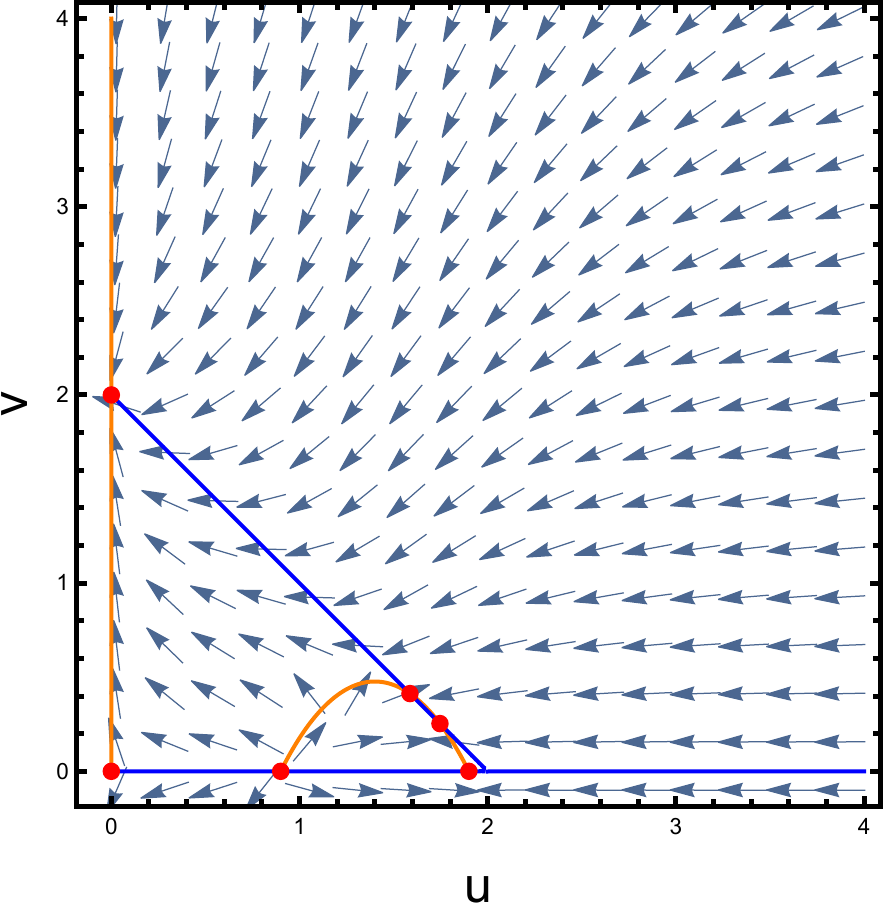}}}
		\newline
		\subfigure[]
		{\scalebox{0.33}[0.45]{
				\includegraphics[width=\linewidth,height=5in]{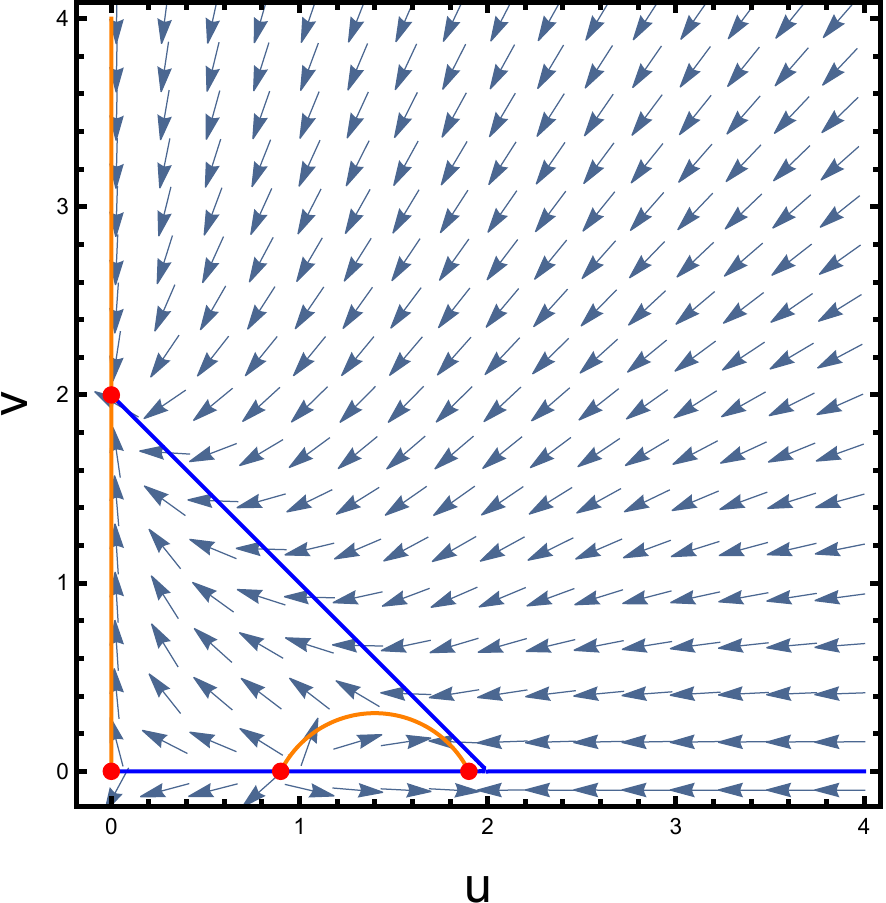}}}
		
		\caption{The phase plots below depict different dynamics for the two-species Lotka-Volterra competition model under the influence of Allee and fear effects \eqref{update_model}. The orange curve represents the $u$-nullcline, while the blue curve represents the $v$-nullcline. The parameters used in all cases are: $a=0.5, b=0.5, c=1, e=2.1$, and $f=2$. The cases shown are: (a) Classical case with no Allee and fear effects. (b) No Allee and small fear  $(q=0.1)$ (c) No Allee and large fear  $(q=2)$. (d) Weak Allee and small fear  $(q=0.1)$. (e) Weak Allee and large fear  $(q=10)$. (f) Strong Allee $(p=0.9)$ and small fear $(q=0.1)$. (g) Strong Allee $(p=0.9)$ and large fear  $(q=2)$.}
		\label{fig:dis_1}
	\end{figure}

	\begin{figure}[h]
		\subfigure[]
		{\scalebox{0.33}[0.45]{
				\includegraphics[width=\linewidth,height=5in]{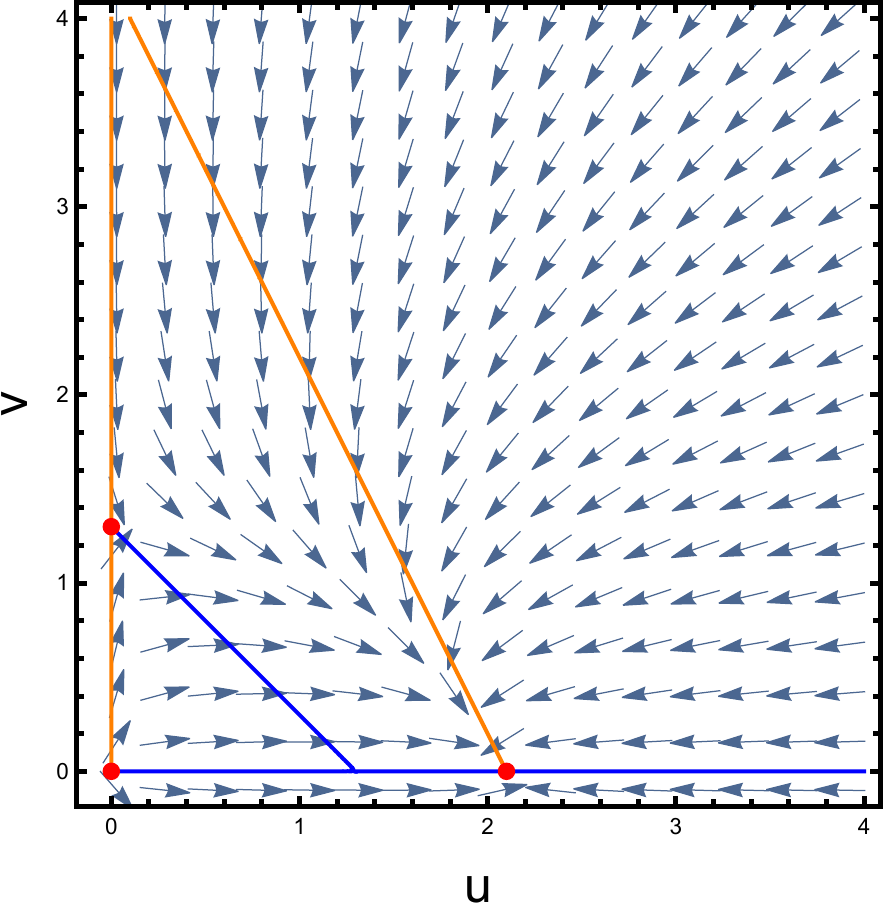}}}
		\subfigure[]
		{\scalebox{0.33}[0.45]{
				\includegraphics[width=\linewidth,height=5in]{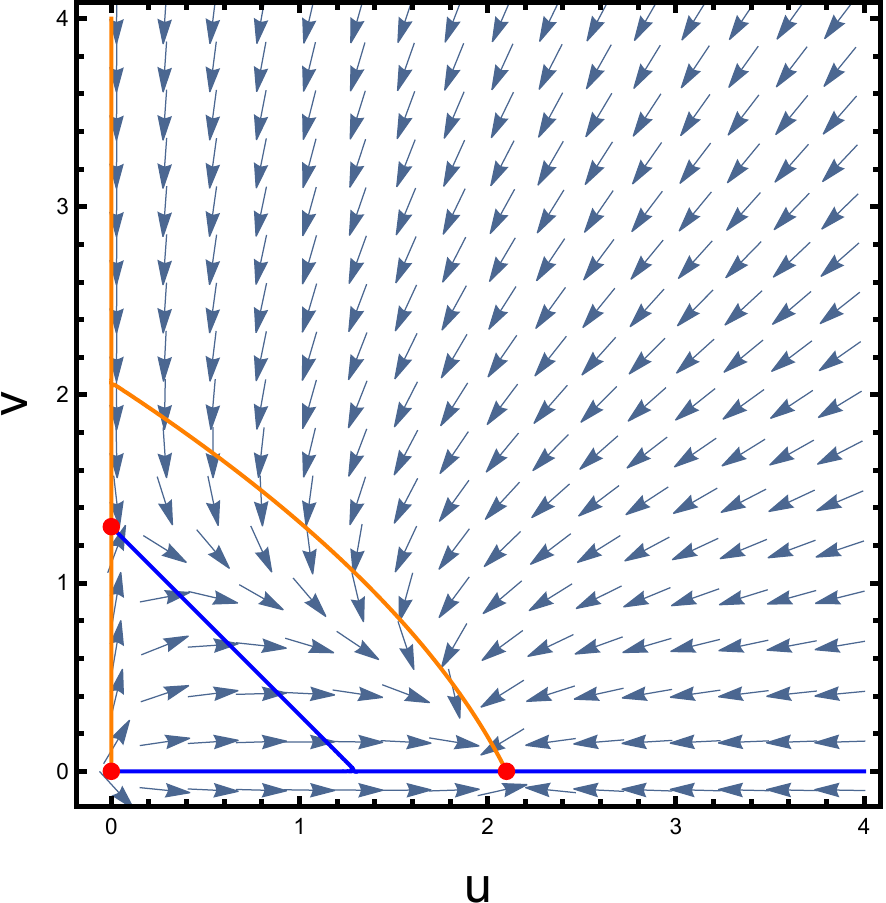}}}
		\subfigure[]
		{\scalebox{0.33}[0.45]{
				\includegraphics[width=\linewidth,height=5in]{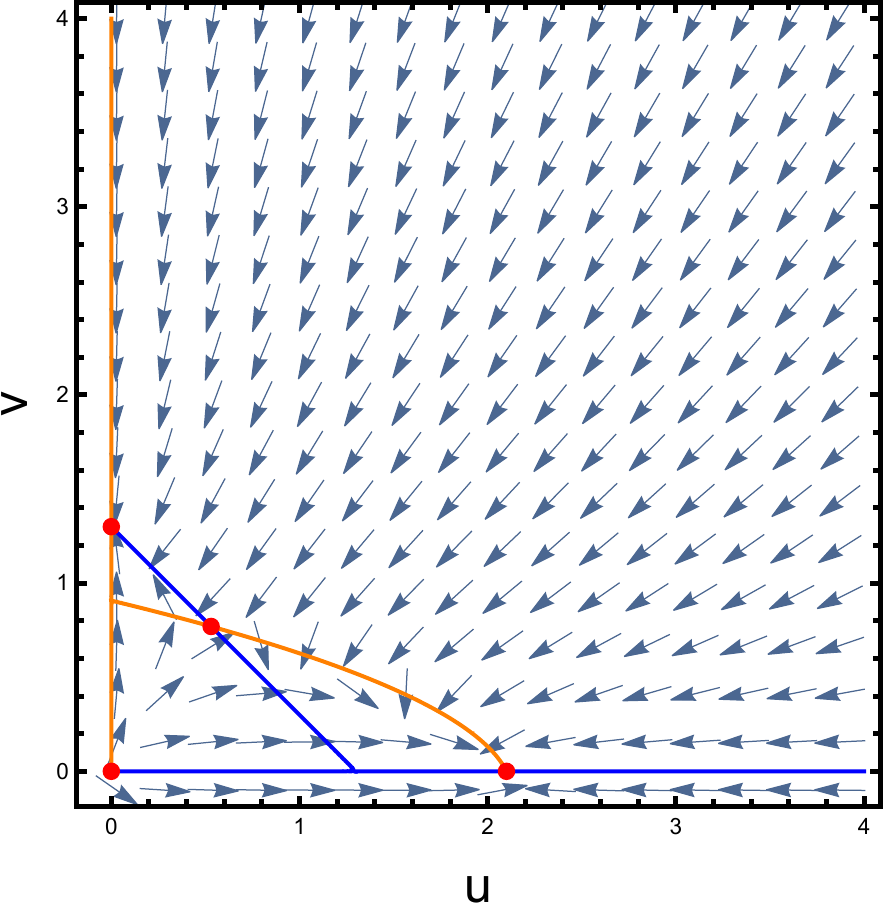}}}
		\newline
		\subfigure[]
		{\scalebox{0.33}[0.45]{
				\includegraphics[width=\linewidth,height=5in]{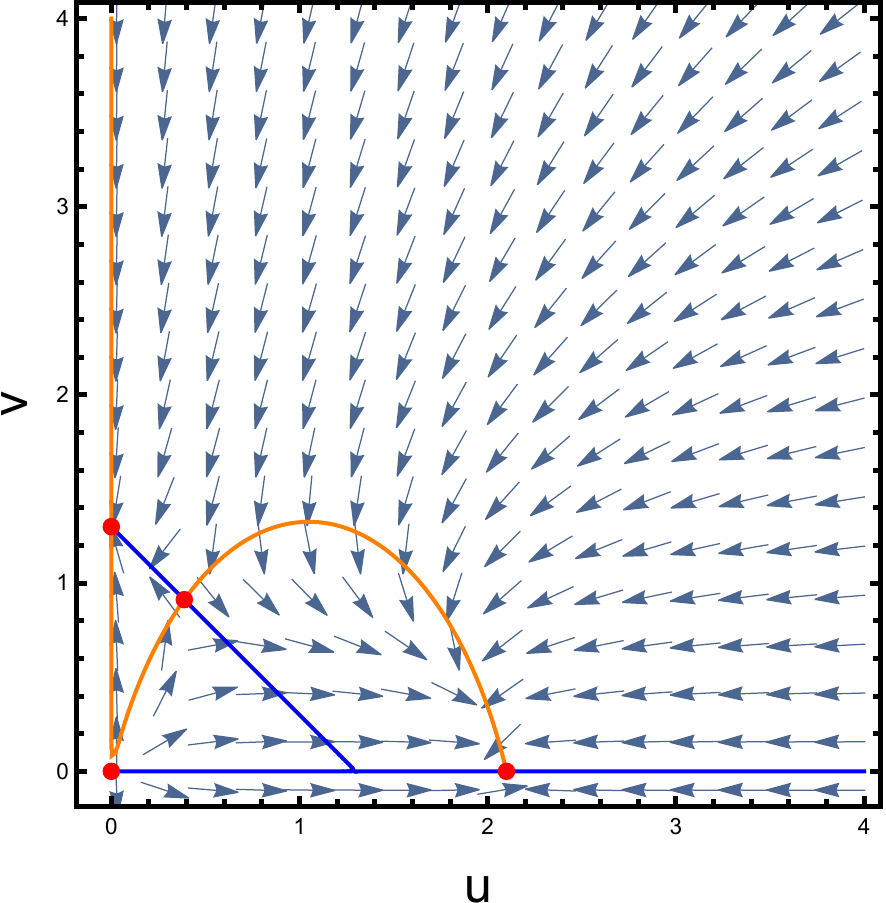}}}
		\subfigure[]
		{\scalebox{0.33}[0.45]{
				\includegraphics[width=\linewidth,height=5in]{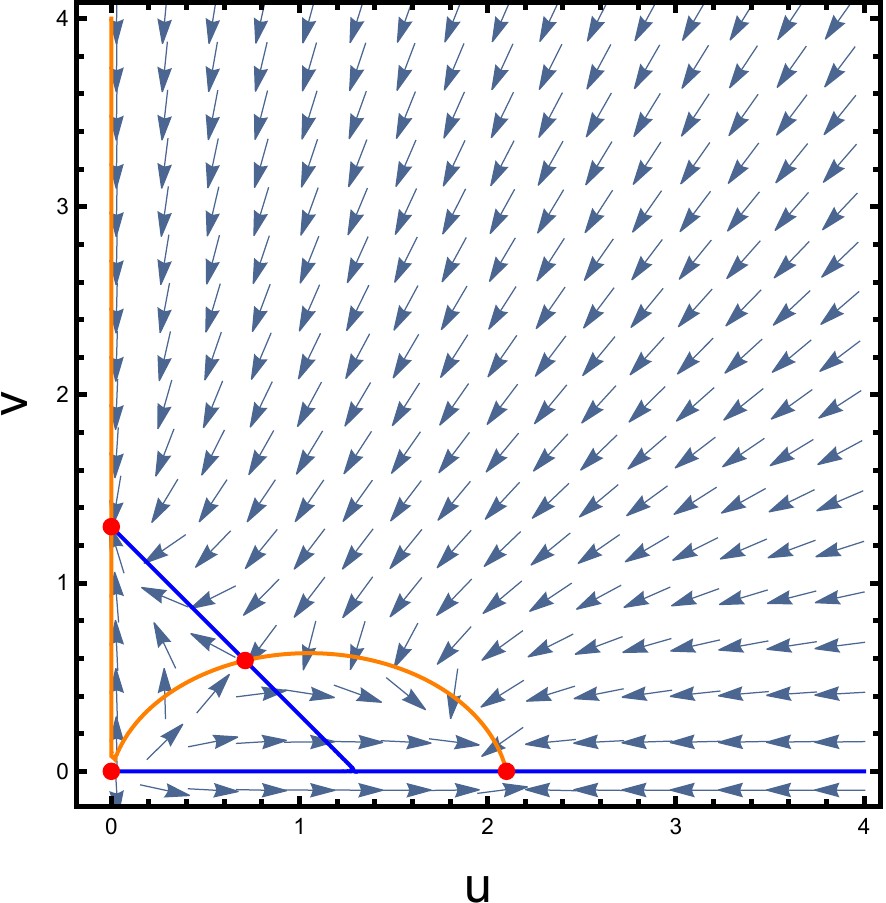}}}
		\subfigure[]
		{\scalebox{0.33}[0.45]{
				\includegraphics[width=\linewidth,height=5in]{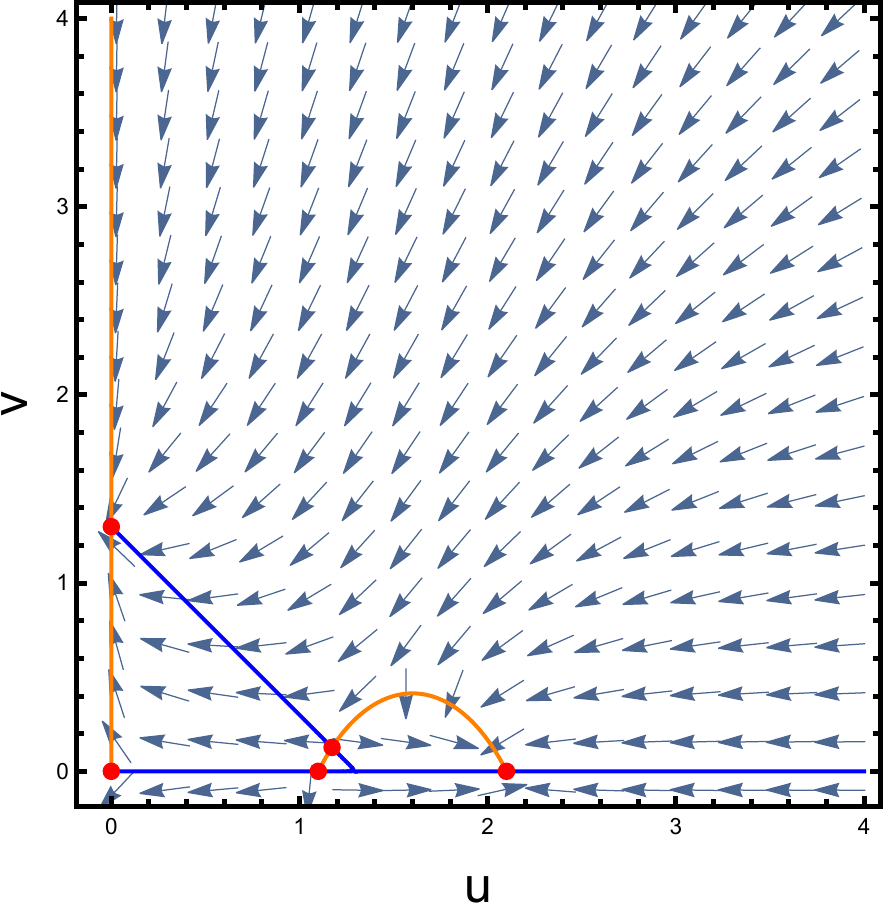}}}
		\newline
		\subfigure[]
		{\scalebox{0.33}[0.45]{
				\includegraphics[width=\linewidth,height=5in]{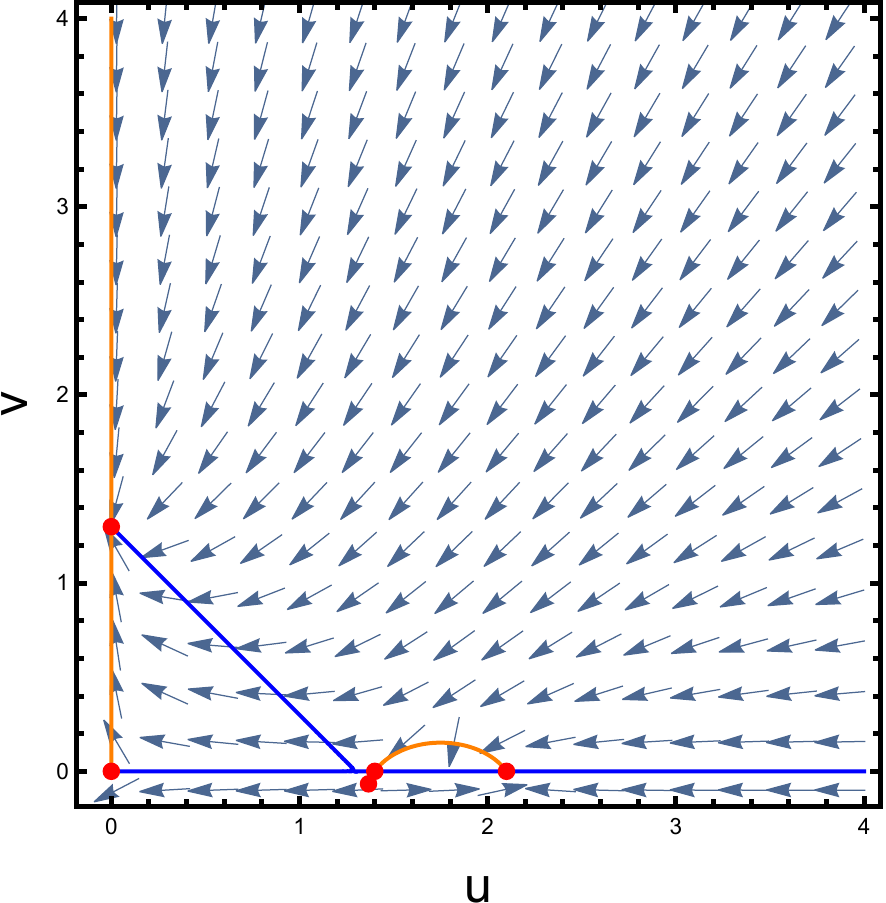}}}
		\caption{The phase plots below depict different dynamics for the two-species Lotka-Volterra competition model under the influence of Allee and fear effects \eqref{update_model}. The orange curve represents the $u$-nullcline, while the blue curve represents the $v$-nullcline. The parameters used in all cases are: $a=0.5, b=0.5, c=1, e=2.1$, and $f=1.3$. The cases shown are: (a) Classical case with no Allee and fear effects. (b) No Allee and small fear  $(q=0.5)$ (c) No Allee and large fear  $(q=4)$. (d) Weak Allee and small fear  $(q=0.5)$. (e) Weak Allee and large fear  $(q=4)$. (f) Strong Allee $(p=1.1)$ and small fear $(q=0.5)$. (g) Strong Allee $(p=1.4)$ and large fear  $(q=4)$}
		\label{fig:dis_2}
	\end{figure}

	\section{Discussion and Conclusion}
	In this paper, we study a competitive system in which the first species is affected by both the Allee and fear effects. First, we can determine that the system \eqref{3} always has four boundary equilibria. Let $f(x,y)=g(x,y)$ to obtain a quadratic equation \eqref{5}, where the real roots of \eqref{5} are the horizontal coordinates of the equilibria. To be consistent with the biological meaning, we only study the positive equilibria. According to Proposition 2.1 and 2.3, we know that the solution of the system \eqref{3} needs to satisfy $0<x(t)<1$ and $0<y(t)<1$. Thus we obtain the conditions to be satisfied by the system's parameters \eqref{3} in the presence of all positive equilibria.
	
	Next, we analyze the stability of the equilibria. Substituting the boundary equilibria into the Jacobi matrix of system \eqref{3} for verification, it is found that the boundary equilibria $E_1$ and $E_2$ may be degenerate equilibria under certain conditions. We use Theorem 7.1 in Chapter in \cite{22} to determine and obtain the stability of $E_1$ and $E_2$ under different conditions. Therefore, it can be concluded that a change in the Allee effect parameter $p$ leads to a change in the extinction of system \eqref{3} as well. For the positive equilibria $E_{1*}$ and $E_{2*}$, we transform the Jacobi matrix of system \eqref{3} into the form of \eqref{10}. By comparing the relationship between the magnitudes of the tangents of the two isoclines at the intersection, we can determine that $E_{1*}$ is always a saddle and $E_{2*}$ is a stable node. By changing the value of the fear effect parameter $q$, $E_{1*}$ and $E_{1*}$ recombine into $E_{3*}$ when $q = q_*$. By calculating $J(E_{3*})$ it is found that $E_{3*}$ is a degenerate equilibrium point. Similarly, we translate $E_{3*}$ to the origin and perform a Taylor expansion on it, using Theorem 7.1 in Chapter in \cite{22} to find that $E_{3*}$ is a saddle-node.
	
	In studying the existence of equilibria, we find that the positive equilibrium point $E_{2*}$ coincides with the boundary equilibrium point $E_1$ when $c=1$. We consider that system \eqref{3} undergoes a transcritical bifurcation at $E_1$. We choose $c_{TR}$ as the bifurcation parameter and use {\it Sotomayor's Theorem} to prove the transversality condition for system \eqref{3} to undergo a transcritical bifurcation. Similarly, we also find that the positive equilibria $E_{1*}$ and $E_{2*}$ will merge into $E_{3*}$ when $q=q_*$. By choosing $q_{SN}$ as the bifurcation parameter and using {\it Sotomayor's Theorem}, we prove that system \eqref{3} undergoes a saddle-node bifurcation around $E_{3*}$.
	
	In summary, when $0<q<q_*$, $2A_1+A_2>0$, and $0<c<1$, there is a stable positive equilibrium point $E_{2*}$, i.e., two species can maintain a coexistence relationship under this condition. This article has some guidance for the conservation of species diversity.
	
	It is worthwhile to comment to what degree the addition of the Allee effect changes the dynamical features of the model presented in \cite{19}, that is when there is only the fear effect present. We discuss this in terms of the classical scenarios of competitive exclusion, weak competition and strong competition, see Appendix B. In a weak competition type scenario, classically there is one globally attracting interior equilibrium. The effect of fear on the species $u$, does not dynamically change this situation for small values of the fear parameter, in that there remains one interior equilibrium which is globally attracting. However, intermediate values of fear could yield two interior equilibrium and large values of fear can enable a competitive exclusion type scenario where $(0,v^{*})$ becomes globally attracting. 
	However, with an added weak Allee effect, two interior equilibriums will always occur, see Fig. \ref{fig:dis_1}, for small values of fear - whereby Theorem \ref{thm:stab_pos}, one is stable and one a saddle. This creates a bi-stability situation whereby attraction to an interior equilibrium is possible, for certain initial conditions, while certain other initial conditions are attracted to $(0,v^*)$ - this is seen by comparing (b) to (d) in Fig. \ref{fig:dis_1}.
	Such monotonicity breaking behavior has been recently observed in variations of the classical competition model \cite{parshad2021some}.
	Note if fear values are chosen large one can again have a competitive exclusion type scenario where $(0,v^{*})$ becomes globally attracting. Similar dynamics are seen with a strong Allee effect in place. Note, for large enough Allee threshold or fear effect, a collision of the two interior equilibriums is possible via a saddle node bifurcation, leading to no interior equilibriums. In this setting we have a competitive exclusion type scenario with $(0,v^*)$, being globally attracting. Thus Combination of Allee threshold and fear parameter, lead to qualitatively different dynamics than with only a fear effect, when in particular the fear effect is \emph{small}.
	
	In the case of a competitive exclusion type scenario with $(u^*,0)$, being globally attracting, a large enough fear effect in $u$, can cause the occurrence of one interior equilibrium, which is a saddle, resulting in a bi-stability type situation. No effect is seen in the event of small fear. However, with an added weak Allee effect one interior equilibriums will always occur, see Fig. \ref{fig:dis_2}, for any values of fear, small or large, see Fig. \ref{fig:dis_2}. Note, with an added strong Allee effect, for any level of fear, there always exists an Allee threshold, s.t. competitive exclusion can be changed to a strong competition type scenario, that is $(0,v^*)$, becomes attracting for certain initial data while $(u^*,0)$, is attracting for certain other initial data, see Fig. \ref{fig:dis_2}. 
	
	These dynamics have interesting applications for bio-control. If we envision a situation where an invasive/pest population $u$ is attempted to be controlled by an introduced predator population $v$, where the invader is not allowing the introduced species to establish (so competitive exclusion of $v$ is seen) then knowledge that there exists a weak Allee effect in $u$, does not warrant investment in $v's$, that could incite fear - rather control should occur by direct consumptive effects. The reason being, any level of fear will not be sufficient to eradicate $u$, as seen in Fig. \ref{fig:dis_2} (d), (e). However, if a strong Allee effect is present in $u$, managers could think of how that threshold could be manipulated or increased - this could ``almost" reverse competitive exclusion and eradicate $u$ \cite{courchamp2008allee}, for most initial conditions - $(u^*,0)$ would still be locally attracting, se . Alternatively if the invader and the introduced species are coexisting, then irrespective of the type of Allee effect, investment into inciting large enough fear in $u$, could be fruitful and lead to its eradication, see Fig. \ref{fig:dis_2} (g). Thus a key theoretical direction for bio-control applications is to investigate ecological mechanisms in competitive systems, such that Theorem \ref{thm:stab_pos} \emph{changes} and one can have an interior equilibrium that changes in stability - perhaps as the equilibrium moves from the positive to negative quadrants. This could create geometries whereby a globally attracting situation would change to competitive exclusion of one species or a bi-stability (strong competition) scenario would change to a globally attracting situation. In summary the Allee effect, weak or strong opposes the fear effect, in the regimes of small fear, in that it qualitatively changes the dynamics from a purely fearful situation. However, the Allee effect reinforces the fear effect, in the regimes of large fear. Thus we purport quantification of the fear levels incited by introduced predators/competitors is a crucial step in designing effective bio-control strategies, if there are Allee effects present.

	\section{Appendix A}
	\begin{flalign}
		\begin{split}
			& e_{10}=\frac{2 \left(a \left(p +1\right) c +p^{2}-2 a -2 p +1\right) c \left(\left(a +2 p \right) c -p -1\right) a}{\left(a^{2} c^{2}-p^{2}+2 p -1\right) \left(a \,c^{2}+c p +c -2\right)}\\
			& e_{01}=\frac{2 \left(a \left(p +1\right) c +p^{2}-2 a -2 p +1\right) \left(\left(a +2 p \right) c -p -1\right) a}{\left(a^{2} c^{2}-p^{2}+2 p -1\right) \left(a \,c^{2}+c p +c -2\right)}\\
			& e_{20}=\frac{2 \left(a \left(p +1\right) c^{2}+\left(p^{2}-3 a -4 p +1\right) c +p +1\right) a}{a^{2} c^{2}-\left(p -1\right)^{2}}\\
			&e_{11}=-\frac{\left(a \left(p +1\right) c +p^{2}-2 a -2 p +1\right) }{\left(a c +p -1\right)^{2} \left(a c -p +1\right)^{2} \left(-1+c \right) \left(c p -1\right)} e_{110}\\
			&e_{110}=\left(a^{3} c^{4}+3 \left(\frac{4 p}{3}+a \right) \left(p +1\right) a \,c^{3}+e_{1100} +\left(-3 p^{3}+3 p^{2}+\left(4 a +3\right) p +4 a -3\right) c +2 \left(p -1\right)^{2}\right) a\\
			&e_{1100}=\left(4 p^{3}+\left(-a -8\right) p^{2}+\left(-14 a +4\right) p -6 a^{2}-a \right) c^{2}\\
			&e_{02}=-\frac{\left(\left(a +2 p \right) c -p -1\right) \left(a^{2} c^{2}+2 a \left(p +1\right) c +p^{2}-4 a -2 p +1\right)^{2} a}{2 \left(a c +p -1\right)^{2} \left(a c -p +1\right)^{2} \left(-1+c \right) \left(c p -1\right)}\\
			&e_{30}=-\frac{2 a \left(-2+a \,c^{2}+\left(p +1\right) c \right)}{a^{2} c^{2}-\left(p -1\right)^{2}}\\
			&e_{21}=-\frac{\left(a \left(p +1\right) c^{2}+\left(p^{2}-3 a -4 p +1\right) c +p +1\right) \left(a^{2} c^{2}+2 a \left(p +1\right) c +p^{2}-4 a -2 p +1\right) a \left(-2+a \,c^{2}+\left(p +1\right) c \right)}{\left(a c +p -1\right)^{2} \left(a c -p +1\right)^{2} \left(-1+c \right) \left(c p -1\right)}\\
			&e_{12}=\frac{\left(a^{2} c^{4} p -2 a \left(p +1\right) \left(p +a \right) c^{3}+\left(-3 p^{3}+6 p^{2}+\left(8 a -3\right) p +3 a^{2}\right) c^{2}-2 \left(p +1\right) \left(-p^{2}+a +2 p -1\right) c -\left(p -1\right)^{2}\right) e_{120}}{2 \left(a c +p -1\right)^{3} \left(a c -p +1\right)^{3} \left(-1+c \right)^{2} \left(c p -1\right)^{2}}\\
			&e_{120}=-\left(a^{2} c^{2}+2 a \left(p +1\right) c +p^{2}-4 a -2 p +1\right)^{2} a \left(-2+a \,c^{2}+\left(p +1\right) c \right)\\
			&e_{03}=\frac{\left(\left(a +2 p \right) c -p -1\right) \left(a^{2} c^{2}+2 a \left(p +1\right) c +p^{2}-4 a -2 p +1\right)^{3} a \left(-2+a \,c^{2}+\left(p +1\right) c \right)}{4 \left(a c +p -1\right)^{3} \left(a c -p +1\right)^{3} \left(-1+c \right)^{2} \left(c p -1\right)^{2}}\\ \nonumber
			&f_{10}=\frac{2 b \left(-1+c \right) \left(c p -1\right) c}{a \,c^{2}+c p +c -2}\\
			&f_{01}=\frac{2 \left(c p -1\right) \left(-1+c \right) b}{a \,c^{2}+c p +c -2}\\
		\end{split}&
	\end{flalign}
	
	\section{Appendix B}
	\label{sec:compd}
	Consider two-species Lotka-Volterra competition model under the influence of Allee and fear effects:
	\begin{equation}
		\left\{\begin{array}{l}
			\displaystyle\frac{\mathrm{d} x}{\mathrm{d} t} =x \left [(e-x)\displaystyle\frac{p(x)}{q(y)}-ay \right ]=xf(x,y)\equiv F(x,y), \vspace{2ex}\\
			\displaystyle\frac{\mathrm{d} y}{\mathrm{d} t} =by\left (f-y-cx \right )=yg(x,y)\equiv G(x,y).
		\end{array}\right.
		\label{update_model}
	\end{equation}
	
	By setting $p(x)\equiv 1$ and $q(y)\equiv 1$, we can derive the classic Lotka-Volterra ODE competition model, which involves two competing species, denoted as $x$ and $y$: 
	\begin{equation}
		\left\{\begin{array}{l}
			\displaystyle\frac{\mathrm{d} x}{\mathrm{d} t} = e x - x^2 - a xy, \vspace{2ex}\\
			\displaystyle\frac{\mathrm{d} y}{\mathrm{d} t}  = bf y - b y^2 - bc xy.
		\end{array}\right.
		\label{Eqn:1}
	\end{equation}
	The intrinsic (per capita) growth rates are represented by $e$ and $bf$, while the intraspecific competition rates are represented by $1$ and $b$, and the interspecific competition rates are represented by $a$ and $bc$.  All parameters considered are positive. The dynamics of this system are well studied \cite{Murray93}. We recap these briefly,
	
	We briefly recap the dynamics of the classical two species Lotka-Volterra ODE competition model,
	\begin{itemize}
		\item $E_0 = (0,0)$ is always unstable.
		\item $E_x = (e,0)$ is globally asymptotically stable if $\dfrac{e}{bf} > \max\left\lbrace\dfrac{1}{bc},\dfrac{a}{b}\right\rbrace$. Herein $x$ is said to competitively exclude $y$.
		\item $E_y = (0,f)$ is globally asymptotically stable if $\dfrac{e}{bf}<\min\left\lbrace\dfrac{1}{a},\dfrac{a}{bc}\right\rbrace$. Herein $y$ is said to competitively exclude $x$. 
		\item $E^* = \Big(\frac{e-fa}{1-ac},\frac{f-ce}{1-ac}\Big)$ exists when $b-abc \neq 0$. The positivity of the equilibrium holds if $bc<\frac{bf}{e}<\frac{b}{a}$ and is globally asymptotically stable if $b(1-ac)>0$. This is said to be the case of weak competition.
		
		\item If $b(1-ac)<0$, then $E^* = \Big(\frac{e-fa}{1-ac},\frac{f-ce}{1-ac}\Big)$ is unstable as a saddle. In this setting, one has initial condition dependent attraction to either $E_x(e,0)$ or $E_y(0,f)$. This is the case of strong competition.
	\end{itemize}

	Similarly, by setting $p(x)\equiv 1$ and $q(y)=1+py$, we can obtain the Lotka-Volterra competition model for two competing species, where the species $y$ is afraid of the species $x$ \cite{19}. On the other hand, if we set $q(y) \equiv 1$ and $p(x)=p-x$, we obtain a two species Lotka-Volterra ODE competition model that is subject to the Allee effect.

	\section*{References}
	{\footnotesize
		\bibliography{biblo}
	}
	
\end{document}